\documentclass[english,11pt]{article}
 
\makeatletter

\usepackage[T1]{fontenc}
\usepackage[latin9]{inputenc}
\usepackage{amsthm}
\usepackage{amsmath}
\usepackage{amssymb}
\usepackage{graphicx}
\usepackage{dsfont}
\usepackage{cite}
\usepackage{amsmath}
\usepackage{array}
\usepackage{multirow}
\usepackage{caption}
\usepackage{color}
\usepackage{varwidth}

\theoremstyle{plain}

\theoremstyle{plain}

\theoremstyle{plain}
\newtheorem{lem}{\protect\lemmaname}[section]
\theoremstyle{plain}
\newtheorem{thm}{\protect\theoremname}[section]
\theoremstyle{plain}

\theoremstyle{definition}
\newtheorem{defn}{\protect\definitionname}[section]
\theoremstyle{definition}

\theoremstyle{definition}
\newtheorem{rem}{\protect\remarkname}[section]

\makeatother
  
\usepackage{babel} 

\providecommand{\claimname}{Claim}
\providecommand{\lemmaname}{Lemma}
\providecommand{\propositionname}{Proposition}
\providecommand{\theoremname}{Theorem}
\providecommand{\corollaryname}{Corollary} 
\providecommand{\definitionname}{Definition}
\providecommand{\assumptionname}{Assumption}
\providecommand{\remarkname}{Remark}

\DeclareMathOperator*{\argmax}{arg\,max}
\DeclareMathOperator*{\argmin}{arg\,min}

\newcommand{\nn}{\nonumber}
\newcommand{\Eloc}{\mathcal{E}_{\mathrm{loc}}}
\newcommand{\Eest}{\mathcal{E}_{\mathrm{est}}}
\newcommand{\Eprunet}{\mathcal{E}_{\mathrm{prune},t}}
\newcommand{\Eloct}{\mathcal{E}_{\mathrm{loc},t}}
\newcommand{\Eestt}{\mathcal{E}_{\mathrm{est},t}}
\newcommand{\Eprune}{\mathcal{E}_{\mathrm{prune}}}
\newcommand{\Ebal}{\mathcal{E}_{\mathrm{balanced}}}
\newcommand{\Scoll}{S_{\mathrm{coll}}}
\newcommand{\Shead}{S_{\mathrm{head}}}
\newcommand{\Stail}{S_{\mathrm{tail}}}
\newcommand{\Median}{\mathrm{Median}}
\newcommand{\CN}{\mathrm{CN}}
\newcommand{\errtilde}{\widetilde{\mathrm{err}}}
\newcommand{\Bmax}{B_{\mathrm{max}}}
\newcommand{\Fmax}{F_{\mathrm{max}}}
\newcommand{\SNR}{\mathrm{SNR}}

\newcommand{\pe}{P_{\mathrm{e}}}

\newcommand{\Xv}{\mathbf{X}}

\newcommand{\Yv}{\mathbf{Y}}

\newcommand{\Ac}{\mathcal{A}}

\newcommand{\Ec}{\mathcal{E}}
\newcommand{\Fc}{\mathcal{F}}

\newcommand{\Ic}{\mathcal{I}}

\newcommand{\Uc}{\mathcal{U}}

\newcommand{\Xc}{\mathcal{X}}

\newcommand{\CC}{\mathbb{C}}
\newcommand{\EE}{\mathbb{E}}
\newcommand{\PP}{\mathbb{P}}
\newcommand{\RR}{\mathbb{R}}
\newcommand{\ZZ}{\mathbb{Z}}

\newcommand{\supp}{\mathrm{supp}}

\newcommand{\sv}{\mathbf{s}}

\usepackage{float}
\usepackage{color}
\usepackage{array}
\usepackage{hyperref}
\usepackage{setspace}
\usepackage[lmargin=0.95in, rmargin=0.95in,tmargin=0.95in,bmargin=0.95in,footskip=0.25in]{geometry}

\usepackage{algorithm}
\usepackage[noend]{algpseudocode}
\usepackage{algpseudocode}

\newcommand{\wh}{\widehat}
\newcommand{\Err}{\mathrm{Err}}
\newcommand{\err}{\mathrm{err}}

\newcommand{\poly}{\mathrm{poly}}

\newtheorem{theorem}{Theorem}[section]

\newtheorem{definition}[theorem]{Definition}
\newtheorem{remark}[theorem]{Remark}

\newenvironment{proofof}[1]{\noindent{\bf Proof of #1:}}{$\qed$\par}

{\makeatletter
 \gdef\xxxmark{%
   \expandafter\ifx\csname @mpargs\endcsname\relax % in minipage?
     \expandafter\ifx\csname @captype\endcsname\relax % in figure/caption?
       \marginpar{xxx}% not in a caption or minipage, can use marginpar
     \else
       xxx % notice trailing space
     \fi
   \else
     xxx % notice trailing space-
   \fi}
 \gdef\xxx{\@ifnextchar[\xxx@lab\xxx@nolab}
 \long\gdef\xxx@lab[#1]#2{{\bf [\xxxmark #2 ---{\sc #1}]}}
 \long\gdef\xxx@nolab#1{{\bf [\xxxmark #1]}}
}

\usepackage{enumitem}
\newcommand{\e}{\epsilon}

\newcommand{\C}{\mathbb{C}}

\begin{document}

\title{An Adaptive Sublinear-Time Block Sparse Fourier Transform}

\author{Volkan Cevher\\EPFL \and Michael Kapralov\\EPFL \and Jonathan Scarlett\\EPFL \and Amir Zandieh\\EPFL}

\maketitle

 \begin{abstract}
The problem of approximately computing the $k$ dominant Fourier coefficients of a vector $X$ quickly, and using few samples in time domain, is known as the Sparse Fourier Transform (sparse FFT) problem. A long line of work on the sparse FFT  has resulted in algorithms with $O(k\log n\log (n/k))$ runtime [Hassanieh \emph{et al.}, STOC'12] and $O(k\log n)$ sample complexity [Indyk \emph{et al.}, FOCS'14]. These results are proved using non-adaptive algorithms, and the latter $O(k\log n)$ sample complexity result is essentially the best possible under the sparsity assumption alone: It is known that even adaptive algorithms must use $\Omega((k\log(n/k))/\log\log n)$ samples [Hassanieh \emph{et al.}, STOC'12]. By {\em adaptive}, we mean being able to exploit previous samples in guiding the selection of further samples.

This paper revisits the sparse FFT problem with the added twist that the sparse coefficients approximately obey a $(k_0,k_1)$-block sparse model. In this model, signal frequencies are clustered in $k_0$ intervals with width $k_1$ in Fourier space, and $k= k_0k_1$ is the total sparsity. Signals arising in applications are often well approximated by this model with $k_0\ll k$.

Our main result is the first sparse FFT algorithm for $(k_0, k_1)$-block sparse signals with a sample complexity of $O^*(k_0k_1 + k_0\log(1+ k_0)\log n)$ at constant signal-to-noise ratios, and sublinear runtime. A similar sample complexity was previously achieved  in the works on {\em model-based compressive sensing} using random Gaussian measurements, but used $\Omega(n)$ runtime. To the best of our knowledge, our result is the first sublinear-time algorithm for model based compressed sensing, and the first sparse FFT result that goes below the $O(k\log n)$ sample complexity bound. 

 Interestingly, the  aforementioned model-based compressive sensing result that relies on Gaussian measurements is non-adaptive, whereas our algorithm crucially uses {\em adaptivity} to achieve the improved sample complexity bound.  We prove that adaptivity is in fact necessary in the Fourier  setting: Any {\em non-adaptive} algorithm must use $\Omega(k_0k_1\log \frac{n}{k_0k_1})$ samples for the $(k_0,k_1$)-block sparse model, ruling out improvements over the vanilla sparsity assumption. Our main technical innovation for adaptivity is a new randomized energy-based importance sampling technique that may be of independent interest.
\end{abstract}

\setcounter{page}{0}
\thispagestyle{empty}

%\setcounter{page}{0}
%\thispagestyle{empty}

%!TEX root = BlockSparseFT-new.tex

\newpage
\section{Introduction} \label{sec:intro}

The discrete Fourier transform (DFT) is one of the most important tools in modern signal processing, finding applications in audio and video compression, radar, geophysics, medical imaging, communications, and many more. The best known algorithm for computing the DFT of a general signal of length $n$ is the Fast Fourier Transform (FFT), taking $O(n\log n)$ time, which matches the trivial $\Omega(n)$ lower bound up to a logarithmic factor.

In recent years, significant attention has been paid to exploiting \emph{sparsity} in the signal's Fourier spectrum, which is naturally the case for numerous of the above applications. By sparse, we mean that the signal can be well-approximated by a small number of Fourier coefficients. Given this assumption, the computational lower bound of  $\Omega(n)$ no longer applies. Indeed, the DFT can be computed in \emph{sublinear time}, while using a sublinear number of samples in the time domain \cite{gilbert2014recent,gilbert2008tutorial}. 

 The problem of computing the DFT of signals that are approximately sparse in the Fourier domain has received significant attention in several communities. The seminal work of~\cite{CTao,RV} in {\em compressive sensing}  first showed that only $k \log^{O(1)} n$ samples in time domain suffice to recover a length $n$ signal with at most $k$ nonzero Fourier coefficients.  A different line of research on the {\em Sparse Fourier Transform} (sparse FFT), with origins in computational complexity and learning theory, has  resulted in algorithms that use $k\log^{O(1)} n$ samples  and $k\log^{O(1)} n$ runtime (i.e., the runtime is {\em sublinear} in the length of the input signal). Many such algorithms have been proposed in the literature \cite{GL,KM,Man,GGIMS,AGS,GMS,Iw,Ak,HIKP,HIKP2,LWC,BCGLS,HAKI,pawar2013computing,heidersparse, IKP, IK14a, K16, PZ15}; we refer the reader to the recent surveys \cite{gilbert2014recent,gilbert2008tutorial} for a more complete overview.

The best known runtime for computing the $k$-sparse FFT is due to Hassanieh \emph{et al.}~\cite{HIKP2}, and is given by $O(k\log n\log(n/k))$, asymptotically improving upon the FFT for all $k=o(n)$. The recent works of~\cite{IKP, K16} also show how to achieve a sample complexity of $O(k\log n)$ (which is essentially optimal) in linear time, or  in time $k\log^{O(1)} n$ at the expense of $\poly(\log\log n)$ factors. Intriguingly, the aforementioned algorithms are all \textit{non-adaptive}. That is, these algorithms do not exploit existing samples in guiding the selection of the new samples to improve approximation quality. In the same setting, it is also known that adaptivity cannot improve the sample complexity by more than an $O(\log\log n)$ factor~\cite{HIKP2}.

%(up to $\mathrm{poly}(\log\log n)$ factors) 
%since it
%Indeed, even if an algorithm can  (i.e., when it is \textit{adaptive}), it must obey the lower bound of $\Omega\big( k\frac{\log n}{\log \log n}\big)$ \cite{Has12}. 

Despite the significant gains permitted by sparsity, designing an algorithm for handling \emph{arbitrary} sparsity patterns may be overly generic; in practice, signals often exhibit more specific sparsity structures.  A common example is \emph{block sparsity}, where significant coefficients tend to cluster on known partitions, as opposed to being unrestricted in the signal spectrum. Other common examples include tree-based sparsity, group sparsity, and dispersive sparsity \cite{Bar10,baldassarre2013group,halabi2015totally,bach2010structured}.  

Such structured sparsity models can be captured via the \emph{model-based} framework \cite{Bar10}, where the number of sparsity patterns may be far lower than ${n \choose k}$. For the compressive sensing problem, this restriction has been shown to translate into a reduction in the sample complexity, even with non-adaptive algorithms. Specifically, one can achieve a sample complexity of $O(k + \log |\mathcal{M}|)$ with dense measurement matrices based on the Gaussian distribution, where $\mathcal{M}$ is the set of permitted sparsity patterns.  Reductions in the sample complexity with other types of measurement matrices, e.g., sparse matrices based on expanders, are typically less  \cite{bah2014model,indyk2013model}. Other benefits of exploiting model-based sparsity include  faster recovery and improved noise robustness \cite{Bar10,bah2014model}.
%In addition, the time complexity of the model-based framework depends on the difficulty of computing the projection onto the model. 

Surprisingly, in stark contrast to the extensive work on exploiting model-based sparsity with general linear measurements, there are no existing sparse FFT algorithms exploiting such structure.  This paper presents the first such algorithm, focusing on the special case of block sparsity. Even for this relatively simple sparsity model, achieving the desiderata turns out to be quite challenging, needing a whole host of new techniques, and intriguingly, requiring \emph{adaptivity} in the sampling. %, as we describe in detail below.

To clarify our contributions, we describe our model and the problem statement in more detail.

\paragraph{\underline{Model and Basic Definitions}}

    The Fourier transform of a signal $X \in \CC^n$ is denoted by $\wh{X}$, and defined as
        $$\wh{X}_f=\frac{1}{n} \sum_{i \in [n]} X_i \omega_n^{-ft}, \quad f \in [n],$$
    where $\omega_n$ is the $n$-th root of unity. %In the sequel, we omit the subscript $n$ in $\omega$ for notational ease. 
    With this definition, Parseval's theorem takes the form 
$\|X\|^2 = n\|\wh{X}\|_2^2$.

We are interested in computing the Fourier transform of signals that, in frequency domain, are well-approximated by a \emph{block sparse} signal with $k_0$ blocks of width $k_1$, formalized as follows.  % More precisely, we divide the range of frequencies $f \in [n]$ into $\frac{n}{k_1}$ intervals of width $k_1$ as follows.

\begin{defn}[Block sparsity] \label{def:block}
    Given a sequence $X \in \CC^n$ and an even block width $k_1$, the \emph{$j$-th interval} is defined as $I_j = \big((j-1/2)k_1,(j+1/2)k_1\big] \cap \ZZ$ for $j \in \big[\frac{n}{k_1}\big]$, and we refer to $\wh{X}_{I_j}$ as the \emph{$j$-th block}.
    We say that a signal is \emph{$(k_0,k_1)$-block sparse} if it contains non-zero values within at most $k_0$ of these intervals.
\end{defn}

Block sparsity is of direct interest in several applications \cite{Bar10,baraniuk2010low}; we highlight two examples here: (i) In spectrum sensing, cognitive radios seek to improve the utilization efficiency in a sparsely used wideband spectrum. In this setting, the frequency bands being detected are non-overlapping and predefined. (ii) Audio signals often contain blocks corresponding to different sounds at different frequencies.  Such blocks may be {\em non-uniform}, and can be modeled by the {\em $(k,c)$ model} in which $k$ coefficients are arbitrarily spread across $c$ different clusters.  It was argued in \cite{Cev09a} that any signal from the $(k,c)$ model is also $(3c,k/c)$-block sparse in the uniform model. 

% (ii)  In multi-sensor processing, a frequency-sparse signal is recorded by an array of sensors.  Each sensor has the same dominant frequencies, but with different delays and amplitudes; hence, the frequency-domain signals of the sensors can be rearranged to produce a block-sparse signal.  While such rearranging does not directly fall into the framework of the present paper, our techniques can be applied, and in fact simplified, for this latter setting.

% We say that a signal is \emph{$(k_0,k_1)$-block sparse} if the energy of the part of the signal that lies on at most $k_0$ of these intervals is more than a threshold. To be more precise, $k_0$ block sparsity is defined as follows,

Our goal is to output a list of frequencies and values estimating $\wh{X}$, yielding an $\ell_2$-distance to $\wh{X}$ not much larger than that of the best $(k_0,k_1)$-block sparse approximation. Formally, we say that an output signal $\wh{X}'$ satisfies the \emph{$\ell_2/\ell_2$ block-sparse recovery guarantee} if
\begin{equation} %\label{eq:l2l2}
\|\wh{X}-\wh{X}'\|_2\leq (1+\epsilon) \min_{\wh{Y}\text{ is }(k_0,k_1)\text{-block sparse}} \|\wh{X} - \wh{Y}\|_2 \nonumber
\end{equation}
for an input parameter $\epsilon > 0$.  

The sample complexity and runtime of our algorithm are parameterized by the {\em signal-to-noise ratio} (SNR) of the input signal, defined as follows.
\begin{defn}[Tail noise and signal-to-noise ratio (SNR)] \label{def:SNR}
    We define the \emph{tail noise level} as
    \begin{align}
        \Err(\wh{X},k_0,k_1) & := \min_{\substack{S \subset [\frac{n}{k_1}] \\ |S| = k_0 }} \sum_{j \in [\frac{n}{k_1}] \backslash S} \|\wh{X}_{I_j}\|_2^2, \label{eq:def_S}
    \end{align}
    and its normalized version as  $\mu^2 := \frac{1}{k_0} \Err^2(\wh{X},k_0,k_1)$, representing the average noise level per block.  The \emph{signal-to-noise ratio} is defined as $\SNR := \frac{\|\wh{X}\|^2 }{\Err^2(\wh{X},k_0,k_1)}$. 
\end{defn}

Throughout the paper, we assume that both $n$ and $k_1$ are powers of two.  For $n$, this is a standard assumption in the sparse FFT literature.  As for $k_1$, the assumption comes without too much loss of generality, since one can always round the block size up to the nearest power of two and then cover the original $k_0$ blocks with at most $2k_0$ larger blocks, thus yielding a near-identical recovery problem other than a possible increase in the SNR.  We also assume that $\frac{n}{k_1}$ exceeds a large absolute constant; if this fails, our stated scaling laws can be obtained using the standard FFT.

We use $O^*(\cdot)$ notation to hide $\log\log\SNR$, $\log\log n$, and $\log\frac{1}{\epsilon}$ factors.  Moreover, to simplify the notation in certain lemmas having free parameters that will be set in terms of $\epsilon$, we assume throughout the paper that $\epsilon = \Omega\big( \frac{1}{\poly \log n} \big)$, and hence $\log\frac{1}{\epsilon} = O(\log \log n)$.  This is done purely for convenience, and since the dependence on $\epsilon$ is not our main focus; the precise expressions with $\log\frac{1}{\epsilon}$ factors are easily inferred from the proofs.  Similarly, since the low-SNR regime is not our key focus, we assume that $\SNR \ge 2$, and thus $\log\SNR$ is positive.

% The signal to noise ratio of the input signal is an important parameter governing the performance of our algorithms, as we see below.

\paragraph{\underline{Contributions.}}

We proceed by informally stating our main result; a formal statement is given in Section \ref{sec:final_statement}.

\begin{thm} \label{thm:final} \emph{(Upper bound -- informal version)}
    There exists an adaptive algorithm for approximating the Fourier transform with $(k_0, k_1)$-block sparsity that achieves the $\ell_2/\ell_2$ guarantee for any constant $\e = \Theta(1)$, with a sample complexity of $O^*\big((k_0 k_1 + k_0\log(1+k_0) \log n) \log \SNR)$, and a runtime of $O^*\big((k_0 k_1 \log^3 n + k_0\log(1+k_0) \log^2 n) \log \SNR)$. 
\end{thm}
Note that while we state the result for $\epsilon = \Theta(1)$ here, the dependence on this parameter is explicitly shown in the formal version.

The sample complexity of our algorithm \emph{strictly improves} upon the sample complexity of  $O(k_0k_1 \log n)$ (essentially optimal under the standard sparsity assumption) when $\log(1+k_0)\log\SNR \ll k_1$ and $\log\SNR \ll \log n$ (e.g., $\SNR= O(1)$). %For the scenario when $SNR=O(1)$ and block sizes are large (e.g. $k_1\geq \log k_0 \log n$), our sample complexity is optimal. 
% If  $k_1\geq \log(1+ k_0) \log^2 n$,  we also asymptotically improve upon the {\em runtime} of sparse FFT algorithms, running in time $O(k_0\log(1+k_0) \log^2 n+k_0 k_1)$. \xxx{we should add this claim}

%For the $(k_0,k_1)$-block sparse model in the Fourier spectrum, our main contribution is an algorithm whose sample complexity is $O(k+ k_0\log k_0 \log n )$ for a constant signal-to-noise ratio (SNR).  This strictly improves on the complexity $O(k_0k_1 \log n)$ for arbitrary sparse signals in several scaling regimes, with the gain being particularly significant when $k_0 \ll k_1$. 

Our algorithm that achieves the above upper bound crucially uses adaptivity. This is in stark contrast with the standard sparse FFT, where we know how to achieve the near-optimal $O(k\log n)$ bound using non-adaptive sampling \cite{IKP}.  While relying on adaptivity can be viewed as a weakness, we provide a lower bound revealing that \emph{adaptivity is essential} for obtaining the above sample complexity gains.  We again state an informal version, which is formalized in Section \ref{sec:lower}.

\begin{thm} \emph{(Lower bound -- informal version)} \label{thm:lb}
    Any non-adaptive sparse FFT algorithm that achieves the $\ell_2/\ell_2$ sparse recovery guarantee with $(k_0, k_1)$-block sparsity must use a number of samples behaving as $\Omega\big( k_0k_1\log \frac{n}{k_0k_1}\big)$.
\end{thm}

To the best of our knowledge, these two theorems provide the first results along several important directions, giving {\bf (a)} the first sublinear-time algorithm for model-based compressive sensing; {\bf (b)} the first model-based result with provable sample complexity guarantees in the Fourier setting; {\bf (c)} the first proven gap between the power of adaptive and non-adaptive sparse FFT algorithms; and {\bf (d)} the first proven gap between the power of structured (Fourier basis) and unstructured (random Gaussian entries)  matrices for model-based compressive sensing. 

To see that {\bf (d)} is true, note that the sample complexity $O(k_0 \log n+k_0 k_1)$ for block-sparse recovery can be achieved {\em non-adaptively} using Gaussian measurements~\cite{Bar10}, but we show that adaptivity is required in the Fourier setting.

% The latter two points are due to the following observations.  For {\bf (c)}, note that the optimal bound on the sample complexity of sparse FFT can be achieved non-adaptively~\cite{}. and it is known that adaptivity does not help significantly~\cite{}, whereas our approach uses adaptivity to achieve the sample-optimal\xxx{not quite} bound (in the constant SNR regime) in the model-based setting, and we show that adaptivity is necessary. 

{\bf Dependence of our results on SNR}. The sample complexity and runtime of our upper bound depend logarithmically on the SNR of the input signal.  This dependence is common for sparse FFT algorithms, and even for the case of standard sparsity, algorithms avoiding this dependence in the runtime typically achieve a suboptimal sample complexity \cite{HIKP, HIKP2}. Moreover, to our knowledge, all existing sparse FFT lower bounds consider the constant SNR regime (e.g., \cite{DIPW, PW, HIKP2}). 

We also note that our main result, as stated above, assumes that upper bounds on the SNR and the tail noise are known that are tight to within a constant factor (in fact, such tightness is not required, but the resulting bound replaces the true values by the assumed values). These assumptions can be avoided at the expense of a somewhat worse dependence on $\log\SNR$, but we prefer to present the algorithm in the above form for clarity. The theoretical guarantees for noise-robust compressive sensing algorithms often require similar assumptions \cite{Fou13}.

% We discuss this assumption from the point of view of lower bounds and algorithms. First,  all lower bounds on the sample complexity  of sparse recovery with Fourier measurements~\cite{DPIW, IW, HIKP2} are proved in the regime of constant SNR. For that setting our bounds are optimal up to lower order terms {\bf not quite true?}. On the algorithmic side, logarithmic dependence of runtime on the SNR is common for sparse FFT algorithms: algorithms for even the vanilla sparse FFT that avoid this dependence, or make it weaker than logarithmic~\cite{IKP}, are only applicable with some restrictions (e.g.~\cite{IKP} achieves both near-optimal sample complexity and $O^*(k\log^2 n)$ runtime, but only in dimension 1), or do not achieve optimal sample complexity~

\vspace{-3mm}
\paragraph{\underline{Our techniques}:} At a high level, our techniques can be summarized as follows: 

\textbf{Upper bound.} The high-level idea of our algorithm is to \emph{reduce} the $(k_0,k_1)$-block sparse signal of length $n$ to a number of downsampled $O(k_0)$-sparse signals of length $\frac{n}{k_1}$, and use standard sparse FFT techniques to locate their dominant values, thereby identifying the dominant blocks of the original signal.  Once the blocks are located, their values can be estimated using hashing techniques.  Despite the high-level simplicity, this is a difficult task requiring novel techniques, the most notable of which is an adaptive \emph{importance sampling} scheme for allocating sparsity budgets to the downsampled signals.  Further details are given in Section \ref{sec:overview}.

\textbf{Lower bound.} Our lower bound for non-adaptive algorithms follows the information-theoretic framework of~\cite{PW}, but uses a significantly different ensemble of \emph{structured} approximately block-sparse signals occupying only a fraction $O\big(\frac{1}{k_0k_1}\big)$ of the time domain.  Hence, whereas the analysis of \cite{PW} is based on the difficulty of identifying one of (roughly) ${n \choose k}$ sparsity patterns, the difficulty in our setting is in \emph{non-adaptively} finding where the signal is non-zero -- one must take enough samples to cover the various possible time domain locations.  The details are given in Section \ref{sec:lower}.

Interestingly, our upper bound uses adaptivity to circumvent the difficulty exploited in this lower bounding technique, by \emph{first} determining where the energy lies, and \emph{then} concentrating the rest of its samples on the ``right'' parts of the signal.

%\if 0
%
%\paragraph{\underline{Related Work}.}
%For an in-depth exposition of the related work on sparse-FFT, we refer to the survey articles \cite{gilbert2014recent,gilbert2008tutorial}. 
%
%\xxx{cite relevant papers}
%
%\noindent Standard sparse FT algorithms:
%\begin{itemize}
%    \item Overview papers (Gilbert/Indyk/Iwen/Schmidt 2014, Gilbert/Strauss/Tropp 2008)
%    \item Near-optimal $\ell_2/\ell_2$ guarantees (HIKP STOC 2012, IKP 2014)
%    \item Higher dimensions and measurement re-use (IK FOCS 2014, K STOC 2016)
%    \item $\ell_{\infty}/\ell_2$ guarantees? (HIKP SODA 2012 and references after their eq. (1.2))
%\end{itemize}
%
%\noindent Model-based works:
%\begin{itemize}
%    \item Original paper \cite{Bar10}
%    \item Sublinear-time paper (Kyrillidis/Cevher 2012, arXiv 1203.4746)
%    \item Approximation algorithms (Hegde/Indyk/Schmidt 2014, arXiv 1406.1579)
%    \item Specific block-sparse papers (Eldar \emph{et al.} T-SP 2010, Stojnic \emph{.} T-SP 2009, Kyrillidis \emph{et al.} AISTATS 2016)
%    \item \emph{State the Gaussian measurement guarantees and lower bounds}
%\end{itemize}
%\fi

\paragraph{\underline{Notation}:} For an even number $n$, we define $[n] := \big(-\frac{n}{2}, \frac{n}{2}\big] \cap \ZZ$, where $\ZZ$ denotes the integers.    When we index signals having a given length $m$, all arithmetic should be interpreted as returning values in $[m]$ according to modulo-$m$ arithmetic. For $x,y \in \CC$ and $\Delta \in \RR$, we write $y = x \pm \Delta$ to mean $|y - x| \le \Delta$.  The support of a vector $X$ is denoted by $\supp(X)$.  For a number $a\in \mathbb{R}$, we write $|a|_+:=\max\{0, a\}$ to denote the positive part of $a$.

\paragraph{\underline{Organization}:}
The paper is organized as follows. In Section~\ref{sec:overview}, we provide an outline of our algorithm and the main challenges involved.  We formalize our energy-based importance sampling scheme in Section \ref{sec:ep_sampling}, and provide the corresponding techniques for energy estimation in Section \ref{sec:prelim}.  The block-sparse FFT algorithm and its theoretical guarantees are given in Section \ref{sec:full_alg}, and the lower bound is presented and proved in Section \ref{sec:lower}.  Several technical proofs are relegated to the appendices.

%!TEX root = BlockSparseFT-new.tex

\section{Overview of the Algorithm} \label{sec:overview}

One of our key technical contributions consists of a reduction from the $(k_0, k_1)$-block sparse recovery problem for signals of length $n$ to $O(k_0)$-sparse recovery on a set of carefully-defined {signals} of {\em reduced length} $n/k_1$, in sublinear time. We outline this reduction below.
%assuming that the signal $\wh{X}$ is exactly $(k_0, k_1)$-block sparse. 

A basic candidate reduction to $O(k_0)$-sparse recovery consists of first convolving $\wh{X}$ with a filter $\wh{G}$ whose support approximates the indicator function of the interval $[-k_1/2, k_1/2]$, and then considering a new signal whose Fourier transform consists of samples of $\wh{X}\star \wh{G}$ at multiples of $k_1$. The resulting signal $\wh{Z}$ of length $n/k_1$  {\bf (a)} naturally represents $\wh{X}$, as  every frequency of this sequence is a (weighted) sum of the frequencies in the corresponding block, and {\bf (b)} can be accessed in time domain using a small number of accesses to $X$ (if $G$ is compactly supported; see below).

This is a natural approach, but its vanilla version does not work: Some blocks in $\wh{X}$ may entirely cancel out, not contributing to $\wh{Z}$ at all, and other blocks may add up constructively and contribute an overly large amount of energy to $\wh{Z}$. To overcome this challenge, we consider not one, but rather $2k_1$ reductions: For each $r\in [2k_1]$, we apply the above reduction to the {\em shift of $X$ by $r\cdot \frac{n}{2k_1}$ in time domain}, and call the corresponding vector $Z^r$.  We show that all shifts cumulatively capture the energy of $X$ well, and the major contribution of the paper is an algorithm for locating the dominant blocks in $\wh{X}$  from a small number of accesses to the $Z^r$'s (via an importance sampling scheme).

\textbf{Formal definitions:} We formalize the above discussion in the following, starting with the notion of a \emph{flat filter} that approximates a rectangle.

\begin{defn}[Flat filter] \label{def:filterG}
    A sequence $G \in \RR^n$ with Fourier transform $\wh{G} \in \RR^n$ symmetric about zero is called an \emph{$(n,B,F)$-flat filter} if (i) $\wh{G}_f \in [0,1]$ for all $f \in [n]$; (ii) $\wh{G}_f \ge 1 - \big(\frac{1}{4}\big)^{F-1}$ for all  $f \in [n]$ such that $|f| \le \frac{n}{2B}$; and (iii) $\wh{G}_f \le \big(\frac{1}{4}\big)^{F-1} \big( \frac{n}{B|f|} \big)^{F-1}$ for all $f \in [n]$ such that $|f| \ge \frac{n}{B}$.
\end{defn}

The following lemma, proved in Appendix \ref{sec:pf_filter}, shows that it is possible to construct such a filter having $O(FB)$ support in time domain.

\begin{lem} \label{lem:filter_properties}
    \emph{(Compactly supported flat filter)}
    Fix the integers $(n,B,F)$ with $n$ a power of two, $B < n$, and $F \ge 2$ an even number.  There exists an $(n,B,F)$-flat filter $\wh{G} \in \RR^n$, which (i) is supported on a length-$O(FB)$ window centered at zero in time domain, and (ii) has a total energy satisfying $\sum_{f \in [n]} | \wh{G}_f |^2 \le \frac{3n}{B}$.
    \label{lem:filter}
\end{lem}

     Throughout the paper, we make use of the filter construction from Lemma \ref{lem:filter_properties}, except where stated otherwise.  To ease the analysis, we assume that $G$ and $\wh{G}$ are pre-computed and can be accessed in $O(1)$ time.  Without this pre-computation, evaluating $\wh{G}$ is non-trivial, but possible using semi-equispaced Fourier transform techniques (\emph{cf.}, Section \ref{sec:semi_equi}).

    With the preceding definition, the set of $2k_1$ downsampled signals is given as follows.
    
    \begin{defn}[Downsampling] \label{def:downsampling}
        Given integers $(n,k_1)$, a parameter  $\delta \in \big(0, \frac{1}{20}\big)$, and a signal $X \in \CC^n$, we say that the set of signals $\{Z^r\}_{r \in [2k_1]}$ with $Z^r \in \CC^{\frac{n}{k_1}}$ is a \emph{$(k_1,\delta)$-downsampling} of $X$ if 
            $$Z^r_j= \frac{1}{k_1} \sum_{i\in [k_1]}(G \cdot X^r)_{j+\frac{n}{k_1} \cdot i}, ~~~ j \in \Big[\frac{n}{k_1}\Big]$$
        for an $\big(n,\frac{n}{k_1},F\big)$-flat filter with $F = 10\log\frac{1}{\delta}$ and support $O\big(F\frac{n}{k_1}\big)$, where we define $X^r_i = X_{i+a_r}$ with $a_r = \frac{nr}{2k_1}$.  Equivalently, in frequency domain, this can be written as
        \begin{equation}
            \wh{Z}^r_j = (\wh{X}^r \star \wh{G} )_{j k_1} = \sum_{f\in [n]}\wh{G}_{f-k_1 \cdot j}  \wh{X}_f  \omega_n^{a_r \cdot f}, ~~~ j\in \Big[\frac{n}{k_1}\Big] \label{equation:4}
        \end{equation}
        by the convolution theorem and the duality of subsampling and aliasing (e.g., see Appendix \ref{sec:pf_uhat}).
    \end{defn}
    
    By the assumption of the bounded support of $G$, along with the choice of $F$, we immediately obtain the following lemma, showing that we do not significantly increase the sample complexity by working with $\{Z^r\}_{r \in [2k_1]}$ as opposed to $X$ itself.

    \begin{lem} \label{lem:downsamp-cost-unit-access}
        \emph{(Sampling the downsampling signals)}
        Let $\{Z^r\}_{r \in [2k_1]}$ be a $(k_1,\delta)$-downsampling of $X \in \CC^n$ for some $(n,k_1,\delta)$.  Then any single entry $Z_i^r$ can be computed in $O\big(\log\frac{1}{\delta}\big)$ time using $O\big(\log\frac{1}{\delta}\big)$ samples of $X$.
    \end{lem}

%    Specifically, we define, for each $r\in [2k_1]$,  
%      $$Z^r_j= \frac{1}{k_1} \sum_{i\in [k_1]}(G \cdot X^r)_{j+\frac{n}{k_1} \cdot i}, \text{~~~for~} j \in \Big[\frac{n}{k_1}\Big],$$
%    so that 
%    $$
%        \wh{Z}^r_j = (\wh{X}^r \star \wh{G} )_{j k_1} = \sum_{f\in [n]}\wh{G}_{f-k_1 \cdot j}  \wh{X}_f  \omega_n^{a_r \cdot f}, \text{~~~for~}j\in \Big[\frac{n}{k_1}\Big]
%    $$
%    by the convolution theorem and the duality of subsampling and aliasing.  Here $\wh{G}$ approximates the indicator function of an interval of length $k_1$, and $G$ is supported on an interval of length about $n/k_1$, so that access to $Z^r$ is cheap given access to $X$.  The formal definitions are given in Section \ref{sec:prelim}; in particular, see Definition \ref{def:downsampling}.
    
    This idea of using $2k_1$ reductions fixes the above-mentioned problem of constructive and destructive cancellations: The  $2k_1$ reduced signals $Z^r$ ($r\in [2k_1]$) cumulatively capture all the energy of $X$ well.   That is, while the energy $|\wh{Z}_j^r|_2^2$ can vary significantly as a function of $r$, we can tightly control the behavior of the sum $\sum_{r \in [2k_1]}|\wh{Z}_j^r|_2^2$.  This is formalized in the following.
    
\begin{lem} \label{claim:1}
    \emph{(Downsampling properties)}
    Fix $(n,k_1)$, a parameter $\delta \in \big(0, \frac{1}{20}\big)$, a signal $X \in \C^n$, and a $(k_1,\delta)$-downsampling $\{Z^r\}_{r \in [2k_1]}$ of $X$.  The following conditions hold:
    \begin{enumerate}
        \item For all $j \in [\frac{n}{k_1}]$,
            \begin{equation}
              \frac{\sum_{r\in[2k_1]} |\wh{Z}^r_{j}|^2 }{2k_1}\ge ( 1-\delta) \|\wh{X}_{I_{j}}\|_2^2 - 3\delta \cdot \bigg( \|\wh{X}_{I_j \cup I_{j-1} \cup I_{j+1}}\|_2^2 + \delta \sum _{j'\in[\frac{n}{k_1}]\backslash \{j\}} \frac{\|\wh{X}_{I_{j'}}\|_2^2}{|j'-j|^{F-1}}\bigg).  \nn
         \end{equation}
         \item The total energy satisfies $(1 - 12\delta)\|\wh{X}\|_2^2 \le \frac{\sum_{r \in [2k_1]} \|\wh{Z}^r\|_2^2 }{2k_1}\le 6 \|\wh{X}\|_2^2.$
    \end{enumerate}
\end{lem}

\noindent The proof is given in Appendix \ref{sec:pf_downsampling}.
    
    \textbf{Location via sparse FFT:}  We expect each $Z^r$ ($r\in[2k_1]$) to be approximately $O(k_0)$-sparse, as every block contributes \emph{primarily} to one downsampled coefficient.  At this point, a natural step is to run $O(k_0)$-sparse recovery on the signals $Z^r$ to recover the dominant blocks. However, there are too many signals $Z^r$ to consider!  Indeed, if we were to run $O(k_0)$-sparse recovery on every $Z^r$, we would recover the locations of the blocks, but at the cost of $O(k_0 k_1 \log n)$ samples. This precludes any improvement on the vanilla sparse FFT.  
    
    It turns out, however, that it is possible to avoid running a $k_0$-sparse FFT on all $2k_1$ reduced signals, and to instead \emph{allocate budgets} to them, some of which are far smaller than $k_0$, and some of which may be zero.  This will be key in reducing the sample complexity.
    
    Before formally defining budget allocation, we present the following definition and lemma, showing that we can use less samples to identify less of the dominant coefficients of a signal, or more samples to identify more dominant coefficients.
    
    \begin{defn} \label{def:covered}
    \emph{(Covered frequency)}
    Given an integer $m$, a frequency component $j$ of a signal $\wh{Z} \in \CC^m$ is called \emph{covered} by budget $s$ in the signal $\wh{Z}$ if
        $|\wh{Z}_j|^2 \geq \frac{\|\wh{Z}\|_2^2}{s}$.
    \end{defn}
    
    \begin{lem} \emph{(\textsc{LocateReducedSignal} guarantees -- informal version) } \label{lem:loc_k}
        There exists an algorithm such that if a signal $X \in \CC^n$, a set of budgets $\{s^r\}_{r \in [2k_1]}$, and a confidence parameter $p$ are given to it as input, then it outputs a list that, with probability at least $1-p$, contains any given $j\in [\frac{n}{k_1}]$ that is covered by $s^r$ in $\wh{Z}^r$ for some $r \in [2k_1]$, where $\{\wh{Z}^r\}_{r \in [2k_1]}$ denotes the $(k_1,\delta)$-downsampling of $X$. Moreover, the list size is $O\big( \sum_{r \in [2k_1]} s^r \big)$, the number of samples that the algorithm takes is $O\big(\sum_{r \in [2k_1]}s^r \log n\big)$, and the runtime is $O\big(\sum_{r \in [2k_1]} s^r \log^2 n\big)$.\footnote{ As stated in the formal version, additional terms in the runtime are needed when it comes to subtracting off a current estimate to form a residual signal. }
    \end{lem}

     The formal statement and proof are given in Appendix \ref{sec:loc_k}, and reveal that $s^r$ essentially dictates how many buckets we hash $\wh{Z}^r$ into in order to locate the dominant frequencies (e.g., see \cite{HIKP2,IKP}).
    
    Hence, the goal of budget allocation is to approximately solve the following covering problem:
    \begin{equation}
        \text{Minimize}_{ \{s^r\}_{r\in[2k_1]} } ~ \sum_{r \in [2k_1]} s^r ~~~  \text{ subject to } ~~  \sum_{\substack{j\text{ is covered by } s_r \\ \text{ in } \wh{Z}^r \text{ for some }r\in[2k_1]}} \|\wh{X}_{I_j}\|_2^2 \ge (1 - \alpha) \cdot \|\wh{X}^*\|_2^2,\label{eq:budget_alloc}
    \end{equation}
    for a suitable constant $\alpha \in (0,1)$, where $s^r$ is the budget allocated to $\wh{Z}^r$, and $\wh{X}^*$ is the best $(k_0.k_1)$-block sparse approximation of $\wh{X}$.  That is, we want to minimize the total budget while accounting for a constant proportion of the signal energy.    
    
    \textbf{Challenges in budget allocation:} Allocating the budgets is a challenging task, as each block in the spectrum of the signal may have very different energy concentration properties in time domain, or equivalently, different variations in $|\wh{Z}_j^r|^2$ as a function of $r$.  To see this more concretely, in Figure \ref{fig:MatrixExample}, we show three \emph{hypothetical} examples of such variations, in the case that $k_0 = 2k_1 = 6$ and all of the blocks have equal energy, leading to equal column sums in the matrices. 
    
    \begin{figure}
    	\centering
        \includegraphics[width=0.9\textwidth]{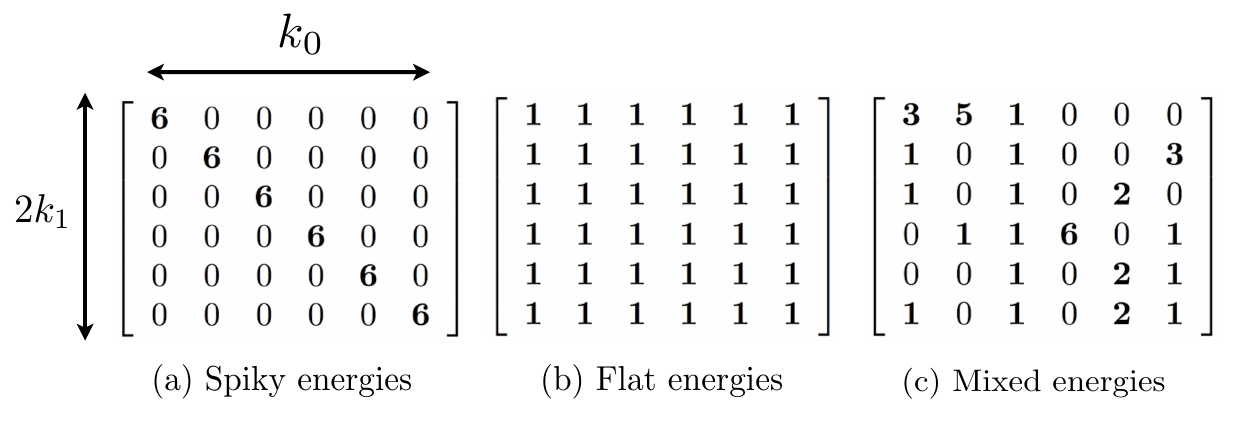}
        \par
        
    	\caption{Three hypothetical examples of matrices with $(r,j)$-th entry given by $|\wh{Z}_j^r|^2$, i.e., each row corresponds to a single sequence $Z^r$, but only at the entries corresponding to the $k_0$ blocks in $X$. \label{fig:MatrixExample}} \vspace*{-2ex}
    \end{figure}
    
    In the first example, each block contributes to a different $Z^r$, and thus the blocks could be located by running $1$-sparse recovery separately on the $2k_1$ signals.  In stark contrast, in the second example, each block contributes equally to each $Z^r$, so we would be much better off running $k_0$-sparse recovery on a single (arbitrary) $Z^r$.  Finally, in the third example, the best budget allocation scheme is completely unclear by inspection alone! We need to design an allocation scheme to handle all of these cases, and to do so \emph{without even knowing the structure of the matrix}.
    
    While the examples in Figure \ref{fig:MatrixExample} may seem artificial, and are not necessarily feasible with the \emph{exact} values given, we argue in Appendix \ref{sec:discussion} that situations exhibiting the same general behavior are entirely feasible.

    % For example, a block all of whose coefficients are equal to each other would have essential support of roughly $n/k_1$ (it would be a sinc function), whereas a block with uncorrelated frequencies would have energy roughly uniformly spread over time domain. The different reductions $Z^r$ correspond to probing the signal on distinct areas of time domain, and we need to allocate a small budget of samples that works {\em simultaneously for all blocks}. 
    
    % For example, if all blocks are highly concentrated in time domain, but each on a disjoint part (an example is given in Appendix~\ref{sec:logk0}), we would run sparse recovery on $\approx k_0$ of the signals $Z^r$, but hash every such signal into $O(1)$ buckets in each vanilla sparse FFT invocation. On the other hand, if all blocks have uncorrelated frequencies and energy is well spread, we would run only a constant number of vanilla sparse FFT procedures, but hash into $\approx k_0$ buckets each time. 
    
    %To tackle the above challenge, we adopt an {\em importance sampling scheme} (see Section~\ref{sec:ep_sampling}) that accesses the signal $X$ in $O(k_0 k_1)$ locations and computes a plan for running vanilla sparse FFT on a subset of reduced signals $Z^r$ such that the total number of samples taken by  sparse FFT on reduced signals is only $O(k_0 \log(1+k_0) \log n)$ as opposed to $O(k_0 k_1 \log n)$. In fact, our plans are, and have to be, very flexible: 
    
    \textbf{Importance sampling:} Our solution is to \emph{sample $r$ values with probability proportional to an estimate of $\|\wh{Z}^r\|_2^2$}, and sample sparsity budgets from a carefully defined distribution (see Section~\ref{sec:ep_sampling}, Algorithm~\ref{alg:2}).  We show that sufficiently accurate estimates of $\|\wh{Z}^r\|_2^2$ for all $r \in [2k_1]$ can be obtained using $O(k_0 k_1)$ samples of $X$ via hashing techniques (\emph{cf.}, Section \ref{sec:prelim}); hence, what we are essentially doing is using these samples to determine where most of the energy of the signal is located, and then favoring the parts of the signal that appear to have more energy.  This is exactly the step that makes our algorithm adaptive, and we prove that it produces a total budget in \eqref{eq:budget_alloc} of the form $O(k_0 \log(1+k_0))$, on average.
    
    % The number of signals $Z^r$ that we run the sparse FFT on may in fact be quite large (up to $O(k_0)$).  However, each invocation of the sparse FFT on a reduced signal hashes the reduced signal into a different number of buckets, and the {\em cumulative} number of buckets we use is only $O(k_0 \log(1+k_0))$, thus requiring $O(k_0 \log(1+k_0) \log n)$ samples as opposed to $O(k_0 k_1 \log n)$.
    
     Ideally, one would hope to solve \eqref{eq:budget_alloc} using a total budget of $O(k_0)$, since there are only $k_0$ blocks.  However, the $\log(1+k_0)$ factor is not an artifact of our analysis: We argue in Appendix \ref{sec:discussion} that very different techniques would be needed to remove it in general.  Specifically, we design a signal $X$ for which the \emph{optimal} solution to \eqref{eq:budget_alloc} indeed satisfies $\sum_{r \in [2k_1]} s^r = \Omega(k_0 \log(1+k_0))$.

    \textbf{Iterative procedure and updating the residual:} The techniques described above allow us to recover a list of blocks that contribute a constant fraction (e.g., $0.9$) of the signal energy. We use $O(\log \SNR)$ iterations of our main procedure to reduce the SNR to a constant, and then achieve $(1+\e)$-recovery with an extra ``clean-up'' step. Most of the techniques involved in this part are more standard, with a notable exception: Running a standard sparse FFT with budgets $s^r$ on the reduced space (i.e., on the vectors $Z^r$) is not easy to implement in $k_0 k_1 \poly(\log n)$ time when $Z^r$ are the {\em residual signals}.  The natural approach is to subtract the current estimate $\wh{\chi}$ of $\wh{X}$ from our samples and essentially run on the residual, but subtraction in $k_0 k_1 \poly(\log n)$ time is not straightforward to achieve. Our solution crucially relies on a novel block semi-equispaced FFT (see Section~\ref{sec:semi_equi}), and the idea of letting the location primitives in the reduced space operate using common randomness (see Appendix~\ref{sec:loc_k}).

\if 0
    
    with one notable exception: We use a novel semi-equispaced FFT procedure to update the signal in time domain. This is based on known semi-equispaced FFT results, but a new twist is required to update all reductions $Z^r$ in $k_0 k_1 \poly(\log n)$ time due to the block structure (see Section~\ref{sec:semi_equi}). 
    
    \fi 
%!TEX root = BlockSparseFT-new.tex

\setlength{\abovedisplayskip}{4pt}
\setlength{\belowdisplayskip}{4pt}

\section{Location via Importance Sampling} \label{sec:ep_sampling}

As outlined above, our approach locates blocks by applying standard sparse FFT techniques to the downsampled signals arising from Definition \ref{def:downsampling}.  In this section, we present the techniques for assigning the corresponding sparsity budgets (\emph{cf.}, \eqref{eq:budget_alloc}).

% to each $\wh{Z}^r$ such that running a sparse recovery on each of them with the assigned budget successfully locates heavy blocks with high probability. These assumptions are summarized in the following Definition.

We use a novel procedure called \emph{energy-based importance sampling}, which approximately samples $r$ values with probability proportional to $\|\wh{Z}^r\|^2$.  Since these energies are not known exactly, we instead sample proportional to a general vector $\gamma = (\gamma^1,\dotsc,\gamma^{2k_1})$, where we think of $\gamma^r$ as approximating $\|\wh{Z}^r\|^2$.  The techniques for obtaining these estimates are deferred to Section \ref{sec:prelim}.

The details are shown in Algorithm  \ref{alg:2}, where we repeatedly sample from the distribution $w_q^r$, corresponding to independently sampling $r$ proportional to $\gamma^r$, and $q$ from a truncated geometric distribution.  The resulting sparsity level to apply to $Z^r$ is selected to be $s^r = 10\cdot 2^q$.

%In more detail, the procedure is as follows for some integer $k_0$, $p \in \big(0,\frac{1}{2}\big)$, and $\delta \in (0,1)$ (see Algorithm \ref{alg:2} for a formal statement):\footnote{Here and subsequently, $\delta$ can be viewed as being $O(\e)$ when we seek $(1+\e)$-approximate sparse recovery.}
%\begin{itemize}
%    \item Take $\frac{10}{\delta} \cdot k_0 \cdot\log \frac{1}{p} $ samples of pairs $(r,q)$ from the distribution 
%   $$w^r_q = \frac{2^{-q} \cdot \gamma^r}{(1-\frac{\delta}{10 k_0})\|\gamma\|_1}, \quad q \in \Big\{1,\dotsc, \log_2 \frac{10 k_0}{\delta} \Big\}.$$
%    \item For each sampled pair $(r,q)$, assign $s^r = 10 \cdot 2^q $ to the signal $\wh{Z}^r$, taking the maximum over $q$ if the same $r$ is sampled multiple times, and setting $s^r = 0$ if $r$ is never sampled. 
%\end{itemize}

According to Definition \ref{def:covered}, $s^r = 10\cdot 2^q$ covers any given frequency $j$ for which $|\wh{Z}^r_j|^2 \ge \frac{ \|\wh{Z}^r\|_2^2 }{ 10 \cdot 2^{q} }$.  The intuition behind sampling $q$ proportional to $2^{-q}$ is that this gives a high probability of producing small $q$ values to cover the heaviest signal components, while having a small probability of producing large $q$ values to cover the smaller signal components.  We only want to do the latter rarely, since it costs significantly more samples.

\begin{algorithm}
\caption{Procedure for allocating sparsity budgets to the downsampled signals}\label{euclid}
\begin{algorithmic}[1]

\Procedure{BudgetAllocation}{$\gamma, k_0, k_1, \delta, p$} %\Comment $C$ is an absolute constant

\State $S \gets \emptyset$

\For{\texttt{$i \in \{1,\dotsc,\frac{10}{\delta} k_0 \cdot \log \frac{1}{p} \}$}}

\State Sample $(r_i,q_i) \in [2k_1] \times \{1,\dotsc,\log_{2} \frac{10 k_0}{\delta} \}$ with probability $w^r_q = \frac{2^{-q}}{1- \delta/(10 k_0)} \frac{\gamma^r}{\|\gamma\|_1}$ \label{line:wrq}

\State $S \gets S \cup \{ (r_i,q_i) \}$

\EndFor

\For{\texttt{$r \in [2 k_1]$}}

\State { $q^* \gets \max_{(r,q') \in S} \{q'\}$} \Comment By convention, $\max \emptyset = -\infty$

\State $s^r \gets 10 \cdot 2^{q^*}$

\EndFor

\State \textbf{return} $s = [s^r]_{r\in[2k_1]}$

\EndProcedure

\end{algorithmic}
\label{alg:2}

\end{algorithm}

% \begin{equation}
%     \frac{1}{\|\gamma\|_2^2} \cdot \sum_{r \in [2k_1]} \Big| \gamma^r - \gamma^r \Big| \le O(\e), \label{eq:74}
% \end{equation}
% where $\gamma^r$ is the $r$-th element of $\gamma$.

We first bound the expected total sum of budgets returned by \textsc{BudgetAllocation}.

\begin{lem}
\emph{(\textsc{BudgetAllocation} budget guarantees)}
For any integers $k_0$ and $k_1$, any positive vector $\gamma \in \RR^{2k_1}$, and any parameters $p \in \big(0,\frac{1}{2}\big)$ and $\delta \in (0,1)$, if the procedure \textsc{BudgetAllocation} in Algorithm \ref{alg:2} is run with inputs $(\gamma, k_0, k_1, \delta , p)$, then the expected value of the total sum of budgets returned, $\{s^r\}_{r \in [2k_1]}$, satisfies $\EE \big[ \sum_{r \in [2k_1]} s^r \big] \le 200\,\frac{k_0}{\delta}  \log \frac{k_0}{\delta} \log \frac{1}{p}$.  The runtime of the procedure is $O\big( \frac{k_0}{\delta}\log\frac{1}{p} + k_1\big)$.
    \label{lemm:10}
\end{lem}
\begin{proof}
    Each time a new $(r,q)$ pair is sampled, the sum of the $s^r$ values increases by at most $10\cdot 2^q$, and hence the overall expected sum is upper  bounded by the number of trials $10k_0\log\frac{1}{p}$ times the expected value of $10\cdot 2^q$ for a single trial:
    \begin{equation}
    \begin{split}
        \EE \Big[ \sum_{r \in [2k_1]} s^r \Big] 
        	&\le \frac{10}{\delta} k_0 \cdot \log \frac{1}{p} \sum_{r \in [2k_1]} \sum_{q = 1}^{\log_{2} \frac{10 k_0}{\delta}} w^r_q \cdot 10 \cdot 2^q \\
            &= \frac{ 100 k_0 \log \frac{1}{p} }{\delta} \cdot \frac{1}{1- \delta/(10 k_0)} \sum_{q = 1}^{\log_{2} \frac{10 k_0}{\delta}} \sum_{r \in [2k_1]}  \frac{\gamma^r}{\|\gamma\|_1} \\
            &\le 200 \frac{k_0}{\delta}  \log \frac{k_0}{\delta} \log \frac{1}{p},
    \end{split} \nn
    \end{equation}
    where the second line follows from the definition of $w_q^r$, and the third line follows from $\frac{\delta}{10k_0} \le \frac{1}{2}$ (since $\delta \le 1$) and $\sum_{r \in [2k_1]}  \frac{\gamma^r}{\|\gamma\|_1} = 1$.
    
   \paragraph{Runtime:} Note that sampling from $w_q^r$ amounts to sampling $q$ and $r$ values independently, and the corresponding alphabet sizes are $O\big(\log\frac{k_0}{\delta}\big)$ and $O(k_1)$ respectively.  The stated runtime follows since we take $O\big(\frac{k_0}{\delta} \log\frac{1}{p}\big)$ samples, and sampling from discrete distributions can be done in time linear in the alphabet size and number of samples \cite{Hag93}.  The second loop in Algorithm \ref{alg:2} need not be done explicitly, since the maximum $q$ value can be updated after taking each sample.
\end{proof}
%\begin{proofsketch}
%    Since the truncated geometric probabilities are proportional to $2^{-q}$ but a given sample
%    produces a sparsity level proportional to $2^q$, the resulting mean sparsity for a single sample is the truncation point $\log_2\frac{10k_0}{\delta}$.  The lemma follows from the fact that we take $O\big(\frac{k_0}{\delta}\log\frac{1}{p}\big)$ such samples.  See Appendix \ref{sec:pf_eps_samples} for the details.
%\end{proofsketch}

As we discussed in Section \ref{sec:overview}, the $\log k_0$ term in the number of samples would ideally be avoided; however, in Appendix \ref{app:logk0}, we argue that even the optimal solution to \eqref{eq:budget_alloc} can contain such a factor.

We now turn to formalizing the fact that the budgets returned by \textsc{BudgetAllocation} are such that most of the dominant blocks are found.  To do this, we introduce the following notion.

\begin{defn}[Active frequencies] \label{eq:s-tilde}
Given $(n,k_0,k_1)$, a signal $X \in \CC^n$, a parameter $\delta \in (0,1)$, and a $(k_1,\delta)$-downsampling $\{Z^r\}_{r \in [2k_1]}$ of $X$, the set of \emph{active frequencies} $\tilde{S}$ is defined as
\begin{equation}
\tilde{S} = \Big\{ j \in \Big[\frac{n}{k_1}\Big] \, : \,  \sum_{r \in [2k_1]} \Big( |\wh{Z}^r_j|^2 \cdot \frac{\gamma^r}{\|\wh{Z}^r\|_2^2} \Big) \geq \delta \cdot \frac{\sum_{r \in [2k_1]} \|\wh{Z}^r\|_2^2}{k_0} \Big\}. \label{eq:active}
\end{equation}
\label{defn15}
\end{defn}
Observe that if $\gamma^r = \|\wh{Z}^r\|_2^2$, this reduces to $\sum_{r \in [2k_1]} |\wh{Z}^r_j|^2  \geq \delta \cdot \frac{\sum_{r \in [2k_1]} \|\wh{Z}^r\|_2^2}{k_0}$, thus essentially stating that the sum of the energies over $r \in [2k_1]$ for the given block index $j$ is an $\Omega\big( \frac{\delta}{k_0} \big)$ fraction of the total energy.  Combined with Lemma \ref{claim:1}, this roughly amounts to $\|\wh{X}_{I_j}\|_2^2$ exceeding an $\Omega\big( \frac{\delta}{k_0} \big)$ fraction of $\|\wh{X}\|_2^2$.

To formalize and generalize this intuition, the following lemma states that the frequencies within $\tilde{S}$ account for most of the energy in $X$, as long as each $\gamma^r$ approximates $\|\wh{Z}^r\|_2^2$ sufficiently well. 

\begin{lem} \label{lemm:11}
\emph{(Properties of active frequencies)}
Fix $(n,k_0,k_1)$, a parameter $\delta \in \big(0,\frac{1}{20}\big)$, a signal $X \in \CC^n$, and a $(k_1,\delta)$-downsampling $\{Z^r\}_{r \in [2k_1]}$ of $X$.  Moreover, fix an arbitrary set $S^* \subseteq \big[\frac{n}{k_1}\big]$ of cardinality at most $10k_0$, and a vector $\gamma \in \RR^{2k_1}$ satisfying
\begin{gather}
    \sum_{r \in [2k_1]} \Big| \|\wh{Z}^r_{S^*}\|_2^2 - \gamma^r \Big|_+ \le  40 \delta \sum_{r \in [2k_1]}  \|\wh{Z}^r\|_2^2. \tag{*}
\end{gather}
Fix the set of active frequencies $\tilde{S}$ according to Definition \ref{defn15}, and define the signal $\wh{X}_{\tilde{S}}$ to equal $\wh{X}$ on all intervals $\{I_j \, ; j \in \tilde{S}\}$ (see Definition \ref{def:block}), and zero elsewhere. Then $\|\wh{X}_{S^* \backslash \tilde{S}}\|_2^2 \leq 100\sqrt{\delta} \|\wh{X}\|_2^2$.
\end{lem}
\noindent The proof of Lemma~\ref{lemm:11} is given in Appendix \ref{sec:pf_active}.

What remains is to show that if $j$ is active, then $j$ is covered by some $s^r$ in $\wh{Z}^r$ with high constant probability upon running Algorithm \ref{alg:2}.  This is formulated in the following.

\begin{lem}
\emph{(\textsc{BudgetAllocation} covering guarantees)}
Fix $(n,k_0,k_1)$, the parameters $\delta \in (0,1)$ and $p \in \big(0,\frac{1}{2}\big)$, a signal $X \in \CC^n$, and a $(k_1,\delta)$-downsampling $\{Z^r\}_{r \in [2k_1]}$ of $X$.  Moreover, fix a vector $\gamma \in \RR^{2k_1}$ satisfying
\begin{equation}
    \|\gamma\|_1 \le 10 \sum_{r \in [2k_1]} \|\wh{Z}^r\|_2^2. \label{eq58} 
\end{equation} 
Suppose that \textsc{BudgetAllocation} in Algorithm \ref{alg:2} is run with inputs $(\gamma, k_0, k_1, \delta , p)$, and outputs the budgets $\{s^r\}_{r\in[2k_1]}$.  Then for any active $j$ (i.e., $j \in \tilde{S}$ as per Definition \ref{eq:s-tilde}), the probability that there exists some $r \in [2k_1]$ such that $j$ is covered by $s^r$ in $\wh{Z}^r$ is at least $1- p$.
\label{lemm:12}
\end{lem}
\begin{proof}
    Recall from Definition \ref{def:covered} that if a pair $(r,q)$ is sampled in the first loop of \textsc{BudgetAllocation}, then $j$ is covered provided that 
        $ |\wh{Z}^r_j|^2 \ge \frac{ \|\wh{Z}^r\|_2^2 }{ 10 \cdot 2^{q} }. $
    We therefore define
    \begin{equation}
        q_j^r = \min \bigg\{ q \in \ZZ_+ \,:\, |\wh{Z}_j^r|^2 \ge \frac{\|\wh{Z}^r\|_2^2}{ 10 \cdot 2^q } \bigg\}, \label{eq:q_jr}
    \end{equation}
    and note that the event described in the lemma statement is equivalent to some pair $(r,q)$ being sampled with $q^r_j \leq q$.  Note that due to the range of $q$ from which we sample (\emph{cf.}, Algorithm \ref{alg:2}), this can only occur if $q^r_j \le \log_2 \frac{10 k_0}{ \delta }$.
    
    \textbf{Taking a single sample:} We first compute the probability of being covered for a \emph{single} random sample of $(q,r)$, denoting the corresponding probability by $\PP_1[\cdot]$.  Recalling from line \ref{line:wrq} of Algorithm \ref{alg:2} that we sample each $(q,r)$ with probability $w_q^r = \frac{2^{-q}}{1- \delta/(10 k_0)} \frac{\gamma^r}{\|\gamma\|_1}$, we obtain
    \begin{align}
        \PP_1[j\text{ covered}] 
            &= \sum_{r \in [2k_1]} \sum_{q^r_j \le q \le \log_2 \frac{10 k_0}{ \delta }} w_q^r \nn \\
            &= \frac{1}{1- \delta/(10 k_0)} \sum_{r \in [2k_1]} \sum_{q^r_j \le q \le \log_2 \frac{10 k_0}{ \delta }} 2^{-q} \frac{\gamma^r}{\|\gamma\|_1} \nn \\
            &\ge \frac{1}{2} \sum_{r \in [2k_1] \,:\, q_j^r \le \log_2 \frac{10 k_0}{ \delta }} 2^{-q_j^r} \frac{\gamma^r}{\|\gamma\|_1} \nn \\
            &= \frac{1}{2} \sum_{r \in [2k_1] } 2^{-q_j^r} \frac{\gamma^r}{\|\gamma\|_1} - \frac{1}{2} \sum_{r \in [2k_1] \,:\, q_j^r > \log_2 \frac{10 k_0}{ \delta }} 2^{-q_j^r} \frac{\gamma^r}{\|\gamma\|_1}, \label{eq:p1_cov}
        \end{align}
    where the third line follows since $\delta/(10k_0) \le \frac{1}{2}$ due to the assumption that $\delta \le 1$.
    
    \textbf{Bounding the first term in \eqref{eq:p1_cov}:} Observe from \eqref{eq:q_jr} that $2^{-q_j^r} \ge \frac{1}{2} \frac{10|\wh{Z}_j^r|^2}{ \|\wh{Z}^r\|_2^2 } $, and recall the definition of being active in \eqref{eq:active}.  Combining these, we obtain the following when $j$ is active:
    \begin{equation}
        \begin{split}
            \frac{1}{2} \sum_{r \in [2k_1] } 2^{-q_j^r} \frac{\gamma^r}{\|\gamma\|_1} 
                &\ge \frac{10}{4 \|\gamma\|_1} \sum_{r \in [2k_1]}  \frac{|\wh{Z}_j^r|^2 \gamma^r}{ \|\wh{Z}^r\|_2^2 } \\
                &\ge \frac{10 \delta}{4 \|\gamma\|_1} \cdot \frac{\sum_{r \in [2k_1]} \|\wh{Z}^r\|_2^2}{k_0} \ge \frac{\delta}{4 k_0},
        \end{split} \nn
    \end{equation}    
    where the last inequality follows from the assumption on $\|\gamma\|_1$ in the lemma statement.
    
    \textbf{Bounding the second term in \eqref{eq:p1_cov}:} We have
    \begin{equation}
        \begin{split}
            \frac{1}{2} \sum_{r \in [2k_1] \,:\, q_j^r > \log_2 \frac{10 k_0}{ \delta }} 2^{-q_j^r} \frac{\gamma^r}{\|\gamma\|_1}
                &\le \frac{1}{2} \sum_{r \in [2k_1] } \frac{ \delta }{10 k_0} \frac{\gamma^r}{\|\gamma\|_1} = \frac{ \delta}{ 20 k_0 }.
        \end{split} \nn
    \end{equation}
    Hence, we deduce from \eqref{eq:p1_cov} that $\PP_1[j\text{ covered}] \ge \frac{ \delta }{5 k_0}$.
    
    \textbf{Taking multiple independent samples:} Since the sampling is done $\frac{10}{\delta} k_0 \cdot \log \frac{1}{p}$ times independently, the overall probability of an active block $j$ being covered satisfies
    \begin{equation}
    \begin{split}
        \PP[j\text{ covered}] 
            &\ge 1- \bigg( 1 - \frac{ \delta }{5 k_0} \bigg)^{ \frac{10}{\delta} k_0 \cdot \log \frac{1}{p} } \ge 1 - \exp\bigg( -\frac{10 }{5} \cdot \log \frac{1}{p} \bigg) \ge 1-p,
    \end{split} \nn
    \end{equation}
    where we have applied the inequality $1-\zeta \le e^{-\zeta}$ for $\zeta \ge 0$.
\end{proof}

% In Section \ref{sec:est_downsampled}, we will provide the techniques for efficiently construction $\{\gamma^r\}$ satisfying the preconditions of Lemmas \ref{lemm:11} and \ref{lemm:12}.  As we will see, the idea is to hash each $Z^r$ into $O(k_0)$ buckets, and let the energy estimate be the energy of the hashed signal.

\subsection{The Complete Location Algorithm} \label{sec:complete_loc}

In Algorithm \ref{alg:Locate}, we give the details of \textsc{MultiBlockLocate}, which performs the above-described energy-based importance sampling procedure, runs the sparse FFT location algorithm (see Appendix \ref{sec:loc_k}) with the resulting budgets, and returns a list $L$ containing the block indices that were identified.  

\textsc{MultiBlockLocate} calls two primitives that are defined later in the paper, but their precise details are not needed in order to understand the location step:
\begin{itemize}
    \item \textsc{EstimateEnergies} (see Section \ref{sec:est_downsampled}) provides us with a vector $\gamma$ providing a good approximation of each $\|\wh{Z}^r\|_2^2$, in the sense of satisfying the preconditions of Lemmas \ref{lemm:11} and \ref{lemm:12};
    \item \textsc{LocateReducedSignals} (see Appendix \ref{sec:loc_k}) accepts the sparsity budgets $\{s^r\}$ and runs a standard $s^r$-sparse fast FFT algorithm on each downsampled signal $Z^r$ in order to locate the dominant frequencies.
\end{itemize}
Note that in addition to $X$, these procedures accept a second signal $\wh{\chi}$; this becomes relevant when we iteratively run the block sparse FFT (\emph{cf.}, Section \ref{sec:full_alg}), representing previously-estimated components that are subtracted off to produce a residual.

The required guarantees on \textsc{LocateReducedSignals} are given in Lemma \ref{lem:loc_k} (and more formally in Appendix~\ref{app:logk0}), and in order to prove our main result on \textsc{MultiBlockLocate}, we also need the following lemma ensuring that we can compute energy estimates satisfying the preconditions of Lemmas \ref{lemm:11} and \ref{lemm:12}; the procedure and proof are presented in Section \ref{sec:est_downsampled}. 

\begin{lem} \label{lem:gammaproperties}
    \emph{(\textsc{EstimateEnergies} guarantees)}
    Given $(n,k_0,k_1)$, the signals $X \in \CC^n$ and $\wh{\chi} \in \CC^n$ with $\|\wh{X} - \wh{\chi}\|_2^2 \ge \frac{1}{\poly(n)} \|\wh{\chi}\|_2$, and the parameter $\delta \in \big(\frac{1}{n},\frac{1}{20}\big)$, the procedure \textsc{EstimateEnergies}$(X,\wh{\chi},n,k_0,k_1$,$\delta)$ returns a vector $\gamma \in \RR^{2k_1}$ such that, for any given set $S^*$ of cardinality at most $10k_0$, we have the following with probability at least $\frac{1}{2}$:
    	\begin{enumerate}
    		\item $\sum_{r \in [2k_1]} \Big| \|\wh{Z}^r_{S^*}\|_2^2 - \gamma^r \Big|_+ \le  40 \delta \sum_{r \in [2k_1]}  \|\wh{Z}^r\|_2^2$;
    		\item $\|\gamma\|_1 \le 10 \sum_{r \in [2k_1]} \|\wh{Z}^r\|_2^2$;
    	\end{enumerate}
    	where $\{Z^r\}_{r \in [2k_1]}$ is the $(k_1,\delta)$-downsampling of $X - \chi$ (see Definition \ref{def:downsampling}).
    	
    	Moreover, if $\wh{\chi}$ is $(O(k_0),k_1)$-block sparse, then the sample complexity is $O (\frac{k_0 k_1}{\delta^2} \log \frac{1}{\delta} \log \frac{1}{\delta p})$, and the runtime is $O (\frac{k_0 k_1}{\delta^2} \log^2 \frac{1}{\delta} \log^2 n)$.
\end{lem}

\begin{remark}
    The preceding lemma ensures that the $\gamma^r$ provide good approximations of $\|\wh{Z}^r\|_2^2$ in a ``restricted'' and ``one-sided'' sense, while not over-estimating the total energy by more than a constant factor.  Specifically, the first part concerns the energy of $\wh{Z}^r$ restricted to a fixed set of size $O(k_0)$, and characterizes the extent to which the energies are under-estimated.  It appears to be infeasible to characterize over-estimation in the same way (e.g., replacing $|\cdot|_+$ by $|\cdot|$), since several of the samples could be overly large due to spiky noise. 
\end{remark}

\begin{remark}
    Here and subsequently, the $\poly(n)$ lower bounds regarding $(\wh{X},\wh{\chi})$ are purely technical, resulting from extremely small errors when subtracting off $\wh{\chi}$.  See Section \ref{sec:semi_equi} for further details.
\end{remark}

\begin{algorithm}
\caption{Multi-block sparse location}
\begin{algorithmic}[1]

\Procedure{MultiBlockLocate}{$X, \wh{\chi}, n, k_0, k_1, \delta , p$}

% \State $a_r \gets \frac{nr}{2k_1}$ for each $r \in [2k_1]$
\For{\texttt{$t \in \{1,\dotsc, 10\log \frac{1}{p}\}$}}
\State $\gamma \gets \textsc{EstimateEnergies}(X,\wh{\chi},n,k_0,k_1,\delta)$ \Comment See Section \ref{sec:est_downsampled}
\State $ \sv^{(t)} \gets \textsc{BudgetAllocation} (\gamma, k_0, k_1, \delta, \frac{1}{2}\delta p)$ \Comment $\gamma = (\gamma^1,\dotsc,\gamma^{2k_1})$ \label{line:call_budget}
\EndFor
\State $\sv \gets \max_{t} \sv^{(t)}$ (element-wise with respect to $r \in [2k_1]$)
\State $L \gets \textsc{LocateReducedSignals}(X, \wh{\chi}, n, k_0, k_1, \textbf{s}, \delta, \frac{1}{2}\delta p)$ \label{line:locate}  \Comment See Appendix \ref{sec:loc_k}
% \State $\{Z_X^r\}_{r \in [2k_1]} \gets (k_1,\delta)$-downsampling of $X$
% \State $\{Z_{\chi}^r\}_{r \in [2k_1]} \gets (k_1,\delta)$-downsampling of $\chi$
%\State $L \gets \emptyset$
%\For{\texttt{$r \in [2k_1]$}}
%\State $L = L \, \cup \textsc{ LocateSignal} (\wh{Z}_X^r, \wh{Z}_{\chi}^r, \textbf{s}_r, \frac{1}{2}\delta p)$ \label{line:locate}  \Comment See Appendix \ref{sec:loc_k}
%\EndFor

\State \textbf{return} $L$
\EndProcedure

\end{algorithmic}
\label{alg:Locate}

\end{algorithm}

\noindent We are now in a position to provide our guarantees on \textsc{MultiBlockLocate}, namely, on the behavior of the list size, and on the energy that the components in the list capture.  Note that the output of \textsc{MultiBlockLocate} is random, since the same is true of \textsc{EstimateEnergies}, \textsc{BudgetAllocation}, and \textsc{LocateReducedSignals}. 

\begin{lem} \label{lem:multi_locate}
    \emph{(\textsc{MultiBlockLocate} guarantees)}
	Given $(n,k_0,k_1)$, the parameters $\delta \in \big(\frac{1}{n},\frac{1}{20}\big)$ and $p \in \big(\frac{1}{n},\frac{1}{2}\big)$, and the signals $X \in \CC^n$ and $\chi \in \CC^n$ with $\wh{\chi}_0$ uniformly distributed over an arbitrarily length-$\frac{\|\wh{\chi}\|^2}{\poly(n)}$ interval,  the output $L$ of the function \textsc{MultiBlockLocate}$(X,\wh{X},k_1,k_0,n, \delta , p)$ has the following properties for any set $S^*$ of cardinality at most $10k_0$:
\begin{enumerate}
\item $\EE \big[ |L| \big] = O \big( \frac{k_0}{\delta} \log \frac{k_0}{\delta}  \log\frac{1}{ p} \log^2 \frac{1}{\delta p} \big) $; % with respect to the randomness of the procedures \textsc{EstimateEnergies}, \textsc{BudgetAllocation}, and \textsc{LocateReducedSignals}.
\item $\sum_{j \in S^*\backslash L} \| (\wh{X}-\wh{\chi})_{I_j} \|_2^2 \leq 200\sqrt{\delta} \|\wh{X}-\wh{\chi}\|_2^2$ with probability at least $1 - p$.
\end{enumerate}
Moreover, if $\wh{\chi}$ is $(O(k_0),k_1)$-block sparse, and we have $\delta = \Omega\big(\frac{1}{\poly(\log n)}\big)$ and $p = \Omega\big(\frac{1}{\poly(\log n)}\big)$, then (i) the expected sample complexity is $O^*\big(\frac{k_0}{\delta} \log (1+k_0) \log n + \frac{k_0 k_1}{\delta^2} \big)$, and the expected runtime is $O^*\big( \frac{k_0}{\delta} \log (1+k_0) \log^2 n + \frac{k_0 k_1}{\delta^2} \log^2 n + \frac{k_0 k_1}{\delta} \log^3 n \big)$; (ii) if the procedure returns $L$, then we are guaranteed that the algorithm used $O^*\big(|L| \cdot \log n + \frac{k_0 k_1}{\delta^2} \big)$ samples and $O^*\big( |L|\cdot \log^2 n + \frac{k_0 k_1}{\delta^2} \log^2 n + \frac{k_0 k_1}{\delta} \log^3 n \big)$ runtime.

%
%\begin{itemize}
%\item the expected sample complexity is $O(k_0\log(1+k_0)\log n + k_0k_1)$;
%\item  the expected runtime is  $O(k_0\log(1+k_0)\log^2 n + k_0k_1 \log^3 n)$.
%\end{itemize}
%
%More generally, if $\wh{\chi}$ is $(O(k_0),k_1)$-block sparse, the expected sample complexity is  
%$
%O\big(\frac{k_0}{\delta} \log \frac{k_0}{\delta}  \log\frac{1}{p}\log\frac{1}{\delta} \log^3 \frac{1}{\delta p} \log n + \frac{k_0 k_1}{\delta^2} \log \frac{1}{\delta} \log \frac{1}{p}\big),
%$
%the expected runtime is 
%$
%O\big( \frac{k_0}{\delta} \log \frac{k_0}{\delta}  \log\frac{1}{ p} \log\frac{1}{\delta} \log^3 \frac{1}{\delta p}\cdot\log^2 n + \frac{k_0 k_1}{\delta^2} \log\frac{1}{\delta} \log \frac{1}{\delta p} \log^2 n + \frac{k_0 k_1}{\delta}\log \frac{1}{\delta p} \log^3 n \big),
%$
%and if the procedure returns $L$, then we are guaranteed that the algorithm used 
%$
%O\big(|L| \cdot  \log\frac{1}{\delta}\cdot\log\frac{1}{\delta p} \log n + \frac{k_0 k_1}{\delta^2} \log \frac{1}{\delta} \log \frac{1}{p}\big)
%$
%samples and 
%$
%O\big( |L|\cdot\log\frac{1}{\delta}\log\frac{1}{\delta p}\cdot\log^2 n + \frac{k_0 k_1}{\delta^2} \log\frac{1}{\delta} \log \frac{1}{\delta p} \log^2 n + \frac{k_0 k_1}{\delta}\log \frac{1}{\delta p} \log^3 n \big)
%$ runtime.

\end{lem}

\begin{rem}
    The procedure \textsc{MultiBlockLocate} is oblivious to the choice of $S^*$ in this lemma statement.
\end{rem}

\begin{proof}
    \textbf{First claim:} Note that in each iteration of the outer loop when we run \textsc{BudgetAllocation}$(\gamma, k_0, k_1, \delta, \frac{1}{2}\delta p)$, Lemma \ref{lemm:10} implies that for any $t$, the following holds true:
        $$\EE \Big[ \sum_{r \in [2k_1]} \textbf{s}^{(t)}_r \Big] = O\Big(\frac{k_0}{\delta} \log \frac{k_0}{\delta} \log \frac{1}{\delta p} \Big),$$
    where $\textbf{s}^{(t)}_r$ is the $r$-th entry of the budget allocation vector $\mathbf{s}^t$ at iteration $t$. Therefore,
        \begin{equation}
        \EE \Big[ \sum_{r \in [2k_1]} \textbf{s}_r \Big] = \EE \Big[ \sum_{r \in [2k_1]} \max_{t=1,\dotsc,10\log\frac{1}{p}} \textbf{s}_r^t \Big]  \le \sum_{t =1}^{10 \log \frac{1}{p}} \EE \Big[ \sum_{r \in [2k_1]}\sv_r^{(t)} \Big] = O\Big(k_0 \log \frac{k_0}{\delta} \log\frac{1}{p} \log\frac{1}{\delta p}\Big). \label{eq:exp_s_r}
        \end{equation}
    We now apply Lemma \ref{lem:loc_k}, which is formalized in Appendix \ref{sec:loc_k}; the assumption $\max_{r\in[2k_1]} \textbf{s}_r = O\big(\frac{k_0}{\delta}\big)$ therein is satisfied due to the range of $q$ from which we sample in \textsc{BudgetAllocation}.  We set the target success probability to $1 - \frac{1}{2}\delta p$, which guarantees that the size of the list returned by the function \textsc{LocateReducedSignals} is $O\big(\sum_{r \in [2k_1]} \textbf{s}_r \log \frac{1}{\delta p}\big)$.  Therefore, by \eqref{eq:exp_s_r}, we have
        $$\EE [ |L| ] = O\Big(\frac{k_0}{\delta} \log \frac{k_0}{\delta}  \log\frac{1}{ p} \log^2 \frac{1}{\delta p}\Big),$$
    yielding the first statement of the lemma.  
    
    % Note that the formal statement of \textsc{LocateSignal} in Appendix \ref{sec:loc_k} requires that the $(k_1,\delta)$-downsamplings $\{\wh{Z}_X^r\}_{r \in [2k_1]}$ and $\{\wh{Z}_{\chi}^r\}_{r \in [2k_1]}$  of $X$ and $\chi$ satisfy $\|\wh{Z}_X^r - \wh{Z}_{\chi}^r\|_2 \ge \frac{1}{\poly(n)} \|\wh{Z}_{\chi}^r\|_2$.  By Definition \ref{def:downsampling}, this is implied by $|\wh{X}_0 - \wh{\chi}_0|_2 \ge \frac{1}{\poly(n)}\|\wh{\chi}\|_2^2$.  \textbf{[TODO: Check/expand as needed]}
    
    \textbf{Second claim:} Let $X'=X-\chi$, and consider the set $S^*$ given in the lemma statement, and an arbitrary iteration $t$. By Lemma \ref{lem:gammaproperties} in Section \ref{sec:est_downsampled} (also stated above), the approximate energy vector $\gamma$ in any given iteration of the outer loop satisfies
    \begin{gather}
    	\sum_{r \in [2k_1]} \Big| \|\wh{Z}^r_{ {S}^*}\|_2^2 - \gamma^r \Big|_+ \le  40\delta \sum_{r \in [2k_1]}  \|\wh{Z}^r\|_2^2 \nonumber \\
    	\|\gamma\|_1 \le 10 \sum_{r \in [2k_1]} \|\wh{Z}^r\|_2^2 \label{eq:gamma_conds}
    \end{gather}
    with probability at least $\frac{1}{2}$. When this is the case, the vector $\gamma$ meets the requirements of Lemmas \ref{lemm:11} and \ref{lemm:12}. That means that the probability of having an energy estimate $\gamma$ that meets these requirements in at least one iteration is lower bounded by $1-(\frac{1}{2})^{10 \log \frac{1}{ p}} \ge 1-p$.
    
    We now consider an arbitrary iteration in which the above conditions on $\gamma$ are satisfied. We write
    \begin{align}
    	\sum_{j \in S^* \backslash L} \| \wh{X}'_{I_j} \|_2^2
    	= \sum_{j \in (S^* \cap \tilde{S}) \backslash L} \| \wh{X}'_{I_j} \|_2^2 + \sum_{j \in S^* \backslash (\tilde{S} \cup L)} \| \wh{X}'_{I_j} \|_2^2. \label{eq:mb_two_terms}
    \end{align}
    The second term is bounded by
    \begin{equation}
    \sum_{j \in S^* \backslash (\tilde{S} \cup L)} \| \wh{X}'_{I_j} \|_2^2
        \le \sum_{j \in S^* \backslash \tilde{S}} \| \wh{X}'_{I_j} \|_2^2
        \le 100\sqrt{\delta} \|\wh{X}'\|_2^2 \label{eq:120}
    \end{equation}
    by Lemma \ref{lemm:11}, which uses the first condition on $\gamma$ in \eqref{eq:gamma_conds}. 
    
    We continue by calculating the expected value of the first term in \eqref{eq:mb_two_terms} with respect to the randomness of \textsc{BudgetAllocation} and \textsc{LocateReducedSignals}:
    \begin{align}
        \EE \Big[ \sum_{j \in (S^* \cap \tilde{S}) \backslash L} \| \wh{X}'_{I_j} \|_2^2 \Big]
        &= \EE \Big[ \sum_{j \in (S^* \cap \tilde{S})} \| \wh{X}'_{I_j} \|_2^2 \Ic\big[ j \notin L \big] \Big] \nn \\
        &\le \sum_{j \in \tilde{S}} \| \wh{X}'_{I_j} \|_2^2 \cdot \PP \big[ j \notin L \big]. \label{eq:exp_S_star_tilde}
    \end{align}
    We thus consider the probability $\PP [ j \notin L ]$ for an arbitrary $j \in \tilde{S}$. If $j \in \tilde{S}$, then by Lemma \ref{lemm:12} and the choice of the final parameter of $\frac{1}{2}\delta p$ passed to \textsc{BudgetAllocation}, there is at least one $r \in [2k_1]$ such that $j$ is covered, with probability at least $1- \frac{1}{2}\delta p$. We also know from Lemma \ref{lem:loc_k} that the failure probability of \textsc{LocateReducedSignals} for some covered $j$ is at most $\frac{1}{2}\delta p$. A union bound on these two events gives
    $$\PP \big[ j \notin L\big] \le \delta p, \quad \forall j \in    \tilde{S}.$$
    Hence, we deduce from \eqref{eq:exp_S_star_tilde} that
    \begin{equation}
    \EE \Big[ \sum_{j \in (S^* \cap \tilde{S}) \backslash L} \| \wh{X}'_{I_j} \|_2^2 \Big] \le \delta p \cdot \|\wh{X}'\|_2^2, \nn
    \end{equation}
    and Markov's inequality gives
    \begin{equation}
    \sum_{j \in (S^* \cap \tilde{S}) \backslash L} \| \wh{X}'_{I_j} \|_2^2 \le \delta \cdot \|\wh{X}'\|_2^2 \nn
    \end{equation}
    with probability at least $1- p$. Combining this with \eqref{eq:mb_two_terms}--\eqref{eq:120}, and using the assumption $\delta \le \frac{1}{20}$ to write $\delta \le 100\sqrt{\delta}$, we complete the proof.
    
    \textbf{Sample complexity and runtime:} We first consider the sample complexity and runtime as a function of the output $L$. 
    
    There are two operations that cost us samples.  The first is the call to \textsc{EstimateEnergies}, which costs $O (\frac{k_0 k_1}{\delta^2} \log^2 \frac{1}{\delta} )$ by Lemma \ref{lem:gammaproperties}.  The second is the call to \textsc{LocateReducedSignals}; by Lemma \ref{lem:loc_k} in Appendix \ref{sec:loc_k}, with $\delta p$ in place of $p$, this costs $O\big(\sum_{r \in [2k_1]}s^r \log\frac{1}{\delta p} \log\frac{1}{\delta} \log n\big)$ samples (recall that $\wh{\chi}$ is $(O(k_0),k_1)$-block sparse by assumption), which is $O\big(|L| \log \frac{1}{\delta p}\log\frac{1}{\delta} \log n\big)$.  Adding these contributions gives the desired result; the $\log\frac{1}{p}$ and $\log\frac{1}{\delta}$ factors are hidden in the $O^*(\cdot)$ notation, since we have assumed that $\delta$ and $p$ behave as $\Omega\big(\frac{1}{\poly \log n}\big)$.
    
    The time complexity follows by the a similar argument, with \textsc{EstimateEnergies} costing $O (\frac{k_0 k_1}{\delta^2} \log^2 \frac{1}{\delta} \log^2 n)$ by Lemma \ref{lem:gammaproperties}, and the call to \textsc{LocateReducedSignals} costing $O\big( |L|  \log\frac{1}{\delta p} \log\frac{1}{\delta} \log^2 n + \frac{k_0k_1}{\delta}\log\frac{1}{\delta p} \log^3 n\big)$ by Lemma \ref{lem:loc_k} in Appendix \ref{sec:loc_k}.  The complexity of \textsc{EstimateEnergies} dominates that of calling \textsc{BudgetAllocation}, which is $O\big(k_1 + \frac{k_0}{\delta} \log \frac{1}{p}\big)$ by Lemma \ref{lemm:10}.
    
    The \emph{expected} sample complexity and runtime follow directly from those depending on $L$, by simply substituting the expectation of $|L|$ given in the lemma statement.
\end{proof}

%!TEX root = BlockSparseFT-new.tex

\newcommand{\fc}{F}
\newcommand{\nsq}{[n]}
\newcommand{\round}{\mathrm{round}}
\newcommand{\rect}{\mathrm{rect}}

\section{Energy Estimation} \label{sec:prelim}

In this section, we provide the energy estimation procedure used in the \textsc{MultiBlockLocate} procedure in Algorithm \ref{alg:Locate}, and prove its guarantees that were used in the proof of Lemma \ref{lem:multi_locate}.  To do this, we introduce a variety of tools needed, including hashing and the semi-equispaced FFT.  While such techniques are well-established for the standard sparsity setting \cite{IKP}, applying the existing semi-equispaced FFT algorithms {\em separately} for each $Z^r$ in our setting would lead to a runtime of $k_0 k_1^2 \mathrm{poly}(\log n)$.  Our techniques allow us to compute the required FFT values for \emph{all} $r$ in $k_0 k_1 \mathrm{poly}(\log n)$ time, as we detail in Section \ref{sec:semi_equi}.

\subsection{Hashing Techniques} \label{sec:hashing}

The notion of hashing plays a central role in our estimation primitives, and in turn makes use of random permutations.

\begin{defn}[Approximately pairwise-independent permutation] \label{def:permutation}
    Fix $n$, and let $\pi \,:\, [n] \to [n]$ be a random permutation.  We say that $\pi$ is \emph{approximately pairwise-independent} if, for any $i,i' \in [n]$ and any integer $t$, we have $\PP[ |\pi(i) - \pi(i')| \le t ] \le \frac{4 t}{n}$.
\end{defn}

It is well known that such permutations exist in the form of a simple modulo-$n$ multiplication; we will specifically use the following lemma from \cite{IK14a}.

\begin{lem} \emph{(Choice of permutation \cite[Lemma 3.2]{IK14a})} \label{lem:perm_interval}
    Let $n$ be a power of two, and define $\pi(i) = \sigma \cdot i$, where $\sigma$ is chosen uniformly at random from the odd numbers in $[n]$.  Then $\pi$ is an approximately pairwise-independent random permutation.
\end{lem}

We now turn to the notion of \emph{hashing} a signal into buckets.  We do this by applying the random permutation from Lemma \ref{lem:perm_interval} along with a random shift in time domain, and then applying a suitable filter according to Definition \ref{def:filterG}.

\begin{defn}[Hashing] \label{def:hashing}
    Given integers $(n,B)$, parameters $\sigma,\Delta \in [n]$, and the signals $X \in \CC^n$ and $G \in \CC^n$, we say that $U \in \CC^B$ is an \emph{$(n,B,G,\sigma,\Delta)$-hashing} of $X$ if 
    \begin{equation}
        U_b = \frac{B}{n} \sum_{i\in [\frac{n}{B}]} X_{\sigma( \Delta + j + B \cdot i)} G_{j + B \cdot i}, \quad j \in [B].  \label{eq:hashing}
    \end{equation} 
    Moreover, we define the following quantities: 
	\begin{itemize}
    	\item $\pi(j) = \sigma\cdot j$, representing the approximately pairwise random permutation; 
    	\item $h(j) = \round\big( j\frac{B}{n} \big)$, representing  the bucket in $[B]$ into which a frequency $j$ hashes; 
    	\item $o_j(j') = \pi(j') - h(j)\frac{n}{B}$, representing the offset associated with two frequencies $(j,j')$.
	\end{itemize}
\end{defn}

\noindent With these definitions, we have the following lemma, proved in Appendix \ref{sec:pf_uhat}.  Note that here we write the exact Fourier transform of $U$ as $\wh{U}^*$, since later we will use $\wh{U}$ for its {\em near-exact} counterpart to simplify notation.

\begin{lem} \label{lem:uhat}
    \emph{(Fourier transform of hashed signal)}
    Fix $(n,B)$ and the signals $X \in \CC^n$ and $G \in \CC^n$ with the latter symmetric about zero. If $U$ is an $(n,B,G,\sigma,\Delta)$-hashing of $X$, then its exact Fourier transform $\wh{U}^*$ is given by
 	$$\wh{U}^*_b = \sum_{f \in [n]} \wh{X}_{f} \wh{G}_{\sigma f-b\frac{n}{B}} \omega_n^{ \sigma\Delta f }, \quad b \in [B]. $$
\end{lem}

We conclude this subsection by stating the following technical lemma regarding approximately pairwise independent permutations and flat filters.

\begin{lem} \label{lem:perm_property}
    \emph{(Additional filter property)}
    Fix $n$, and let $G$ be an $(n,B,F)$-flat filter.  Let $\pi(\cdot)$ be an approximately pairwise-independent random permutation (\emph{cf.}, Definition \ref{def:permutation}), and for $f,f' \in [n]$, define $o_f(f') = \pi(f') - \frac{n}{B}\round\big(\pi(f)\frac{B}{n}\big)$. Then for any $x \in \CC^n$ and $f \in [n]$, we have
    \begin{equation}
        \sum_{f' \ne f} | \wh{X}_{f'} |^2 \EE_{\pi}\big[ |\wh{G}_{o_f(f')}|^2 \big] \le \frac{10}{B} \|\wh{X}\|^2. \label{eq:filter_exp}
   \end{equation}
\end{lem}

\noindent The proof is given in Appendix \ref{sec:pf_perm}.

\subsection{Semi-Equispaced {FFT}} \label{sec:semi_equi}

\begin{algorithm}
\caption{Semi-equispaced inverse FFT for approximating the inverse Fourier transform, with standard sparsity (top) and block sparsity (bottom)} \label{alg:semi_equi}
    \begin{algorithmic}[1]
    
    \Procedure{SemiEquiInverseFFT}{$\wh{X},n,k,\zeta$}
    
    \State $\wh{G} \gets \textsc{Filter}(n,k,\zeta)$ \Comment See \cite[Sec.~12]{IKP}; same as proof of Lemma \ref{lem:semi_equi}
    \State $\wh{Y}_i \gets (\wh{X} \star \wh{G})_{\frac{in}{2k}}$ for each $i \in [2k]$
    \State ${Y} \gets \textsc{InverseFFT}(\wh{Y})$
    \State {\bf return} $\{ {Y}_j \}_{ |j| \le \frac{k}{2} }$
    \EndProcedure
    
    \Procedure{SemiEquiInverseBlockFFT}{$\wh{X},n,k_0,k_1,c$}
    
    \State $\wh{G} \gets \textsc{Filter}(n,k_1,n^{-c})$ \Comment See proof of Lemma \ref{lem:semi_equi}
    \For{ $j \in \big[ \frac{2n}{k_1} \big]$ such that $(\wh{X} \star \wh{G})_{\frac{k_1}{2} j}$ may be non-zero ($O(ck_0 \log n)$ in total) }
        \State $\widetilde{Y}_j^b \gets \frac{k_1}{2}\sum_{l=1}^{\frac{n}{2k_1}} \wh{X}_{b+2k_1 l} \wh{G}_{\frac{k_1}{2} j - (b+2k_1 l) }$ for each $b \in [2k_1]$
        \State $(\wh{Y}_j^1,\dotsc,\wh{Y}_j^{2k_1}) \gets \textsc{{InverseFFT}}( \widetilde{Y}_j^1, \dotsc, \widetilde{Y}_j^{2k_1}  )$ 
    \EndFor
    
    \For{ $r \in [2k_1]$ }
        \State $\wh{Y}^{r} \gets (\wh{Y}^r_1,\dotsc,\wh{Y}^r_{n/k_1})$
        \State ${Y}^{r} \gets \textsc{SemiEquiInverseFFT}(\wh{Y}^r, \frac{n}{k_1}, k_0, n^{-(c+1)}) $
    \EndFor
    
    \State {\bf return} $\{ {Y}_j^{r} \}_{ r \in [2k_1], |j| \le \frac{k_0}{2} }$
    \EndProcedure
    
    \end{algorithmic}

\end{algorithm}

One of the steps of our algorithm is to take the inverse Fourier transform of our current estimate of the spectrum, so that it can be subtracted off and we can work with the residual.  The \emph{semi-equispaced inverse FFT} provides an efficient method for doing this, and is based on the application of the standard inverse FFT to a filtered and downsampled signal. 

We start by describing an existing technique of this type for \emph{standard} sparsity; the details are shown in the procedure \textsc{SemiEquiInverseFFT} in Algorithm \ref{alg:semi_equi}, and the resulting guarantee from \cite[Sec.~12]{IKP} is stated as follows.\footnote{Note that the roles of time and frequency are reversed here compared to \cite{IKP}.}

\begin{lem} \emph{(\textsc{SemiEquiInverseFFT} guarantees \cite[Lemma 12.1, Cor.~12.2]{IKP})} \label{lem:semi_equi_std}
    (i) Fix $n$ and a parameter $\zeta > 0$.  If $\wh{X} \in \CC^n$ is $k$-sparse for some $k$, then \textsc{SemiEquiInverseFFT}($\wh{X},n,k,\zeta$) returns a set of values $\{ {Y}_j\}_{|j| \le k/2}$ in time $O( k \log \frac{n}{\zeta})$, satisfying
    $$ |{Y}_j - {X}_j| \le \zeta \|X\|_2. $$
    
    (ii) Given two additional parameters $\sigma,\Delta \in [n]$ with $\sigma$ being odd, it is possible to compute a set of values $\{ {Y}_j\}$ for all $j$ equaling $\sigma j' + \Delta$ for some $j'$ with $|j'| \le k/2$, with the same runtime and approximation guarantee.
\end{lem}

For the block-sparse setting, we need to adapt the techniques of \cite{IKP}, making use of a \emph{two-level} scheme that calls \textsc{SemiEquiInverseFFT}.  The resulting procedure, \textsc{SemiEquiInverseBlockFFT}, is described in Algorithm \ref{alg:semi_equi}.  The main result of the procedure is the following analog of Lemma \ref{lem:semi_equi_std}.

\begin{lem} \label{lem:semi_equi}
    \emph{(\textsc{SemiEquiInverseBlockFFT} guarantees)}
    (i) Fix $(n,k_0,k_1)$, a $(k_0,k_1)$-block sparse signal $\wh{X} \in \CC^n$, and a constant $c \ge 1$.  Define the shifted signals $\{ {X}^r\}_{r \in [2k_1]}$ with $ {X}^r_i = {X}_{i + \frac{nr}{2k_1}}$.  The procedure \textsc{SemiEquiInverseBlockFFT}($\wh{X},n,k_0,k_1,c$) returns a set of values ${Y}^r_j$ for all $r \in [2k_1]$ and $|j| \le \frac{k_0}{2}$ in time $O(c^2 k_0 k_1 \log^2 n)$, satisfying
    \begin{equation}
        |{Y}^{r}_j - {X}^r_j| \le 2n^{-c} \|X\|_2. \label{eq:semi_equi}
    \end{equation}
    
    (ii) Given two additional parameters $\sigma,\Delta \in \big[\frac{n}{k_1}\big]$ with $\sigma$ odd, it is possible to compute a set of values ${Y}^{r}_j$ for all $r \in [2k_1]$ and $j$ equaling $\sigma j' + \Delta$ (modulo $\frac{n}{k_1}$) for some $|j'| \le \frac{k_0}{2}$, with the same runtime and approximation guarantee.
\end{lem}

\noindent The proof is given in Appendix \ref{sec:pf_semi}.

\begin{rem}
    When applying the preceding lemmas, the signal sparsity and the number of values we wish to estimate will not always be identical.  However, this can immediately be resolved by letting the parameter $k$ or $k_0$ therein equal the maximum of the two.
\end{rem}

\subsection{Combining the Tools} \label{sec:hash2bins}

In Algorithm \ref{alg:hash2bins}, we describe two procedures combining the above tools.  The first, \textsc{HashToBins}, accepts the signal $\wh{X}$ and its current estimate $\wh{\chi}$, uses \textsc{SemiEquiInverseFFT} to approximate the relevant entries of $\chi$, and computes a hashing of $X - \chi$ as per Definition \ref{def:hashing}.  The second, \textsc{HashToBinsReduced}, is analogous, but instead accepts a $(k_1,\delta)$-downsampling of $X$, and uses \textsc{SemiEquiInverseBlockFFT}.  It will prove useful to allow the function to hash into a different number of buckets for differing $r$ values, and hence accept $\{G^r\}_{r\in[2k_1]}$ and $\{B^r\}_{r\in[2k_1]}$ as inputs.  For simplicity, Algorithm \ref{alg:hash2bins} states the procedures without precisely giving the parameters passed to the semi-equispaced FFT, but the details are given in the proof of the following.

\begin{algorithm}
\caption{Hash to bins functions for original signal (top) and reduced signals (bottom)} \label{alg:hash2bins}
    \begin{algorithmic}[1]
    
    \Procedure{HashToBins}{$X,\wh{\chi},G,n,B,\sigma,\Delta$}

    \State Compute $\{\chi_i\}$ using $\textsc{SemiEquiInverseFFT}$ with input $(\wh{\chi},n,O(FB),n^{-c'})$ \\ 
        \Comment See Lemma \ref{lem:semi_equi_std}; $F$ equals the parameter of filter $G$, and $c'$ is a large constant
    \State $U_X \gets (n,B,G,\sigma,\Delta)$-hashing of $X$ \Comment See Definition \ref{def:hashing}
    \State $U_{\chi} \gets (n,B,G,\sigma,\Delta)$-hashing of $\chi$
    \State $\wh{U} \gets \text{ FFT of }{U}_X - {U}_{\chi}$
    \State \textbf{return} $\wh{U}$
    
    \EndProcedure
    
    \Procedure{HashToBinsReduced}{$\{Z_X^r\}_{r\in[2k_1]},\wh{\chi},\{G^r\}_{r\in[2k_1]},n,k_1,\{B^r\}_{r\in[2k_1]},\sigma,\Delta$}
    
    \State $\Bmax \gets \max_{r \in [2k_1]} B^r$
    \State $k_0 \gets$ minimal value such that $\wh{\chi}$ is $(k_0,k_1)$-block sparse
    \State Compute $\{\chi_i\}$ using $\textsc{SemiEquiInverseBlockFFT}$;  input $(\wh{\chi},n,O(\Fmax\Bmax + k_0),k_1,c')$  \\
         \Comment See Lemma \ref{lem:semi_equi}; $\Fmax$ equals the maximal parameter of the filters $\{G^r\}$, and $c'$ is a large constant
    \State $\{Z_{\chi}^r\}_{r \in [2k_1]} \gets (k_1,\delta)$-downsampling of $\chi$ \Comment See Definition \ref{def:downsampling}
        
    \For{ $r \in [2k_1]$ }
        \State ${U}_X^r \gets \big(\frac{n}{k_1}, B^r, G^r,\sigma,\Delta)$-hashing of ${Z}_X^r$ \Comment See Definition \ref{def:hashing}
        \State ${U}_{\chi}^r \gets \big(\frac{n}{k_1}, B^r, G^r,\sigma,\Delta)$-hashing of ${Z}_{\chi}^r$ \label{line:hash}
        \State $\wh{U}^r \gets \text{ FFT of }{U}^r_X - {U}^r_{\chi}$
    \EndFor
    
    \State {\bf return} $\{ \wh{U}^r \}_{r \in [2k_1]}$
    \EndProcedure
    
    \end{algorithmic}

\end{algorithm}

\begin{lem} \label{lem:hash2bins}
    \emph{(\textsc{HashToBins} and \textsc{HashToBinsReduced} guarantees)}
    (i) Fix $(n,k,B,F)$, an $(n,B,F)$-flat filter $G$ supported on an interval of length $O(FB)$, a signal $X \in \CC^n$, a $k$-sparse signal $\wh{\chi}$  For any $(\sigma,\Delta)$, the procedure $\textsc{HashToBins}(X,\wh{\chi},G,n,B,\sigma,\Delta)$ returns a sequence $\wh{U}$ such that 
        $$ \| \wh{U} - \wh{U}^* \|_{\infty} \le n^{-c} \|\wh{\chi}\|_2, $$
    where $\wh{U}^*$ is the exact Fourier transform of the $(n,B,G,\sigma,\Delta)$-hashing of $X - \chi$ (see Definition \ref{def:hashing}), and $c = c' + O(1)$ for $c'$ in Algorithm \ref{alg:hash2bins}.  Moreover, the sample complexity is $O(FB)$, and the runtime is $O(c F (B+k) \log n)$.
    
    (ii) Fix $(n,k_0,k_1)$ and the parameters $(\{B^r\}_{r\in[2k_1]},F,\delta)$.  For each $r \in [2k_1]$, fix an $\big(\frac{n}{k_1},B^r,F\big)$-flat filter $G^r$ supported on an interval of length $O(FB^r)$. Moreover, fix a signal $X \in \CC^n$ and its $(k_1,\delta)$-downsampling $\{Z^r\}_{r \in [2k_1]}$ with $\delta \in \big(\frac{1}{n},\frac{1}{20}\big)$, and a $(k_0,k_1)$-block sparse signal $\wh{\chi}$.  For any $(\sigma,\Delta)$, the procedure $\textsc{HashToBinsReduced}(\{Z^r\}_{r \in [2k_1]},\wh{\chi},\{G^r\}_{r\in[2k_1},n,k_1,\{B^r\}_{r\in[2k_1},\sigma,\Delta)$ returns a set of sequences $\{\wh{U}^r\}_{r\in[2k_1]}$ such that 
            $$ \| \wh{U}^r - \wh{U}^{*r} \|_{\infty} \le n^{-c} \|\wh{\chi}\|_2, \quad r \in [2k_1], $$
    where $\wh{U}^{*r}$ is the exact Fourier transform of the $\big(\frac{n}{k_1},B^r,G^r,\sigma,\Delta\big)$-hashing for the $(k_1,\delta)$-downsampling of $X - \chi$, and $c = c' + O(1)$ for $c'$ in Algorithm \ref{alg:hash2bins}. Moreover, the sample complexity is $O\big(F \sum_{r\in[2k_1]}B^r \log \frac{1}{\delta}\big)$, and the runtime is $O\big(c^2 (\Bmax F + k_0) k_1 \log^2 n )$ with $\Bmax = \max_{r \in [2k_1]} B^r$.
\end{lem}

\noindent The proof is given in Appendix \ref{sec:pf_hash2bins}.

\begin{rem} \label{rem:k0}
    Throughout the paper, we consider $c$ in Lemma \ref{lem:hash2bins} to be a large absolute constant.  Specifically, various results make assumptions such as $\|\wh{X} - \wh{\chi}\|_2 \ge \frac{1}{\poly(n)} \|\wh{\chi}\|_2$, and the results hold true when $c$ is sufficiently large compared to implied exponent in the $\poly(n)$ notation.  Essentially, the $n^{-c}$ error term is so small that it can be thought of as zero, but we nevertheless handle it explicitly for completeness.
\end{rem}

\subsection{Estimating the Downsampled Signal Energies} \label{sec:est_downsampled}

We now come to the main task of this section, namely, approximating the energy of each $\wh{Z}^r$.  To do this, we hash into $B = \frac{4}{\delta^2} \cdot k_0$ buckets (\emph{cf.}, Definition \ref{def:hashing}), and form the estimate as the energy of the hashed signal.  The procedure is shown in Algorithm \ref{alg:est_energy}. 

Before stating the guarantees of Algorithm \ref{alg:est_energy}, we provide the following lemma characterizing the approximation quality for an \emph{exact} hashing of a signal, as opposed to the approximation returned by \textsc{HashToBinsReduced}.  Intuitively, the first part states that we can accurately estimate the top coefficients well without necessarily capturing the noise, and the second part states that, in expectation, we do not over-estimate the total signal energy by more than a small constant factor.

\begin{algorithm}
\caption{Procedure for estimating energies of downsampled signals} \label{alg:est_energy}
    \begin{algorithmic}[1]
    
    \Procedure{EstimateEnergies}{$X,\wh{\chi},n,k_0,k_1$,$\delta$}
    
    \State $B \gets \frac{4}{\delta^2} \cdot k_0$
    \State $F \gets 10 \log \frac{1}{\delta}$
    \State $H \gets (\frac{n}{k_1}, B, F)$-flat filter \Comment See Definition \ref{def:filterG}

    \State $\Delta \gets$ uniform random sample from $[\frac{n}{k_1}]$
    \State $\sigma \gets$ uniform random sample from odd numbers in  $[\frac{n}{k_1}]$
    % \State $\pi \gets$ approximately pairwise independent random permutation over $[\frac{n}{k_1}]$
    \State $\{Z^r\}_{r \in [2k_1]} \gets (k_1,\delta)$-downsampling of $X - \chi$  \Comment See Definition \ref{def:downsampling}
    % \State $\gamma \gets \text{ empty vector of length } (2k_1)$
    \State $\mathbf{H} \gets (H,\dotsc,H)$
    \State $\mathbf{B} \gets (B,\dotsc,B)$
    \State $\{\wh{U}^r\}_{r \in [2k_1]} \gets \textsc{HashToBinsReduced}(\{Z^r\}_{r\in[2k_1]},\wh{\chi},\mathbf{H},n,k_1,\mathbf{B},\sigma,\Delta)$ \label{line:mbl_hash} \Comment See Section \ref{sec:hashing}
    \For{\texttt{$r \in [2k_1]$}}
        \State $\gamma^r \gets \|\wh{U}^r\|_2^2$
    \EndFor
    
    \State {\bf return} $\gamma$ \Comment Length-$2k_1$ vector of $\gamma^r$ values
    \EndProcedure
    
    \end{algorithmic}

\end{algorithm}

\begin{lem} \label{lemm:14}
    \emph{(Properties of exact hashing)}
	Fix the integers $(m,B)$, the parameters $\delta \in \big(0,\frac{1}{20}\big)$ and $F' \ge 10\log\frac{1}{\delta}$, and the signal $Y \in \CC^m$ and $(m,B,F')$-flat filter $H$ (\emph{cf.}, Definition \ref{def:filterG}).  Let $U$ be an $(m,B,H,\sigma,\Delta)$-hashing of $Y$ for uniformly random $\sigma,\Delta \in [m]$ with $\sigma$ odd, and let $\pi(\cdot)$ be defined as in Definition \ref{def:hashing}.  Then, letting $\wh{U}^*$ denote the exact Fourier transform of $U$, we have the following:
\begin{enumerate}
    \item For any set $S \subset [m]$,
    \begin{equation}
    \EE_{\Delta,\pi} \Big[ \Big| \|\wh{Y}_S\|_2^2 - \|\wh{U}^*\|_2^2 \Big|_+ \Big] \le \bigg(10\sqrt{\frac{|S|}{B}} + 15 \frac{|S|}{B} +2\delta^2 \bigg) \|\wh{Y}\|_2^2, \nn %\label{eq:lem14_cond1}
    \end{equation}
    where $\|\wh{Y}_S\|_2^2$ denotes $\sum_{j \in S}|\wh{Y}_j|^2$.
    \item We have
    \begin{equation}
    \EE_{\Delta,\pi} \big[ \|\wh{U}^*\|_2^2 \big] \le 3 \|\wh{Y}\|_2^2. \nn %\label{eq88}
    \end{equation}
\end{enumerate}
\end{lem}
\noindent The proof is given in Appendix \ref{sec:pf_energy_est}.

% We now formally define our energy estimates according to the hashed signals.

%\begin{defn}[Energy estimates]
%    The approximate energy vector $\gamma$ has size $2k_1$, and is constructed by stacking the $\ell_2$ norm of vectors $\wh{U}^r$, where $U^r$ is an $\big(\frac{n}{k_1},B,H)$-hashing of $\wh{Z}^r$ with $B=\frac{4}{\delta^2} k_0$, where $H$ is an $\big(\frac{n}{k_1},B,F'\big)$-flat filter (\emph{cf.}, Definition \ref{def:filterG}) with $F' \ge 10\log\frac{1}{\delta}$ for suitable $\delta \in \big[0,\frac{1}{10}\big]$.  That is, the $r$-th element of $\gamma$ is given by $\gamma^r = \|\wh{U}^r\|_2^2$. 
%    \label{def:energy_approximation}
%\end{defn}

We now present the following lemma, showing that the procedure \textsc{EstimateEnergies} provides us with an estimator satisfying the preconditions of Lemmas \ref{lemm:11} and \ref{lemm:12}.

\medskip
\noindent \textbf{Lemma \ref{lem:gammaproperties}} (\textsc{EstimateEnergies} guarantees -- re-stated from Section \ref{sec:complete_loc}) {\em 
	Given $(n,k_0,k_1)$, the signals $X \in \CC^n$ and $\wh{\chi} \in \CC^n$ with $\|\wh{X} - \wh{\chi}\|_2^2 \ge \frac{1}{\poly(n)} \|\wh{\chi}\|_2$, and the parameter $\delta \in \big(\frac{1}{n},\frac{1}{20}\big)$, the procedure \textsc{EstimateEnergies}$(X,\wh{\chi},n,k_0,k_1$,$\delta)$ returns a vector $\gamma \in \RR^{2k_1}$ such that, for any given set $S^*$ of cardinality at most $10k_0$, we have the following with probability at least $\frac{1}{2}$:
	\begin{enumerate}
		\item $\sum_{r \in [2k_1]} \Big| \|\wh{Z}^r_{S^*}\|_2^2 - \gamma^r \Big|_+ \le  40 \delta \sum_{r \in [2k_1]}  \|\wh{Z}^r\|_2^2$;
		\item $\|\gamma\|_1 \le 10 \sum_{r \in [2k_1]} \|\wh{Z}^r\|_2^2$;
	\end{enumerate}
	where $\{Z^r\}_{r \in [2k_1]}$ is the $(k_1,\delta)$-downsampling of $X - \chi$ (see Definition \ref{def:downsampling}).
	
	Moreover, if $\wh{\chi}$ is $(O(k_0),k_1)$-block sparse, then the sample complexity is $O (\frac{k_0 k_1}{\delta^2} \log^2 \frac{1}{\delta} )$, and the runtime is $O (\frac{k_0 k_1}{\delta^2} \log^2 \frac{1}{\delta} \log^2 n)$.
}
\begin{proof}
    \noindent\textbf{Analysis for the exact hashing sequence:} We start by considering the case that the call to  \textsc{HashToBinsReduced} is replaced by an evaluation of the exact hashing sequence $\wh{U}^{*r}$, i.e., Definition \ref{def:hashing} applied to $Z^r$ resulting from the $(k_1,\delta)$-downsampling of $X - \chi$.  In this case, by applying Lemma \ref{lemm:14} with $\wh{Y} = \wh{Z}^r$, $B = \frac{4}{\delta^2} k_0$ and $S = S^*$ (and hence $|S| \le 10k_0$), the right-hand side of the first claim therein becomes $(5\delta + (\frac{15}{4} + 2)\delta^2)\|\wh{Z}^r\|_2^2 \le 6\delta \|\wh{Z}^r\|_2^2$, since $\delta \le \frac{1}{20}$. By applying the lemma separately for each $r \in [2k_1]$ with $\wh{Y} = \wh{Z}^r$, and summing the corresponding expectations in the two claims therein over $r$, we obtain $\sum_{r \in [2k_1]} \EE\big[ \big| \|\wh{Z}^r_{S^*}\|_2^2 - \|\wh{U}^*\|_2^2 \big|_+ \big] \le 6\delta \sum_{r \in [2k_1]} \|\wh{Z}^r\|_2^2$ and $\sum_{r \in [2k_1]}  \EE\big[ \|\wh{U}^*\|_2^2 \big] \le 3 \sum_{r \in [2k_1]} \|\wh{Z}^r\|_2^2$.  We apply Markov's inequality with a factor of $6$ in the former and $3$ in the latter, to conclude that the quantities $\gamma^{*r} = \|\wh{U}^{*r}\|_2^2$ satisfy
    \begin{gather}
        \sum_{r \in [2k_1]} \Big| \|\wh{Z}^r_{S^*}\|_2^2 - \gamma^{*r} \Big|_+ \le  36 \delta \sum_{r \in [2k_1]}  \|\wh{Z}^r\|_2^2 \label{eq:gamma_exact_1} \\
        \|\gamma^*\|_1 \le 9 \sum_{r \in [2k_1]} \|\wh{Z}^r\|_2^2, \label{eq:gamma_exact_2} 
    \end{gather}
    with probability at least $1/2$. 
    
    \textbf{Incorporating $\bf \frac{1}{n^c}$ error from use of semi-equispaced FFT in \textsc{HashToBinsReduced}:} Since $\wh{U}^r$ is computed using \textsc{HashToBinsReduced}, the energy vector $\gamma$ is different from the exact one $\gamma^*$, and we write
    \begin{equation}
    \begin{split}
    \sum_{r \in [2k_1]} \Big| \|\wh{Z}^r_{S^*}\|_2^2 - \gamma^r \Big|_+ 
    &\le \sum_{r \in [2k_1]} \Big| \|\wh{Z}^r_{S^*}\|_2^2 - \gamma^{*r} \Big|_+ + \big| \gamma^{r} - \gamma^{*r} \big|.
    \end{split}
    \end{equation}
    By substituting $\gamma^r = \|\wh{U}^r\|_2^2$ and $\gamma^{*r} = \|\wh{U}^{*r}\|_2^2$, and using the identity $\big|\|a\|_2^2 - \|b\|_2^2\big| \le 2\|a-b\|_2\cdot\|b\|_2 +  \|a-b\|_2^2$, we can write
    \begin{equation}
    \begin{split}
    \sum_{r \in [2k_1]} \big| \gamma^{r} - \gamma^{*r} \big|
    &\le \sum_{r \in [2k_1]} \Big( 2\|\wh{U}^r - \wh{U}^{*r}\|_2 \|\wh{U}^{*r}\|_2 + \|\wh{U}^r - \wh{U}^{*r}\|_2^2 \Big).
    \end{split} \label{eq:g_diff_init}
    \end{equation}
    Upper bounding the $\ell_2$ norm by the $\ell_{\infty}$ norm times the vector length, we have $\|\wh{U}^r - \wh{U}^{*r}\|_2 \le \sqrt{n}\|\wh{U}^r - \wh{U}^{*r}\|_{\infty} \le n^{-c+1/2} \|\wh{\chi}\|_2$, where the second inequality follows from Lemma \ref{lem:hash2bins}. Moreover, from the definition of $\wh{U}^{*r}$ in Definition \ref{def:hashing} applied to $Z^r$, along with the filter property $\|\wh{G}\|_{\infty}$ in Definition \ref{def:filterG}, it follows that $\|\wh{U}^{*r}\|_2 \le \|\wh{G}\|_{\infty} \|\wh{Z}^r\|_1 \le \sqrt{n} \|\wh{Z}^r\|_2$.  Combining these into \eqref{eq:g_diff_init} gives
    \begin{align}
    \sum_{r \in [2k_1]} \big| \gamma^{r} - \gamma^{*r} \big|
        &\le \sum_{r \in [2k_1]} \Big( 2 n^{-c+1}\|\wh{\chi}\|_2 \|\wh{Z}^r\|_2 +  n^{-2c+1}\|\wh{\chi}\|_2^2 \Big) \nn \\
        &\le 2n^{-c + 2} \sqrt{ \sum_{r \in [2k_1]} \|\wh{\chi}\|_2^2 \cdot \sum_{r \in [2k_1]} \|\wh{Z}^r\|_2^2 } + n^{-2c+1} k_1\|\wh{\chi}\|_2^2 \nn  \\
        &\le 2n^{-c + 3} \|\wh{\chi}\|_2 \sqrt{ \sum_{r \in [2k_1]} \|\wh{Z}^r\|_2^2 } + n^{-2c+2} \|\wh{\chi}\|_2^2. \label{eq:gamma-err}
    \end{align}
    where the second line is by Cauchy-Schwarz, and the third by $k_1 \le n$.
    
    By the second part of Lemma \ref{claim:1} and the assumption $\delta \le \frac{1}{20}$, we have $\sum_{r \in [2k_1]}  \|\wh{Z}^r\|_2^2 \ge \frac{1}{4} \|\wh{X} - \wh{\chi}\|_2^2 \ge \frac{1}{4n^{c'}} \|\wh{\chi}\|_2^2$, where the second equality holds for some $c' > 0$ by the assumption $\|\wh{X} - \wh{\chi}\|_2^2 \ge \frac{1}{\poly(n)} \|\wh{\chi}\|_2$.  Hence, \eqref{eq:gamma-err} gives
	\begin{equation}
    	\sum_{r \in [2k_1]} \big| \gamma^{r} - \gamma^{*r} \big| \le 4\big(n^{-c + 3}n^{{c'}/2} + n^{-2c+2} n^{c'}\big) \sum_{r \in [2k_1]} \|\wh{Z}^r\|_2^2.  \label{eq:gamma-err2}
	\end{equation}
	Since we have chosen $\delta > 1/n$, the coefficient to the summation is upper bounded by $4\delta$ when $c$ is sufficiently large, thus yielding the first part of the lemma upon combining with \eqref{eq:gamma_exact_1}. 
    
    To prove the second part, note that by the triangle inequality,
    \begin{equation}
    \begin{split}
    \|\gamma\|_1 
    &\le \|\gamma^*\|_1 + \Big| \|\gamma\|_1 - \|\gamma^*\|_1 \Big|\\
    &\le 9 \sum_{r \in [2k_1]} \|\wh{Z}^r\|_2^2 + \sum_{r \in [2k_1]} \big| \gamma^{r} - \gamma^{*r} \big|,
    \end{split}
    \end{equation}
    where we have applied \eqref{eq:gamma_exact_2}.  Again applying \eqref{eq:gamma-err2} and noting that the coefficient to the summation is less than one for sufficiently large $c$, the second claim of the lemma follows.
    
    \textbf{Sample complexity and runtime:} The only step that uses samples is the call to \textsc{HashToBinsReduced}.  By Lemma \ref{lem:hash2bins} and the choices $B = \frac{4}{\delta^2} k_0$ and $F = 10\log\frac{1}{\delta}$, this uses $O\big(k_1 F B \log\frac{1}{\delta}\big)= O (\frac{k_0k_1}{\delta^2} \log^2 \frac{1}{\delta})$ samples per call.  
    % Since this line is called $O\big(\log\frac{1}{p}\big)$ times in a loop, this amounts to a total of $O (\frac{k_0 k_1}{\delta^2} \log \frac{1}{\delta} \log \frac{1}{\delta p})$. 
    The time complexity follows by the same argument along the assumption that $\wh{\chi}$ is $(O(k_0),k_1)$-block sparse, with an additional $\log^2 n$ factor following from Lemma \ref{lem:hash2bins}.  Note that the call to \textsc{HashToBinsReduced} dominates the computation of $\gamma^r$, which is $O(k_1 B)$,
\end{proof}

%!TEX root = BlockSparseFT-new.tex

\section{The Block-Sparse Fourier Transform} \label{sec:full_alg}

In this section, we combine the tools from the previous sections to obtain the full sublinear-time block sparse FFT algorithm, and provide its guarantees.

\subsection{Additional Estimation Procedures} \label{sec:utilities}

Before stating the final algorithm, we note the main procedures that it relies on: \textsc{MultiBlockLocate},  \textsc{PruneLocation}, and \textsc{EstimateValues}.  We presented the first of these in Section \ref{sec:ep_sampling}.  The latter two are somewhat more standard, and hence we relegate them to the appendices.  However, for the sake of readability, we provide some intuition behind them here, and state their guarantees.  

We begin with \textsc{PruneLocation}.  The procedure \textsc{MultiBlockLocate} gives us a list of block indices containing the dominant signal blocks with high probability, with a list size $L = O^*\big(k_0 \log k_0 \big)$.  Estimating the values of all of these blocks in every iteration would not only cost $O^*(k_0k_1\log k_0 )$ samples, but would also destroy the sparsity of the input signal: Most of the blocks correspond to noise, and thus the estimation error may dominate the values being estimated.  The \textsc{PruneLocation} primitive is designed to alleviate these issues, pruning $L$ to a list that contains mostly ``signal'' blocks, i.e., blocks that contain a large amount of energy. Some false positives and false negatives occur, but are controlled by Lemma~\ref{lem:prune} below.  The  procedure is given in Algorithm \ref{alg:prune} in Appendix \ref{sec:prune}.

%Hence, we prune the list in order to exclude the blocks that are not signal blocks. Specifically, for every element $j \in L$, we perform a test to determine whether it is a signal block or not.  See Algorithm \ref{alg:prune} in Appendix \ref{sec:prune} for the details.

The following lemma shows that with high probability, the pruning algorithm retains most of the energy in the head elements, while removing most tail elements.

\begin{lem}
    \emph{(\textsc{PruneLocation} guarantees)}
	Given $(n,k_0,k_1)$, a list  of block indices $L$, the parameters $\theta > 0$, $\delta \in \big(\frac{1}{n}, \frac{1}{20}\big)$ and $p \in(0,1)$, and the signals $X \in \CC^n$ and $\wh{\chi} \in \CC^n$ with $\|\wh{X} - \wh{\chi}\|_2 \ge \frac{1}{\poly(n)}\|\wh{\chi}\|_2$, the output  list $L'$ of \textsc{PruneLocation}$(X,\wh{\chi}, L, n, k_0, k_1, \delta , p, \theta)$ has the following properties:
	
	\begin{enumerate}[label={\bf \alph*.}]
		\item Let $\Stail$ denote the tail elements in the signal $\wh{X} - \wh{\chi}$, defined as
		$$\Stail = \Big \{ j \in \Big[\frac{n}{k_1}\Big] \, : \, \|(\wh{X} - \wh{\chi})_{I_j}\|_2 \leq \sqrt{\theta} - \sqrt{\frac{\delta}{k_0}}\|\wh{X} - \wh{\chi}\|_2 \Big \},$$
		where $I_j$ is defined in Definition \ref{def:block}.    Then, we have
		$$\EE \Big[ \big| L' \cap \Stail \big| \Big] \le \delta p \cdot |L \cap \Stail|.$$
		\item Let $\Shead$ denote the head elements in the signal $\wh{X} - \wh{\chi}$, defined as
		$$\Shead = \Big \{ j \in \Big[\frac{n}{k_1}\Big] \, : \, \|(\wh{X} - \wh{\chi})_{I_j}\|_2 \geq \sqrt{\theta} + \sqrt{\frac{\delta}{k_0}} \|\wh{X} - \wh{\chi}\|_2 \Big \}.$$
		Then, we have
		$$\EE \Big[ \sum_{j \in (L \cap \Shead) \backslash L' } \| (\wh{X}-\wh{\chi})_{I_j} \|_2^2 \Big] \leq  \delta p \sum_{j \in L \cap \Shead} \|(\wh{X}-\wh{\chi})_{I_j}\|_2^2.$$
	\end{enumerate}
	Moreover, provided that $\|\wh{\chi}\|_0 = O(k_0k_1)$, the sample complexity is $O(\frac{k_0 k_1}{\delta} \log \frac{1}{\delta p} \log \frac{1}{\delta})$, and the runtime is $O(\frac{k_0 k_1}{\delta} \log \frac{1}{\delta p} \log \frac{1}{\delta} \log n + k_1 \cdot |L| \log\frac{1}{\delta p})$.
    \label{lem:prune}
\end{lem}

\noindent The proof is given in Appendix \ref{sec:prune}.

We are left with the procedure \textsc{EstimateValues}, which is a standard procedure for estimating the values at the frequencies within the blocks after they have been located.  The details are given in Algorithm \ref{alg:est_values} in Appendix \ref{sec:est_vals}.

\begin{lem}
    \emph{(\textsc{EstimateValues} guarantees)}
	For any integers $(n,k_0,k_1)$, list of block indices $L$, parameters $\delta \in \big(\frac{1}{n},\frac{1}{20}\big)$ and $p \in (0,1/2)$, and signals $X \in \CC^n$ and $\wh{\chi} \in \CC^n$ with $\| \wh{X} - \wh{\chi} \|_2 \ge \frac{1}{\poly(n)} \| \wh{\chi} \|_2$, the output $W$ of the function \textsc{EstimateValues}$(X,\wh{\chi}, L, n, k_0, k_1, \delta , p)$ has the following property:
		
		\begin{equation}
		\sum_{f \in \bigcup_{j \in L} I_j} |W_f - (\wh{X}-\wh{\chi})_f|^2 \le \delta \frac{|L|}{3k_0} \| \wh{X}-\wh{\chi} \|_2^2 \nn
		\end{equation}
		with probability at least $1-p$, where $I_j$ is the $j$-th block.  Moreover, provided that $\|\wh{\chi}\|_0 = O(k_0k_1)$, the sample complexity is $O(\frac{k_0 k_1}{\delta} \log \frac{1}{p} \log\frac{1}{\delta})$, and the runtime is $O(\frac{k_0 k_1}{\delta} \log \frac{1}{p} \log\frac{1}{\delta} \log n + k_1 \cdot |L|\log \frac{1}{p})$.
	\label{lem:estimate}
\end{lem}

\noindent The proof is given in Appendix \ref{sec:est_vals}.

\subsection{Statement of the Algorithm and Main Result} \label{sec:final_statement}

Our overall block-sparse Fourier transform algorithm is given in Algorithm \ref{alg:final}.  It first calls \textsc{ReduceSNR}, which performs an iterative procedure that picks up high energy components of the signal, subtracts them from the original signal, and then recurses on the residual signal $X^{(i)} = X - \chi^{(i)}$. Once this is done, the procedure \textsc{RecoverAtConstSNR} performs a final ``clean-up'' step to obtain the $(1+O(\e))$-approximation guarantee.

\begin{algorithm}
	\caption{Block-sparse Fourier transform.}\label{euclid-1}
	\begin{algorithmic}[1]
		
		\Procedure{BlockSparseFT}{$X, n, k_0, k_1 , \SNR', \nu^2 , \e$}  \\
    		\Comment $X \in \CC^n$ is approximately $(k_0,k_1)$-block sparse \\
    		\Comment $(\SNR',\nu^2)$ are upper bounds on $(\SNR,\mu^2)$ from Definition \ref{def:SNR} \\
    		\Comment $\e$ is the parameter for $(1+O(\e))$-approximate recovery
		
		\State $\wh{\chi} \gets$ \textsc{ReduceSNR}($X, n, k_0, k_1, \SNR', \nu^2$). 
		
		\State $\wh{\chi} \gets$ \textsc{RecoverAtConstSNR}($X, \wh{\chi} , n,  k_0, k_1 , \nu^2, \e$).
		
		\State \textbf{return} $\wh{\chi}$
		\EndProcedure
		
		\Procedure{ReduceSNR}{$X, n, k_0, k_1, \SNR', \nu^2$} \Comment Iteratively locate/estimate to reduce SNR
		
		\State $T \gets \log \SNR'$
		\State $\delta \gets $ small absolute constant
		\State $p \gets \frac{\delta}{\log^2 \frac{k_0}{\delta} \log^4 \SNR'}$ \Comment{Failure probability for subroutines}
		\State $\wh{\chi}^{(0)} \gets 0$ \Comment $\wh{\chi}^{(t)}$ is our current estimate of $\wh{X}$
		\For{$t\in\{1,\dotsc,T\}$}
		\State $L \gets$ \textsc{MultiBlockLocate}($X, \wh{\chi}^{(t-1)}, n, k_1, k_0 , \delta, p$)
		
		\State $\theta \gets 10 \cdot 2^{-t}\cdot \nu^2 \SNR' $ \Comment{Threshold for pruning}
		\State $L' \gets$ \textsc{PruneLocation}($X, \wh{\chi}^{(t-1)}, L, n, k_0, k_1 , \delta , p, \theta$)
		
		\State $\wh{\chi}^{(t)} \gets \wh{\chi}^{(t-1)} +$ \textsc{EstimateValues}($X, \wh{\chi}^{(t-1)}, L', n, k_0, k_1, \delta , p$)
		
		\EndFor

		\State \textbf{return} $\wh{\chi}^{T}$
		\EndProcedure
		
		\Procedure{RecoverAtConstSNR}{$X, \wh{\chi}, n, k_0, k_1, \e$} \Comment A final ``clean-up'' step
		
		\State $\eta \gets \text{ small absolute constant }$
		\State $p \gets \frac{ \eta \e } {\log^2 \frac{k_0}{\e} }$\Comment{Upper bound on failure probability for subroutines}
		
		\State $L \gets$ \textsc{MultiBlockLocate}($X, \wh{\chi}, n, k_1, k_0 , \e^2, p$)
		
		\State $\theta \gets 200 \e \nu^2 $
		\State $L' \gets$ \textsc{PruneLocation}($X, \wh{\chi}, L, n, k_0, k_1 , \e , p, \theta$)
		
		\State $W \gets $ \textsc{EstimateValues}($X, \wh{\chi}, L', n, 3k_0/\e, k_1, \e , p$)
		\State $\wh{\chi}' \gets W + \wh{\chi}$

		\State \textbf{return} $\wh{\chi}'$
		\EndProcedure
		
	\end{algorithmic}
	\label{alg:final}
	
\end{algorithm}

With these definitions in place, we can now state our final result, which formalizes Theorem \ref{thm:final}.

\medskip
\noindent \textbf{Theorem \ref{thm:final}} (Upper bound -- formal version) {\em 
	Given $(n,k_0,k_1)$, the parameter $\e \in \big(\frac{1}{n}, \frac{1}{20}\big)$, and the signal $X \in \CC^n$, if $X$, $\SNR'$, $\mu^2$, and $\nu^2$ satisfy the following for $(\mu^2,\SNR)$ given in Definition \ref{def:SNR}:
	\begin{enumerate}
		\item $\mu^2 \le  \nu^2$;
		\item $\|\wh{X}\|_2^2 \le (k_0 \nu^2) \cdot \SNR'$;
		\item $\SNR' = O(\poly(n))$;
		\item $ \mu^2 \ge \frac{\|\wh{X}\|_2^2}{\poly(n)}$;
	\end{enumerate}
	then with probability at least $0.8$, the procedure \textsc{BlockSparseFT}$(X, n, k_0, k_1 , \SNR', \nu^2 , \e)$ satisfies the following: 
	
	(i) The output $\wh{\chi}$ satisfies
    \begin{equation}
	\|\wh{X}-\wh{\chi}\|_2^2  \le k_0 (\mu^2 + O(\e \nu^2)). \nonumber
	\end{equation}
	
	(ii) The sample complexity is is $O^*(k_0 \log (1+k_0) \log \SNR' \log n + k_0 k_1 \log\SNR' + \frac{k_0}{\e^2} \log (1+k_0) \log n + \frac{k_0 k_1}{\e^4} )$, and the runtime is $O^*( k_0\log (1+k_0) \log \SNR' \log^2 n + k_0k_1 \log\SNR' \log^3 n + \frac{k_0}{\e^2}  \log (1+k_0) \log^2 n + \frac{k_0 k_1}{\e^4} \log^2 n + \frac{k_0k_1}{\e^2} \log^3 n)$.
}

The assumptions of the theorem are essentially that we know upper bounds on the tail noise $\mu^2$ and SNR.  Moreover, in order to get the $(1+O(\e))$-approximation guarantee, the former upper bound should be tight to within a constant factor.

In the remainder of the section, we provide the proof of Theorem \ref{thm:final}, deferring the technical details to the appendices.

\paragraph{Guarantees for \textsc{ReduceSNR} and \textsc{RecoverAtConstSNR}.}

The following lemma proves the success of the function \textsc{ReduceSNR}.  We again recall the definitions of $\Err^2$, $\mu^2$, and $\SNR$ in Definition \ref{def:SNR}.

\begin{lem} \label{lem:reduceSNR}
    \emph{(\textsc{ReduceSNR} guarantees)}
	Given $(n,k_0,k_1)$, parameters $(\nu,\SNR')$, and a signal $X \in \CC^n$, if $X$, $\SNR'$, and $\nu^2$ satisfy the following for $(\mu^2,\SNR)$ given in Definition \ref{def:SNR}:
	\begin{enumerate}
		\item $\mu^2 \le  \nu^2$;
		\item $\|\wh{X}\|_2^2 \le (k_0 \nu^2) \cdot \SNR'$;
        \item $\SNR' = O(\poly(n))$;
        \item $ \nu^2 \ge \frac{\|\wh{X}\|_2^2}{\poly(n)}$;
	\end{enumerate}
	then the procedure \textsc{ReduceSNR}$(X, n, k_0, k_1 , \SNR', \nu^2)$ satisfies the following guarantees with probability at least $0.9$ when the constant $\delta$ therein is sufficiently small: 
	
	(i) The output $\wh{\chi}^T$ satisfies
	\begin{gather}
		\wh{\chi}^T \text{ \em is } (3k_0,k_1)\text{\em-block sparse} \nn \\ 
		\|\wh{X}-\wh{\chi}^T\|_2^2  \le 100 k_0 \nu^2. \nn
	\end{gather}
	
	(ii) The number of samples used is $O^*(k_0 \log (1+ k_0) \log \SNR'  \log n + k_0 k_1 \log\SNR')$, and the runtime is $O^*( k_0\log(1+ k_0) \log \SNR' \log^2 n + k_0k_1 \log\SNR' \log^3 n )$.
    % Moreover, the sample complexity of the algorithm is $O(\frac{k_0}{\delta} \log^2 \SNR' \cdot  \log \frac{k_0}{\delta}  \log^4 \frac{1}{\delta p} + \frac{k_0 k_1}{\delta^2} \log \frac{1}{\delta p} \log \frac{1}{\delta} \log\SNR')$, and the runtime is $O( \frac{k_0 k_1}{\delta^2} \log\frac{1}{\delta p} \log^2 n + \frac{k_0}{\delta}  \log \frac{k_0}{\delta} \log \SNR' \cdot \log^4 \frac{1}{\delta p} \log^2 n)$, where $p$ and $\delta$ are local parameters defined inside the procedure.
\end{lem}

\noindent The proof is given in Appendix \ref{sec:pf_reduce_snr}

The following lemma proves the success of the function \textsc{RecoverAtConstSNR}.  

\begin{lem} 	\label{lem:constSNR}
    \emph{(\textsc{RecoverAtConstSNR} guarantees)}
	Given $(n,k_0,k_1)$, parameters $\nu^2 \ge \frac{\|\wh{X}\|_2^2}{\poly(n)}$ and $\e \in \big(\frac{1}{n}, \frac{1}{20}\big)$, and the signals $X \in \CC^n$ and $\wh{\chi} \in \CC^n$ satisfying
	\begin{enumerate}
		\item $\Err^2 (\wh{X} - \wh{\chi} , 10k_0 , k_1) \le  k_0 \nu^2$;
		\item $\|\wh{X}- \wh{\chi}\|_2^2 \le 100k_0 \nu^2$;
	\end{enumerate}
	the procedure  \textsc{RecoverAtConstSNR}$(X, \wh{\chi}, n, k_0, k_1, \e)$ satisfies the following guarantees with probability at least $0.9$ when the constant $\eta$ therein is sufficiently small: (i) The output $\wh{\chi}'$ satisfies
    \begin{equation}
	\|\wh{X}-\wh{\chi}'\|_2^2 \le \Err^2 (\wh{X} - \wh{\chi} , 10k_0 , k_1) + (4 \cdot 10^5) \e k_0 \nu^2.
	\end{equation}
	(ii) If $\wh{\chi}$ is $(O(k_0),k_1)$-block sparse, then the number of samples used is $ O^*( \frac{k_0}{\e^2} \log(1+ k_0) \log n + \frac{k_0 k_1}{\e^4} )$ and the runtime is $O^*( \frac{k_0}{\e^2}  \log(1+k_0) \log^2 n + \frac{k_0 k_1}{\e^4} \log^2 n + \frac{k_0k_1}{\e^2} \log^3 n )$.	
\end{lem}

\noindent The proof is given in Appendix \ref{sec:pf_constant_snr}.

\paragraph{Proof of Theorem \ref{thm:final}.}
We are now in a position to prove Theorem \ref{thm:final} via a simple combination of Lemmas \ref{lem:reduceSNR} and \ref{lem:constSNR}.

\paragraph {Success event associated with \textsc{ReduceSNR}:} Define a successful run of \textsc{ReduceSNR}($X, n, k_0, k_1 , \SNR, \nu^2$) to mean mean the following conditions on the output $\wh{\chi}^T$:
\begin{gather}
	\|\wh{X}-\wh{\chi}^T\|_2^2  \le 100 k_0 \nu^2 \nn \\
	\Err^2(\wh{X}-\wh{\chi}^T, 10k_0 , k_1) \le k_0\mu^2. \nn
\end{gather}
By Lemma \ref{lem:reduceSNR}, it follows that the probability of having a successful run of \textsc{ReduceSNR} is at least $0.9$.  Note that the second condition is not explicitly stated in Lemma \ref{lem:reduceSNR}, but it follows by using $3k_0$ blocks to cover the parts where $\wh{\chi}^T$ is non-zero, and $k_0$ blocks to cover the dominant blocks of $\wh{X}$, in accordance with Definition \ref{def:SNR}.

\paragraph {Success event associated with \textsc{RecoverAtConstSNR}:} Define a successful run of \textsc{RecoverAtConstSNR}($X, \wh{\chi}, k_0, k_1 , n, \e$) to mean the following conditions on the output $\wh{\chi}$:
$$	\|\wh{X}-\wh{\chi}'\|_2^2 \le \Err^2 (\wh{X} - \wh{\chi} , 10k_0 , k_1) + (4\cdot10^5) \e k_0 \nu^2.$$
Conditioning on event of having a successful run of \textsc{ReduceSNR}, by Lemma \ref{lem:constSNR}, it follows that the probability of having a successful run to \textsc{RecoverAtConstSNR} is at least $0.9$.

By a union bound, the aforementioned events occur simultaneously with probability at least $0.8$, as desired.  Moreover, the sample complexity and runtime are a direct consequence of summing the contributions from Lemmas \ref{lem:reduceSNR} and \ref{lem:constSNR}.

%\begin{lem}
%At each iteration $i$ of the algorithm we have the following:
%$$\|\wh{X} - \wh{\chi}^{i+1}\| \le \beta \|\wh{X} - \wh{\chi}^i\|_2^2$$
%with probability $1-\frac{a}{\log \SNR}$ if $\|\wh{X} - \wh{\chi}^i\|_2^2 \ge \frac{\|\wh{X}\|_2^2}{2^i}$ otherwise we will have:
%$$\|\wh{X} - \wh{\chi}^{i+1}\| \le (1+\lambda) \|\wh{X} - \wh{\chi}^i\|_2^2$$
%with probability $1-\frac{a}{\log \SNR}$.
%\end{lem}
%!TEX root = BlockSparseFT-new.tex

\section{Lower Bound} \label{sec:lower}

Our upper bound in Theorem \ref{thm:final}, in several scaling regimes, provides a strict improvement over standard sparse FFT algorithms in terms of sample complexity.  The corresponding algorithm is inherently \emph{adaptive}, which raises the important question of whether adaptivity is necessary in order to achieve these improvements.  In this section, we show that the answer is affirmative, by proving the following formalization of Theorem \ref{thm:lb}.

\medskip
\noindent \textbf{Theorem \ref{thm:lb}} (Lower bound -- formal version) {\em 
    Fix $(n,k_0,k_1)$ and $C > 0$, and suppose that there exists a non-adaptive algorithm that, when given a signal $Y$ with Fourier transform $\wh{Y}$, outputs a signal $\wh{Y}'$ satisfying the following $\ell_2/\ell_2$-guarantee with probability at least $\frac{1}{2}$:
    \begin{equation}
        \|\wh{Y} - \wh{Y}'\|_2^2 \le C \min_{\wh{Y}^*\mathrm{~is~}(k_0,k_1)\mathrm{-block~sparse}} \|\wh{Y} - \wh{Y}^*\|_2^2. \label{eq:algo_guarantee}
    \end{equation}
    Then the number of samples taken by the algorithm must behave as $\Omega\big( k_0 k_1 \log \frac{n}{k_0k_1} \big)$.
}

Hence, for instance, if $k_0 = O(1)$ and $\SNR= O(1)$ then our adaptive algorithm uses $O( k_1 + \log n )$ samples, whereas any non-adaptive algorithm must use $\Omega\big(k_1 \log\frac{n}{k_1}\big)$ samples.

The remainder of this section is devoted to the proof of Theorem \ref{thm:lb}.  Throughout the section, we let $k = k_0k_1$ denote the total sparsity.

\textbf{High-level overview:}
Our analysis follows the information-theoretic framework of~\cite{PW}.  However, whereas \cite{PW} considers a signal with $k$ arbitrary dominant frequency locations and uniform noise, we consider signals where the $k = k_0k_1$ dominant frequencies are (nearly) contiguous, and both the noise and signal are concentrated on an $O\big(\frac{1}{k}\big)$ fraction of the time domain.  

As a result, while the difficulty in~\cite{PW} arises from the fact that the algorithm needs to recover roughly $\log\frac{n}{k}$ bits per frequency location for $k$ such locations, our source of difficulty is different. In our signal, there are only roughly $\log\frac{n}{k}$ bits to be learned about the location of \emph{all} the blocks in frequency domain, but the signal is tightly concentrated on an $O\big(\frac{1}{k}\big)$ fraction of the input space. As a consequence, any non-adaptive algorithm is bound to waste most of its samples on regions of the input space where the signal is zero, and only an $O\big(\frac{1}{k}\big)$ fraction of its samples can be used to determine the single frequency that conveys the location of the blocks.  In the presence of noise, this results in a lower bound on sample complexity of $\Omega\big(k \log \frac{n}{k}\big)$. 

\textbf{Information-theoretic preliminaries:}
We will make use of standard results from information theory, stated below.  Here and subsequently, we use the notations $H(X)$, $H(Y|X)$, $I(X;Y)$ and $I(X;Y|U)$ for the (conditional) Shannon entropy and (conditional) mutual information (e.g., see \cite{Cov01}).  % Note that the random variables $(X,Y,Z)$ here are not related to random variables in the rest of the section, nor are they necessarily related among the different lemmas stated here.
We first state Fano's inequality, a commonly-used tool for proving lower bounds by relating a conditional entropy to an error probability.

\begin{lem} \emph{(Fano's Inequality \cite[Lemma 7.9.1]{Cov01})} \label{thm:Fano}
    Fix the random variables $(X,Y)$ with $X$ being discrete, let $X'$ be an estimator of $X$ such that $X \to Y \to X'$ forms a Markov chain (i.e., $X$ and $X'$ are conditionally independent given $Y$), and define $\pe := \PP[X' \ne X]$.  Then
    \begin{equation}
        H(Y|X) \le 1 + \pe \log|\Xc|, \nn
    \end{equation}
    where $\Xc = \supp(X)$.  Consequently, if $X$ is uniformly distributed, then
    \begin{equation}
        I(X;Y) \ge -1 + (1-\pe) \log|\Xc|. \nn
    \end{equation}
\end{lem}

The next result gives the formula for the capacity of a complex-valued additive white Gaussian noise channel, often referred to as the Shannon-Hartley theorem.  Here and subsequently, $\CN(\mu,\sigma^2)$ denotes the complex normal distribution.

\begin{lem} \emph{(Complex Gaussian Channel Capacity \cite[Thm.~2.8.1]{Cov01})} \label{thm:shannon} % \cite[Sec.~5.2.1]{Tse05}
    For $Z \sim \CN(0,\sigma_z^2)$ and  any complex random variable $X$ with $\EE[|X|^2] = \sigma_x^2$, we have 
    \begin{equation}
        I(X;X+Z) \le \log\bigg(1+\frac{\sigma_x^2}{\sigma_z^2}\bigg), \nn
    \end{equation}
    with equality if $X \sim \CN(0,\sigma_x^2)$.
\end{lem}

The following lemma states the data processing inequality, which formalizes the statement that processing a channel output cannot increase the amount of information revealed about the input.

\begin{lem} \emph{(Data Processing Inequality \cite[Thm.~2.8.1]{Cov01})} \label{thm:dpi}
    For any random variables $(X,Y,Z)$ such that $X \to Y \to Z$ forms a Markov chain, we have $I(X;Z) \le I(X;Y)$.
\end{lem}

Finally, the following lemma bounds the mutual information between two vectors in terms of the individual mutual information terms between components of those vectors.

\begin{lem} \emph{(Mutual Information for Vectors \cite[Lemma 7.9.2]{Cov01})} \label{thm:mi_vec}
    For any random vectors $\Xv = (X_1,\dotsc,X_n)$ and $\Yv  = (Y_1,\dotsc,Y_n)$, if the entries of $\Yv$ are conditionally independent given $\Xv$, then $I(\Xv;\Yv) \le \sum_{i=1}^n I(X_i;Y_i)$.
\end{lem}

\textbf{A communication game:}
We consider a communication game consisting of channel coding with a state known at both the encoder (Alice) and decoder (Bob), where block-sparse recovery is performed at the decoder.  Recalling that we are in the non-adaptive setting, by Yao's minimax principle, we can assume that the samples are deterministic and require probability-$\frac{1}{2}$ recovery over a random ensemble of signals, as opposed to randomizing the samples and requiring constant-probability recovery for any given signal in the ensemble.  Hence, we denote the \emph{fixed} sampling locations by $\Ac$.

We now describe our hard input distribution.  Each signal in the ensemble is indexed by two parameters $(u,f^*)$, and is given by
\begin{equation}
    X_t = 
    \begin{cases}
        \omega^{f^* t} & t \in \{ u, \dotsc, u + \frac{C'n}{k} - 1\} \\
        0 & \mathrm{otherwise},
    \end{cases} \label{eq:lb_xt}
\end{equation}
where $C' > 0$ is a constant that will be chosen later, and where all indices are modulo-$n$.  Hence, each signal is non-zero only in a window of length $\frac{C'n}{k}$, and within that window, the signal oscillates at a rate dictated by $f^*$.  Specifically, $u$ specifies where the signal is non-zero in time domain, and $f^*$ specifies where the energy is concentrated in frequency domain.  We restrict the values of $u$ and $f^*$ to the following sets:
\begin{gather}
    \Uc = \Big\{ \frac{C'n}{k}, \frac{2C'n}{k} \dotsc, \Big(\frac{k}{C'} - 1\Big)\frac{C'n}{k}, n \Big\}\nn \\
    \Fc = \Big\{ k, 2k, \dotsc, \Big(\frac{n}{k} - 1\Big)k, n \Big\}. \nn
\end{gather}
The communication game is as follows:
\begin{enumerate}
    \item Nature selects a \emph{state} $U$ and a \emph{message} $F$ uniformly from $\Uc$ and $\Fc$, respectively.  
    \item An \emph{encoder} maps $(U,F)$ to the signal $X$ according to \eqref{eq:lb_xt}.
    \item A \emph{state-dependent channel} adds independent $\CN(0,\alpha)$ noise to $X_t$ for each $t \in \{ u, \dotsc, u + \frac{C'n}{k} - 1\}$, while keeping the other entries noiseless.  This is written as $Y_t = X_t + W_t$, where
    \begin{equation}
        W_t \sim
        \begin{cases}
            \CN(0,\alpha) & \{ u, \dotsc, u + \frac{C'n}{k} - 1\} \\
            0 & \mathrm{otherwise}.
        \end{cases} \label{eq:Wt}
    \end{equation}
    The channel output is given by $Y = \{Y_t\}_{t \in \Ac}$ for the sampling locations $\Ac$.
    \item A \emph{decoder} receives $U$ and $Y$, applies $(k_0,k_1)$-block sparse recovery to $Y$ to obtain a signal $\wh{Y}'$, and then selects $F'$ to be the frequency $f' \in \Fc$ such that the energy in $\wh{Y}'$ within the length-$k$ window centered at $f'$ is maximized:
    \begin{equation}
        F' = \argmax_{f' \in \Fc} \|\wh{Y}'_{I_k(f')}\|_2^2, \label{eq:F_decoder}
    \end{equation}
    where $I_k(f') = \{f' + \Delta\,:\, \Delta \in [k]\}$.
\end{enumerate}
We observe that if adaptivity were allowed, then the knowledge of $U$ at the decoder would make the block-sparse recovery easy -- one could let all of the samples lie within the window given in \eqref{eq:lb_xt}.  The problem is that we are in the \emph{non-adaptive} setting, and hence we must take enough samples to account for all of the possible choices of $U$.

We denote the subset of $\Ac$ falling into $\{ u, \dotsc, u + \frac{C'n}{k} - 1\}$ by $\Ac_u$, its cardinality by $m_u$ and the total number of measurements by $m = |\Ac| = \sum_{u} m_u$.  Moreover, we let $X_{\Ac_u}$ and $Y_{\Ac_u}$ denote the sub-vectors of $X$ and $Y$ indexed by $\Ac_u$. 

% Conditioned on $U = u$, we have the Markov chain $F \to X \to Y \to Y' \to F'$, but $U$ also affects the transitions $F \to X$, $X \to Y$, and $Y' \to F'$.

\textbf{Information-theoretic analysis:}
We first state the following lemma.

\begin{lem} \label{lem:mi_upper}
    \emph{(Mutual information bound)}
    In the setting described above, the conditional mutual information $I(F;Y|U)$ satisfies $I(F;Y|U) \le \frac{C'm}{k} \log\big(1+\frac{1}{\alpha}\big)$, where $\alpha$ is variance of the additive Gaussian noise within $\Ac_u$.
\end{lem}
\begin{proof}
    We have { \allowdisplaybreaks
    \begin{align}
        I(F;Y|U) 
            &= \frac{C'}{k}\sum_u I(F;Y|U=u) \nn \\ % \label{eq:mi_bound1} \\
            &= \frac{C'}{k}\sum_u I(F;Y_{\Ac_u}|U=u) \nn \\ % \label{eq:mi_bound2} \\
            &\le \frac{C'}{k}\sum_u I(X_{\Ac_u};Y_{\Ac_u}|U=u) \nonumber \\ % \label{eq:mi_bound3} \\
            &\le \frac{C'}{k}\sum_u \sum_{t \in \Ac_u} I(X_t;Y_t|U=u) \nn \\ % \label{eq:mi_bound4} \\
            &\le \frac{C'}{k}\sum_u m_u \log\Big(1+\frac{1}{\alpha}\Big) \nn \\ % \label{eq:mi_bound5} \\
            &= \frac{C'm}{k} \log\Big(1+\frac{1}{\alpha}\Big), \label{eq:mi_bound6}
    \end{align} }
    where:
    \begin{itemize}
        \item Line 1 follows since $U$ is uniform on a set of cardinality $\frac{k}{C'}$;
        \item Line 2 follows since given $U=u$, only the entries of $Y$ indexed by $\Ac_u$ are dependent on $F$ (\emph{cf.}, \eqref{eq:lb_xt});
        \item Line 3 follows by noting that given $U=u$ we have the Markov chain $F \to X_{\Ac_u} \to Y_{\Ac_u}$, and applying the data processing inequality (Lemma \ref{thm:dpi});
        \item Line 4 follows from Lemma \ref{thm:mi_vec} and \eqref{eq:Wt}, where the conditional independence assumption holds because we have assumed the random variables $W_t$ are independent;
        \item Line 5 follows from the Shannon-Hartley Theorem (Lemma \ref{thm:shannon});  in our case, the signal power is exactly one by \eqref{eq:lb_xt}, and the average noise energy is $\alpha$ by construction.  
    \end{itemize}
\end{proof}

Next, defining $\delta_u := \PP[ F' \ne F \,|\, U=u ]$, Fano's inequality (Lemma \ref{thm:Fano}) gives
\begin{equation}
    I(F;Y|U=u) \ge -1 + (1-\delta_u)\log\frac{n}{k}, \nn
\end{equation}
and averaging both sides over $U$ gives
\begin{equation}
    I(F;Y|U) \ge -1 + (1-\delta)\log\frac{n}{k}, \nn
\end{equation}
where $\delta := \EE[\delta_U] = \PP[ F' \ne F ]$.

Hence, and by Lemma \ref{lem:mi_upper}, if we can show that our $\ell_2/\ell_2$-error guarantee \eqref{eq:algo_guarantee} gives $F' = F$ constant probability, then we can conclude that
\begin{equation}
    m \ge \frac{k\big( (1-\delta)\log\frac{n}{k} - 1 \big)}{ C' \log\big(1+\frac{1}{\alpha}\big) } = \Omega\bigg( k \log \frac{n}{k} \bigg). \nn
\end{equation}
We therefore conclude the proof of Theorem \ref{thm:lb} by proving the following lemma.

\begin{lem}
    \emph{(Probability of error characterization)}
    Fix $(n,k_0,k_1)$ and $C > 0$.  If the $(k_0,k_1)$-block sparse recovery algorithm used in the above communication game satisfies \eqref{eq:algo_guarantee} with probability at least $\frac{1}{2}$, then there exist choices of $C'$ and $\alpha$ such that the decoder's estimate of $F'$ according to \eqref{eq:F_decoder} satisfies $F' = F$ with probability at least $\frac{1}{4}$.
\end{lem}
\begin{proof}
    By the choice of estimator in \eqref{eq:F_decoder}, it suffices to show that $\wh{Y}'$, the output of the block-sparse Fourier transform algorithm, has more than half of its energy within the length-$k$ window $I_k(f^*)$ centered of $f^*$.  We show this in three steps.
    
    \textbf{Characterizing the energy of $\wh{X}$ within $I_k(f^*)$: }
    The Fourier transform of $X$ in \eqref{eq:lb_xt} is a shifted sinc function of ``width'' $\frac{k}{C'}$ centered at $f^*$ when the time window is centered at zero, and more generally, has the same magnitude as this sinc function.  Hence, by letting $C'$ be suitably large, we can ensure that an arbitrarily high fraction of the energy of $\wh{X}$ falls within the length-$k$ window centered at $f^* \in \Fc$.  Formally, we have
    \begin{equation}
        \| \wh{X}_{I_k(f^*)} \|_2^2 \ge (1-\eta) \|\wh{X}\|_2^2 \label{eq:XB*}
    \end{equation}
    for $\eta \in (0,1)$ that we can make arbitrarily small by choosing $C'$ large.
    
    \textbf{Characterizing the energy of $\wh{Y}$ within $I_k(f^*)$:} 
    We now show that, when the noise level $\alpha$ in \eqref{eq:Wt} it sufficiently small, the energy in $\wh{Y}$ within $I_k(f^*)$ is also large with high probability:
    \begin{equation}
        \sum_{f \in I_k(f^*)} |\wh{Y}_f|^2 \ge (1-2\eta) \|\wh{X}\|_2^2. \label{eq:sum_f_lb0}
    \end{equation}
    To prove this, we first note that $|\wh{Y}_f|^2 = |\wh{X}_f + \wh{W}_f|^2$ for all $f \in [n]$, from which it follows that
    \begin{align}
        \bigg| \sum_{f \in I_k(f^*)} |\wh{Y}_f|^2 - \sum_{f \in I_k(f^*)} |\wh{X}_f|^2\bigg| \le  \sum_{f \in I_k(f^*)} |\wh{W}_f|^2 + 2\sum_{f \in I_k(f^*)} |\wh{X}_f|\cdot|\wh{W}_f|. \nn
    \end{align}
    Upper bounding the summation over $|\wh{W}_f|^2$ by the total noise energy, and upper bounding the summation over $|\wh{X}_f|\cdot|\wh{W}_f|$ using the Cauchy-Schwarz inequality, we obtain
    \begin{equation}
        \bigg|\sum_{f \in I_k(f^*)} |\wh{Y}_f|^2 - \sum_{f \in I_k(f^*)} |\wh{X}_f|^2 \bigg| \le \|\wh{W}\|_2^2 + 2\|\wh{X}\|_2\cdot\|\wh{W}\|_2. \label{eq:block_noise_2}
    \end{equation}
    We therefore continue by bounding the \emph{total} noise energy $\|\wh{W}\|_2^2$; the precise distribution of the noise across different frequencies is not important for our purposes.
    
    Recall that every non-zero entry of $X$ has magnitude one, and every non-zero time-domain entry of $W$ is independently distributed as $\CN(0,\alpha)$.  Combining these observations gives $\EE[\|W\|_2^2] = \alpha\|X\|_2^2$, or equivalently $\EE[\|\wh{W}\|_2^2] = \alpha\|\wh{X}\|_2^2$ by Parseval.  Therefore, by Markov's inequality, we have $\|\wh{W}\|_2^2 \le 4\alpha \|\wh{X}\|_2^2$ with probability at least $\frac{3}{4}$.  When this occurs, \eqref{eq:block_noise_2} gives
    \begin{equation}
        \bigg|\sum_{f \in I_k(f^*)} |\wh{Y}_f|^2 - \sum_{f \in I_k(f^*)} |\wh{X}_f|^2\bigg| \le 4(\alpha + \sqrt{\alpha}) \|\wh{X}\|_2^2. \label{eq:block_noise_3}
    \end{equation}
    If we choose $\alpha = \frac{\eta^2}{100}$, then we have $4(\alpha + \sqrt{\alpha}) = \frac{\eta^2}{25} + \frac{4\eta}{10} \le \eta$.  In this case, by \eqref{eq:XB*} and \eqref{eq:block_noise_3}, the length-$k$ window $I_k(f^*)$ centered at $f^*$ satisfies \eqref{eq:sum_f_lb0}.

%    and similarly, the window $I_k(f')$ centered at any other $ f^{\prime} \in \Fc \backslash \{f^*\}$ satisfies
%    \begin{equation}
%        \sum_{f \in I_k(f')} |\wh{Y}_f|^2 \le 2\eta \|\wh{X}\|_2^2. \nn
%    \end{equation}
    \textbf{Characterizing the energy of $\wh{Y}'$ within $I_k(f^*)$:} 
    The final step is to prove that \eqref{eq:sum_f_lb0} and \eqref{eq:algo_guarantee} imply the following with constant probability for a suitable choice of $\eta$:
    \begin{equation}
        \sum_{f \in I_k(f^*)} |\wh{Y}'_{f}|^2 > \frac{1}{2} \|\wh{Y}'\|_2^2, \label{eq:lb_final}
    \end{equation}
    where $\wh{Y}'$ is the output of the block-sparse recovery algorithm.  This clearly implies that $F = F'$, due to our choice of estimator in \eqref{eq:F_decoder}.
        
    As a first step towards establishing \eqref{eq:lb_final}, we rewrite \eqref{eq:sum_f_lb0} as 
        \begin{equation}
            \sum_{f \in [n] \backslash I_k(f^*)} |\wh{Y}_f|^2 \le \|\wh{Y}\|_2^2 - (1-2\eta) \|\wh{X}\|_2^2. \label{eq:sum_f_lb}
        \end{equation}
    We can interpret \eqref{eq:sum_f_lb} as an error term $\|\wh{Y} - \wh{Y}^*\|_2^2$ for a signal $\wh{Y}^*$ coinciding with $\wh{Y}$ within $I_k(f^*)$ and being zero elsewhere.  Since $I_k(f^*)$ contains $k$ contiguous elements, this signal is $(k_0,k_1)$-block sparse, and hence if the guarantee in \eqref{eq:algo_guarantee} holds, then combining with \eqref{eq:sum_f_lb} gives
    \begin{equation}
        \|\wh{Y} - \wh{Y}'\|_2^2 \le C \Big( \|\wh{Y}\|_2^2 - (1-2\eta) \|\wh{X}\|_2^2 \Big). \label{eq:y_est_lb}
    \end{equation}
    We henceforth condition on both \eqref{eq:algo_guarantee} and the above-mentioned event $\|\wh{W}\|_2^2 \le 4\alpha \|\wh{X}\|_2^2$.  Since the former occurs with probability at least $\frac{1}{2}$ by assumption, and the latter occurs with probability at least $\frac{3}{4}$, their intersection occurs with probability at least $\frac{1}{4}$.

    Next, we write the conditions in \eqref{eq:sum_f_lb0} and \eqref{eq:y_est_lb} in terms of $\|\wh{Y}\|_2^2$, rather than $\|\wh{X}\|_2^2$.  Since $\wh{X} = \wh{Y} - \wh{W}$, we can use the triangle inequality to write $\|\wh{X}\|_2 \ge \|\wh{Y}\|_2  - \|\wh{W}\|_2$, and combining this with  $\|\wh{W}\|_2^2 \le 4\alpha \|\wh{X}\|_2^2$, we obtain $\|\wh{X}\|_2 \ge \frac{\|\wh{Y}\|_2}{1 + 2\sqrt{\alpha}}$.   Hence, we can weaken \eqref{eq:sum_f_lb0} and \eqref{eq:y_est_lb} to
    \begin{gather}
        \sum_{f \in I_k(f^*)} |\wh{Y}_f|^2 \ge \frac{1-2\eta}{  (1+2\sqrt{\alpha})^2 } \|\wh{Y}\|_2^2 \ge 0.99 \|\wh{Y}\|_2^2 \label{eq:lb_final1a} \\
        \|\wh{Y} - \wh{Y}'\|_2^2 \le C\Big(1 - \frac{1-2\eta}{ (1 + 2\sqrt{\alpha})^2 }\Big)  \|\wh{Y}\|_2^2 \le 0.01 \|\wh{Y}\|_2^2, \label{eq:lb_final2a} 
    \end{gather}
    where the second step in each equation holds for sufficiently small $\eta$ due to the choice $\alpha = \frac{\eta^2}{100}$.
    
    It only remains to use \eqref{eq:lb_final1a}--\eqref{eq:lb_final2a} to bound the left-hand side of \eqref{eq:lb_final}.  To do this, we first note that by interpreting both \eqref{eq:lb_final1a} and \eqref{eq:lb_final2a} as bounds on $\|\wh{Y}\|_2^2$, and using $\|\wh{Y} - \wh{Y}'\|_2^2 \ge \|(\wh{Y} - \wh{Y}')_{I_k(f^*)}\|_2^2$ in the latter, we have
    \begin{equation}
        \sum_{f \in I_k(f^*)} |\wh{Y}_f - \wh{Y}'_f|^2 \le 0.02 \sum_{f \in I_k(f^*)} |\wh{Y}_f|^2, \nonumber
    \end{equation}
    since $\frac{0.01}{0.99} \le 0.02$.  Taking the square root and applying the triangle inequality to the $\ell_2$-norm on the left-hand side, we obtain
    \begin{equation}
        \sum_{f \in I_k(f^*)} |\wh{Y}'_f|^2 \ge (1- \sqrt{0.02})^2 \sum_{f \in I_k(f^*)} |\wh{Y}_f|^2. \label{eq:lb_final2b}
   \end{equation}
   
   Next, writing $\|\wh{Y}'\|_2 = \|\wh{Y} + (\wh{Y}' - \wh{Y})\|_2$, and applying the triangle inequality followed by \eqref{eq:lb_final2a}, we have $\|\wh{Y}'\|_2 \le 1.1 \|\wh{Y}\|_2$, and hence $\|\wh{Y}\|_2 \ge 0.9 \|\wh{Y}'\|_2$.  Squaring and substituting into \eqref{eq:lb_final1a}, we obtain
    \begin{gather}
        \sum_{f \in I_k(f^*)} |\wh{Y}_f|^2 \ge 0.8 \|\wh{Y}'\|_2^2. \label{eq:lb_final1b}
    \end{gather}
    Finally, combining \eqref{eq:lb_final2b} and \eqref{eq:lb_final1b} yields \eqref{eq:lb_final}, and we have thus shown that \eqref{eq:lb_final} holds (and hence $F = F'$) with probability at least $\frac{1}{4}$.
\end{proof}

\section{Acknowledgements}
MK would like to thank Piotr Indyk for helpful discussions.  VC and JS are supported by the European Commission (ERC Future Proof), SNF (200021-146750 and CRSII2-147633), and `EPFL Fellows' program (Horizon2020 665667).

\appendix 

%!TEX root = BlockSparseFT-new.tex

\section{Omitted Proofs from Section \ref{sec:overview}}

\subsection{Proof of Lemma \ref{lem:filter_properties}} \label{sec:pf_filter}

Our filter construction is similar to \cite{IK14a}, but we prove and utilize different properties, and hence provide the details for completeness.  

\begin{definition}[Rectangular pulse]
    For an even integer $B'$, let $\rect^{B'}$ denote the rectangular pulse of width $B'-1$, i.e. 
    \begin{equation*}
    \rect^{B'}_t=\left\{
    \begin{array}{ll}
    1,&\text{~if~}|t|< \frac{B'}{2} \\
    0&\text{otherwise.}
    \end{array}
    \right.
    \end{equation*}
\end{definition}

For an integer $B'>0$ a power of $2$, define the length-$n$ signal
\begin{equation}\label{eq:boxcar-time}
    {W}=\left(\frac{{n}}{B'-1}\cdot \rect^{B'}\right) \star \cdots  \star \left(\frac{{n}}{B'-1}\cdot \rect^{B'}\right),
\end{equation}
where the convolution is performed $F$ times. As noted in \cite{IK14a}, we have $\supp( W^{\fc}) \subseteq [-\fc \cdot B', \fc\cdot B']$, and the Fourier transform is given by
\begin{equation}\label{eq:boxcar-freq}
\begin{split}
\wh{W}_f&=\left(\frac{1}{B'-1}\sum_{|f'|< \frac{B'}{2}} \omega_n^{ff'}\right)^{\fc}=\left(\frac{\sin(\pi (B'-1) f/n)}{(B'-1)\sin (\pi f/n)}\right)^{\fc}
\end{split}
\end{equation}
for $f \ne 0$, and $W_0 = 1$.

\begin{lem}\label{cl:filter-h-prop}
    \emph{(Properties of $W$)}
    For every  even $\fc\geq 2$, the following hold for the signal $W$ defined in \eqref{eq:boxcar-time}--\eqref{eq:boxcar-freq}:
    \begin{description}
    \item[1] $\wh{W}_f\in [0, 1]$ for all $f \in \nsq$;
    \item[2] There exists an absolute constant $C\geq 0$ such that for every $\lambda>1$, 
    $$
    \sum_{f\in \nsq,\, |f|\geq \frac{\lambda\cdot n}{2B'}} \wh{W}_f\leq (C/\lambda)^{F-1} \sum_{f\in \nsq} \wh{W}_f.
    $$
    \end{description}
\end{lem}
\begin{proof}
    First note that the maximum of $\wh{W}_f$  is achieved at $0$ and equals $1$. Since $F$ is even by assumption, we have from \eqref{eq:boxcar-freq} that $\wh{W}_f\geq 0$ for all $f$. These two facts establish the first claim.
    
    To prove the second claim, note that for all $f \in \nsq$, we have
    \begin{align}
        \wh{W}_f = \left|\frac{\sin(\pi (B'-1) f/n)}{(B'-1)\sin (\pi f/n)}\right|^{\fc}&\leq \left|\frac{1}{(B'-1)\sin (\pi f/n)}\right|^{\fc} \text{~~~~(since $|\sin (\pi x)|\leq 1$)} \nonumber \\
        &\leq \left|\frac{1}{(B'-1)2 |f|/n}\right|^{\fc} \text{~~~~(since $|\sin (\pi x)|\geq 2|x|$ for $|x|\leq 1/2$)}. \label{eq:filter-ub} %\\
        %&= \left|\frac{B'}{2(B'-1)}\right|^{\fc} \cdot \left|\frac{n}{B' |i|}\right|^{\fc}, \\
    \end{align}
    We claim that this can be weakened to
    \begin{equation}
        \wh{W}_f\leq \left(\frac{n}{B'|f|}\right)^{\fc}. \label{eq:W2_weakened}
    \end{equation}
    For $f\in [-n/B', n/B']$ the right-hand side is at least one, and hence this claim follows directly from the first claim above.  On the other hand, if $|f|\geq n/B'$, we have
    \begin{equation}
        \begin{split}
        2(B'-1)|f|/n&=2B'|f|/n-2|f|/n\\
        &\geq 2B'|f|/n-1\text{~~~(since $|f|\leq n/2$)}\\
        &\geq B'|f|/n\text{~~~(since $|f|\geq n/B'$)},
        \end{split} \nn
    \end{equation}
    and hence \eqref{eq:W2_weakened} follows from \eqref{eq:filter-ub}.
    
    Using \eqref{eq:W2_weakened}, we have
    \begin{equation}\label{eq:h-sum-tail-ub}
    \begin{split}
    \sum_{|f|\geq \frac{\lambda\cdot n}{2B'}} \wh{W}_f
        &\leq \sum_{|f| \geq \frac{\lambda\cdot n}{2B'}} \left(\frac{n}{B'|f|}\right)^F = O(\lambda^{-F+1})\cdot \frac{n}{B'}.
    \end{split}
    \end{equation}
    At the same time, for any $f \in [-\frac{n}{2B'}, \frac{n}{2B'}]$, we have
    \begin{equation*}
    \begin{split}
    \wh{W}_f&=\left|\frac{\sin(\pi (B'-1) f/n)}{(B'-1)\sin (\pi f/n)}\right|^{\fc}\\
    &\geq \left|\frac{2(B'-1) f/n}{(B'-1)\sin (\pi f/n)}\right|^{\fc}\text{~~~~(since $|\sin (\pi x)|\geq 2|x|$ for $|x|\leq 1/2$)}\\
    &\geq \left|\frac{2(B'-1) f/n}{(B'-1) \pi (f/n)}\right|^{\fc}\text{~~~~(since $|\sin (\pi x)|\leq \pi|x|$)}\\
    &= \left(\frac{2}{\pi}\right)^{\fc}.
    \end{split}
    \end{equation*}
    This means that 
    \begin{equation}\label{eq:h-sum-lb}
    \sum_{f\in \nsq} \wh{W}_f\geq \sum_{f\in [-\frac{n}{2B'}, \frac{n}{2B'}]} \wh{W}_f\geq \left(\frac{2}{\pi}\right)^{\fc} \cdot \frac{n}{B'}.
    \end{equation}
    Putting \eqref{eq:h-sum-tail-ub} together with \eqref{eq:h-sum-lb}, we get 
    \begin{equation*}
        \begin{split}
        \sum_{f\in \nsq,\, |f|\geq \frac{\lambda\cdot n}{2B'}} \wh{W}_f &= O(\lambda^{-F+1})\cdot \frac{n}{B'}\leq C'\cdot \Big(\frac{\pi}{2}\Big)^F\lambda^{-F+1}\cdot \sum_{f \in \nsq} W_f
        \end{split}
    \end{equation*}
    for an absolute constant $C'>0$. The desired claim follows since  $C' (\pi/2)^F=((C')^{1/(F-1)}\cdot (\pi/2)^{F/(F-1)}))^{F-1}\leq (C'\cdot (\pi/2)^{2})^{F-1}$ (due to the assumption that $F \ge 2$).
\end{proof}

We now fix an integer $B$, and define $\wh{G}$ by 
$$
    \wh{G}_f = \frac{1}{Z}\sum_{\Delta=-\frac{3n}{4B}}^{\frac{3n}{4B}} \wh{W}_{f-\Delta}.
$$
where $Z=\sum_{f\in \nsq} \wh{W}_f$.  By interpreting this as a convolution with a rectangle, we obtain that the inverse Fourier transform $G_t$ is obtained via the multiplication of $W_t$ with a sinc pulse.
%\begin{equation}\label{eq:g-def}
%{G}_i=\frac{1}{Z}\cdot {W}_f\cdot \left(\frac{\sin(\pi (4B/3-1) f/n)}{(4B/3-1)\sin (\pi f/n)}\right).
%\end{equation}

We proceed by showing that, upon identifying $B' = 8C B$ (where $B'$ was used in defining $\wh{W}$, and $C$ is the implied constant in Lemma \ref{cl:filter-h-prop}), this filter satisfies the claims of Lemma \ref{lem:filter_properties}.  We start with the three properties in Definition \ref{def:filterG}.
  
\begin{proofof}{[Lemma \ref{lem:filter_properties} (filter property 1)]} 
    For every $f$, we have
    $$
    \wh{G}_f=\frac{1}{Z}\cdot\sum_{\Delta=-\frac{3n}{4B}}^{\frac{3n}{4B}} \wh{W}_{f-\Delta}\leq \frac{1}{Z}\sum_{\Delta \in [n]} \wh{W}_{f-\Delta}=1.
    $$
    Similarly, the non-negativity of $\wh{G}$ follows directly from that of $\wh{W}$.
\end{proofof}

\begin{proofof}{[Lemma \ref{lem:filter_properties} (filter property 2)]}
    For every $f\in \nsq$ with $|f|\leq \frac{n}{2B}$, we have 
    \begin{equation*}
        \begin{split}
        \wh{G}_f &=\frac{1}{Z}\cdot\sum_{\Delta=-\frac{3n}{4B}}^{\frac{3n}{4B}} \wh{W}_{f-\Delta}\\
        &=1-\frac{1}{Z}\sum_{|\Delta|>\frac{3n}{4B}} \wh{W}_{f-\Delta}\\
        &\ge 1-\frac{2}{Z}\sum_{f'>\frac{n}{4B}} \wh{W}_{f'}\text{~~~~(since $|f|\leq \frac{n}{2B}$ and $W$ is symmetric)}\\
        &=1-\frac{2}{Z}\sum_{f'>\frac{B'}{2B}\cdot \frac{n}{2B'}} \wh{W}_{f'} \\
        &\ge 1-\Big(2C\cdot \frac{B}{B'}\Big)^{F-1}. \text{~~~~(by Lemma~\ref{cl:filter-h-prop})}
        \end{split}
    \end{equation*}
    Since $B'/B = 8C$ by our choice of $B'$ above, we get $\wh{G}_f \geq 1-(1/4)^{F-1}$, as required.
\end{proofof}

\begin{proofof}{[Lemma \ref{lem:filter_properties} (filter property 3)]}
    For every $f \in \nsq$ with $|f| \ge \frac{n}{B}$, we have 
    \begin{equation*}
        \begin{split}
        \wh{G}_f &=\frac{1}{Z}\cdot\sum_{\Delta=-\frac{3n}{4B}}^{\frac{3n}{4B}} \wh{W}_{f-\Delta}\\
            &\le \frac{1}{Z}\cdot\sum_{f' \,:\, |f'| \ge |f| - \frac{3n}{4B} } \wh{W}_{f'}. ~~~~\text{(by $|f| \ge \frac{n}{B}$)}
        \end{split}
    \end{equation*}
    Defining $\zeta \ge 1$ such that $|f| = (3+\zeta)\frac{n}{4B}$, this becomes
    \begin{equation*}
        \begin{split}
        \wh{G}_f &\le \frac{1}{Z}\cdot\sum_{f' \,:\,|f'| \ge \frac{\zeta n}{4B} } \wh{W}_{f'} \\
            &= \frac{1}{Z}\cdot\sum_{f' \,:\,|f'| \ge \frac{\zeta B'}{2B} \cdot \frac{n}{2B'} } \wh{W}_{f'} \\
            &\le \Big( \frac{2C B}{\zeta B'}\Big)^{F-1} \text{~~~~(by Lemma~\ref{cl:filter-h-prop})} \\
            &= \Big( \frac{1}{4 \zeta}\Big)^{F-1} \text{~~~~(since~$B' = 8CB$)}.
        \end{split}
    \end{equation*}
    Rearranging the definition of $\zeta$, we obtain $\zeta = \frac{4B |f|}{n} - 3$, and hence $\zeta \ge \frac{B|f|}{n} $ due to the fact that $|f| \ge \frac{n}{B}$.  Therefore, $\wh{G}_f \le \big( \frac{n}{4 B|f|}\big)^{F-1}$.
\end{proofof}

\begin{proofof}{[Lemma \ref{lem:filter_properties} (additional property 1)]}
    We have already shown that ${W}$ is supported on a window of length $O(F B') = O( F B )$ centered at zero.  The same holds for ${G}$ since it is obtained via a pointwise multiplication of ${W}$ with a sinc pulse.
\end{proofof}

\begin{proofof}{[Lemma \ref{lem:filter_properties} (additional property 2)]}
    Since $\wh{G}_f \in [0,1]$, the total energy across $|f| < \frac{n}{B}$ is at most $\frac{2n}{B}$.  On the other hand, we have from the third property in Definition \ref{def:filterG} that
    \begin{equation*}
        \begin{split}
            \sum_{|f| \ge \frac{n}{B}} |\wh{G}_f|^2
                & \le 2 \sum_{f \ge \frac{n}{B}} \Big(\frac{1}{4}\Big)^{2(F-1)}   \Big( \frac{n}{Bf} \Big)^{2(F-1)} \\
                &\le \frac{1}{8} \sum_{f \ge \frac{n}{B}} \Big( \frac{n}{Bf} \Big)^{2} \text{~~~~(since~$F \ge 2$)} \\
                &\le \frac{1}{8}\cdot\frac{n}{B} \sum_{f=1}^{\infty} \frac{1}{f^2} \\
                &\le \frac{n}{B} \text{~~~~(since~$\sum_{f=1}^{\infty} \frac{1}{f^2} < 8$)}.
        \end{split}
    \end{equation*}
    Combining this with the contribution from $|f| < \frac{n}{B}$ concludes the proof.
\end{proofof}

\subsection{Proof of Lemma \ref{claim:1}} \label{sec:pf_downsampling}

We are interested in the behavior of $\sum_{r \in [2k_1]} |\wh{Z}^r_{j}|^2$ for each $j$ (first part), and summed over all $j$ (second part).  We therefore begin with the following lemma, bounding this summation in terms of the signal $X$ and the filter $G$.

\begin{lem} \label{lem:pre1}
    \emph{(Initial downsampling bound)}
    For any integers $(n,k_1)$, parameter $\delta \in \big(0, \frac{1}{20}\big)$, signal $X \in \CC^n$ and its corresponding $(k_1,\delta)$-downsampling $\{Z^r\}_{r \in [2k_1]}$, the following holds for all $j \in [\frac{n}{k_1}]$:
    \begin{equation*}
    \bigg| \frac{1}{2k_1}\sum_{r \in [2k_1]} |\wh{Z}^r_{j}|^2 - \sum_{f=1}^{n} |\wh{G}_{f-k_1j}|^2 \cdot |\wh{X}_f|^2 \bigg| \le 3\delta \sum_{f=1}^{n} |\wh{G}_{f-k_1j}| \cdot |\wh{X}_f|^2. 
    \end{equation*}
\end{lem}
\begin{proof}
\textbf{Directly evaluating the sum:} Using the definition of the signals $\wh{Z}^r$ in \eqref{equation:4}, we write
\begin{align}
    \sum_{r \in [2k_1]} |\wh{Z}^r_{j}|^2
        &=\sum_{r \in [2k_1]} \bigg (\sum_{f=1}^{n} \wh{G}_{f-k_1 j} \cdot \wh{X}_f \cdot \omega_n^{a_r f}\bigg)^{\dagger} \bigg(\sum_{f'=1}^{n}\wh{G}_{f'-kj} \cdot \wh{X}_{f'} \cdot \omega_n^{a_r f'}\bigg) \nn \\
    	& =\sum_{f=1}^{n} \sum_{f'=1}^{n} \wh{G}^*_{f-k_1 j} \cdot \wh{X}^*_f \cdot \wh{G}_{f'-k_1 j} \cdot \wh{X}_{f'} \cdot \Big( \sum_{r \in [2k_1]}\omega_n^{a_r (f'-f)} \Big) \nn
\end{align}
where $(\cdot)^{\dagger}$ denotes the complex conjugate. Since $a_r = \frac{nr}{2k_1}$, the term $\sum_{r \in [2k_1]}\omega_n^{a_r (f'-f)}$ is equal to $2k_1$ if $f-f'$ is a multiple of $2k_1$ (including $f = f'$) and zero otherwise, yielding
\begin{align}
    \sum_{r \in [2k_1]} |\wh{Z}^r_{j}|^2
	&= 2k_1 \cdot \sum_{f=1}^{n}\Bigg( |\wh{G}_{f-k_1j}|^2 \cdot |\wh{X}_f|^2 + \sum_{\substack{j' \in [\frac{n}{2k_1}] \\ j' \neq 0}} \wh{G}^*_{f-k_1 j} \cdot \wh{X}^*_f \cdot \wh{G}_{f - 2k_1 j' -k_1 j} \cdot \wh{X}_{f - 2k_1 j'} \Bigg). \label{eq:energydiff3}
\end{align}
Without loss of generality, we consider $j=0$; otherwise, we can simply consider a version of $X$ shifted in frequency domain by $k_1 j$. Setting $j = 0$ in \eqref{eq:energydiff3} and applying the triangle inequality, we obtain
\begin{equation}
    \bigg| \sum_{r \in [2k_1]} |\wh{Z}^r_{0}|^2 - 2k_1 \cdot \sum_{f=1}^{n} |\wh{G}_{f}|^2 \cdot |\wh{X}_f|^2 \bigg| \le 2k_1 \cdot \sum_{\substack{|j'| \le \frac{n}{2(2k_1)} \\ j' \neq 0}}\sum_{f=1}^{n} \bigg| \wh{G}^*_{f} \cdot \wh{X}^*_f \cdot \wh{G}_{f - 2k_1 j'} \cdot \wh{X}_{f - 2k_1 j'}\bigg|. \label{eq:two_terms}
\end{equation}
\textbf{Bounding the right-hand side of \eqref{eq:two_terms}:} We write
\begin{align}
    &\sum_{\substack{|j'| \le \frac{n}{2(2k_1)} \\ j' \neq 0}}\sum_{f=1}^{n} \bigg| \wh{G}^*_{f} \cdot \wh{X}^*_f \cdot \wh{G}_{f - 2k_1 j'} \cdot \wh{X}_{f - 2k_1 j'}\bigg| \nonumber \\
    & \qquad = \sum_{\substack{|j'| \le \frac{n}{2(2k_1)} \\ j' \neq 0}}\sum_{f=1}^{n} \Big(|\wh{G}_{f}|^{1/2} \cdot |\wh{G}_{f - 2k_1 j'}|^{1/2}\Big) \cdot |\wh{G}_{f}|^{1/2} \cdot |\wh{X}^*_f| \cdot |\wh{G}_{f - 2k_1 j'}|^{1/2} \cdot |\wh{X}_{f - 2k_1 j'}|. \label{eq:secondterm3}
\end{align}
In Lemma \ref{lemm:2} below, we show that
$$|\wh{G}_{f}|^{1/2} \cdot |\wh{G}_{f - 2k_1 j'}|^{1/2} \leq \frac{(\frac{1}{2})^{F-1}}{|j'|^{(F-1)/2}}$$
for all $f \in [n]$ and all $|j'| \le \frac{n}{2(2k_1)}$ with $j' \neq 0$.  Definition \ref{def:downsampling} ensures that $\big(\frac{1}{2}\big)^{F-1} \le \delta$, and substitution into \eqref{eq:secondterm3} gives
\begin{multline}
    \sum_{\substack{|j'| \le \frac{n}{2(2k_1)} \\ j' \neq 0}}\sum_{f=1}^{n} \bigg| \wh{G}^*_{f} \cdot \wh{X}^*_f \cdot \wh{G}_{f - 2k_1 j'} \cdot \wh{X}_{f - 2k_1 j'}\bigg|
        \\ \leq \delta \cdot \sum_{\substack{|j'| \le \frac{n}{2(2k_1)} \\ j' \neq 0}} \frac{1}{|j'|^{(F-1)/2}}  \sum_{f=1}^{n} |\wh{G}_{f}|^{1/2} \cdot |\wh{X}^*_f| \cdot |\wh{G}_{f - 2k_1 j'}|^{1/2} \cdot |\wh{X}_{f - 2k_1 j'}|. \label{eq:secondterm5}
\end{multline}
Next, we apply the Cauchy-Schwarz inequality to upper bound the inner summation over $f$ above for any fixed $j' \in [\frac{n}{2k_1}]$, yielding
\begin{align}
    \sum_{f=1}^{n} |\wh{G}_{f}|^{1/2} \cdot |\wh{X}^*_f| \cdot |\wh{G}_{f - 2k_1 j'}|^{1/2} \cdot |\wh{X}_{f - 2k_1 j'}| 
    &\leq \sqrt{\sum_{f=1}^{n} |\wh{G}_{f}| \cdot |\wh{X}^*_f|^2} \cdot \sqrt{\sum_{f=1}^n |\wh{G}_{f - 2k_1 j'}| \cdot |\wh{X}_{f - 2k_1 j'}|^2} \nn \\
    &=\sum_{f=1}^{n} |\wh{G}_{f}| \cdot |\wh{X}^*_f|^2, \label{eq:secondterm7}
\end{align}
where we used the fact that $\big \{ |\wh{G}_{f - 2k_1 j'}| \cdot |\wh{X}_{f - 2k_1 j'}|^2 \big \}_{f=1}^n$ is a permutation of $\big \{ |\wh{G}_{f}| \cdot |\wh{X}^*_f|^2 \big \}_{f=1}^n$. % and therefore has the same $\ell_2$ norm.

\textbf{Wrapping up:} Substituting \eqref{eq:secondterm7} into \eqref{eq:secondterm5} gives
\begin{align}
	&\sum_{\substack{|j'| \le \frac{n}{2(2k_1)} \\ j' \neq 0}}\sum_{f=1}^{n} \bigg| \wh{G}^*_{f} \cdot \wh{X}^*_f \cdot \wh{G}_{f - 2k_1 j'} \cdot \wh{X}_{f - 2k_1 j'}\bigg|  \nn \\
	&\qquad\leq \delta \cdot \sum_{\substack{|j'| \le \frac{n}{2(2k_1)} \\ j' \neq 0}} \frac{1}{|j'|^{(F-1)/2}}  \sum_{f=1}^{n} |\wh{G}_{f}| \cdot |\wh{X}_f|^2 \nn  \\
	&\qquad\leq 3\delta \sum_{f=1}^{n} |\wh{G}_{f}| \cdot |\wh{X}_f|^2, \nn
\end{align}
where the last inequality follows from the fact that $\sum_{\substack{|j'| \le \frac{n}{2(2k_1)} \\ j' \neq 0}} \frac{1}{|j'|^{(F-1)/2}} \le 2 \sum_{j'=1}^{\infty} \frac{1}{|j'|^{(F-1)/2}}$, which is upper bounded by $3$ for $F \ge 8$, a condition guaranteed by Definition \ref{def:downsampling}. We therefore obtain the following bound from \eqref{eq:two_terms}:
\begin{equation*}
\bigg| \frac{1}{2k_1} \cdot \sum_{r \in [2k_1]} |\wh{Z}^r_{0}|^2 - \sum_{f=1}^{n} |\wh{G}_{f}|^2 \cdot |\wh{X}_f|^2 \bigg| \le 3\delta \sum_{f=1}^{n} |\wh{G}_{f}| \cdot |\wh{X}_f|^2 .
\end{equation*}
The lemma follows by recalling that the choice $j=0$ was without loss of generality, with the general case amounting to replacing $\wh{Z}_0$ by $\wh{Z}_j$ and $\wh{G}_f$ by $\wh{G}_{f-k_1j}$.
\end{proof}

In the preceding proof, we made use of the following technical result bounding the product of the filter $G$ evaluated at two locations separated by some multiple of $2k_1$.

\begin{lem}
    \emph{(Additional filter property)}
    Given $(n,k_1)$ and a parameter $F \ge 2$, if $G$ is an $(n, \frac{n}{k_1}, F)$-flat filter, then the following holds for all $f \in [n]$ and all $j\in\big[\frac{n}{k_1}\big]$ with $|j'| \le \frac{n}{2(2k_1)}$ and $ j' \neq 0$:
    $$|\wh{G}_{f}|^{1/2} \cdot |\wh{G}_{f - 2k_1j'}|^{1/2} \leq \frac{(\frac{1}{2})^{F-1}}{|j'|^{(F-1)/2}}. $$
    \label{lemm:2}
\end{lem}
\begin{proof}
    For clarity, let $f_1$ and $f_2$ denote the frequencies corresponding to $f$ and $f - 2k_1 j'$ respectively, defined in the range $(-n/2 , n/2]$ according to modulo-$n$ arithmetic.  By definition, $f_1 - f_2$ is equal to $2k_1 j'$ modulo-$n$, and since $|j'| \le \frac{n}{2(2k_1)}$, we have $|2k_1 j'| \le \frac{n}{2}$.  This immediately implies that the distance $\Delta = |f_1 - f_2|$ according to \emph{regular arithmetic} is lower bounded by the distance according to modulo-$n$ arithmetic: $\Delta \ge 2k|j'|$.
    
    Since $f_1$ and $f_2$ are at a distance $\Delta$ according to regular arithmetic, it must be the case that either $|f_1| \ge \frac{\Delta}{2}$ or $|f_2| \ge \frac{\Delta}{2}$. Moreover, since $j' \ne 0$, we have, from the above-established fact $\Delta \ge 2k_1|j'|$, that $\frac{\Delta}{2} \ge k_1$, and hence we can apply the third filter property in Definition \ref{def:filterG} to conclude that
    $ |G_{f_{\nu}}| \le \big(\frac{1}{4}\big)^{F-1} \big( \frac{2k_1}{\Delta} \big)^{F-1}$
    for either $\nu = 1$ or $\nu = 2$.  Substituting $\Delta \ge 2k_1|j'|$, upper bounding $G_{f_{\nu'}} \le 1$ (\emph{cf.}, Definition \ref{def:filterG}) for the index $\nu' \in \{1,2\}$ differing from $\nu$, and taking the square root, we obtain the desired result.
\end{proof}

We are now in a position to prove the claims of Lemma \ref{claim:1}

\begin{proofof}{first part of Lemma \ref{claim:1}}
    Recall from Lemma \ref{lem:pre1} that
    \begin{equation}
    \bigg|\frac{1}{2k_1}\sum_{r \in [2k_1]} |\wh{Z}^r_{j}|^2 - \sum_{f=1}^{n} |\wh{G}_{f-k_1j}|^2 \cdot |\wh{X}_f|^2\bigg| \le 3\delta \sum_{f=1}^{n} |\wh{G}_{f-k_1j}| \cdot |\wh{X}_f|^2. \label{eqq:14}
    \end{equation}
    We proceed by lower bounding $\sum_{f=1}^{n} |\wh{G}_{f-k_1j}|^2 \cdot |\wh{X}_f|^2$ and upper bounding  $\sum_{f=1}^{n} |\wh{G}_{f-k_1j}| \cdot |\wh{X}_f|^2$.  Starting with the former, recalling that $I_j = \big((j-1/2)k_1,(j+1/2)k_1\big] \cap \ZZ$, we have
    \begin{align}
         \sum_{f=1}^{n} |\wh{G}_{f-k_1j}|^2 \cdot |\wh{X}_f|^2 
             &\ge \sum_{f \in I_j} |\wh{G}_{f-k_1j}|^2 \cdot |\wh{X}_f|^2 \nn \\
             &\ge \bigg( 1 - \Big(\frac{1}{4}\Big)^{F-1} \bigg)^2 \|\wh{X}_{I_{j}}\|_2^2 \nn \\
             &\ge ( 1- \delta ) \|\wh{X}_{I_{j}}\|_2^2, \label{eqq:25}
    \end{align}
    where the second line is by the second filter property in Definition \ref{def:filterG}, and the third line is by the choice of $F$ in Definition \ref{def:downsampling}.
    
    Next, we upper bound  $\sum_{f=1}^{n} |\wh{G}_{f-k_1j}| \cdot |\wh{X}_f|^2$ as follows:
    \begin{equation}
         \sum_{f=1}^{n} |\wh{G}_{f-k_1j}| \cdot |\wh{X}_f|^2 \le \sum_{f \in I_j \cup I_{j-1} \cup I_{j+1}} |\wh{G}_{f-k_1j}| \cdot |\wh{X}_f|^2 + \sum_{f\in [n] \,:\, |f-k_1j| \geq \frac{3k_1}{2}} |\wh{G}_{f-k_1j}| \cdot |\wh{X}_f|^2. \label{equation:15}
    \end{equation}
    By the third property in Definition \ref{def:filterG}, the filter decays as $|\wh{G}_f| \leq (\frac{1}{4})^{F-1}(\frac{k_1}{|f|})^{F-1}$ for $|f| \ge k_1$, and therefore the second term in \eqref{equation:15} is bounded by
    \begin{align}
    \sum_{|f| \geq k_1j + \frac{3k_1}{2}} |\wh{G}_{f - k_1j}| \cdot |\wh{X}_f|^2 
        &\leq \Big(\frac{1}{4}\Big)^{F-1} \cdot \sum _{j'\in[\frac{n}{k_1}] \,:\, |j'-j|  \ge 2}\frac{\|\wh{X}_{I_{j'}}\|_2^2}{(|j' - j| - 1)^{F-1}} \nn  \\
        &\le \Big(\frac{1}{2}\Big)^{F-1} \cdot \sum_{j'\in[\frac{n}{k_1}]\backslash \{j\}}\frac{\|\wh{X}_{I_{j'}}\|_2^2}{|j' - j|^{F-1}} \nn  \\
        &\le \delta \cdot \sum _{j'\in[\frac{n}{k_1}]\backslash \{j\}} \frac{\|\wh{X}_{I_{j'}}\|_2^2}{|j' - j|^{F-1}}, \label{eqq:16}
    \end{align}
    where the second line follows from $|j'-j| -1 \ge \frac{|j'-j|}{2}$, and the third line follows since the choice of $F$ in Definition \ref{def:downsampling} ensures that $\big(\frac{1}{2}\big)^{F-1} \le \delta$.  We bound the term $\sum_{f \in I_j \cup I_{j-1} \cup I_{j+1}} |\wh{G}_{f-k_1j}| \cdot |\wh{X}_f|^2$ in \eqref{equation:15} using the first property in Definition \ref{def:filterG}, namely, $\wh{G}_f\leq 1$:
    \begin{equation}
    \sum_{f \in I_j \cup I_{j-1} \cup I_{j+1}} |\wh{G}_{f-k_1j}| \cdot |\wh{X}_f|^2 \le \|\wh{X}_{I_{j} \cup I_{j-1} \cup I_{j+1}}\|_2^2 . \label{eqq:24}
    \end{equation}
    Hence, combining \eqref{equation:15}--\eqref{eqq:24}, we obtain
    \begin{equation}
         \sum_{f=1}^{n} |\wh{G}_{f-k_1j}| \cdot |\wh{X}_f|^2 \le \|\wh{X}_{I_{j} \cup I_{j-1} \cup I_{j+1}}\|_2^2 + \delta \cdot \sum _{j'\in[\frac{n}{k_1}]\backslash \{j\}}\frac{\|\wh{X}_{I_{j'}}\|_2^2}{|j' - j|^{F-1}}. \label{eq:combined_abs}
    \end{equation}
    The first claim of the lemma follows by combining \eqref{eqq:14}, \eqref{eqq:25}, and \eqref{eq:combined_abs}.
\end{proofof}

\begin{proofof}{second part of Lemma \ref{claim:1}}
    %\textbf{Second part of lemma:} In order to prove the second claim, we recall that the upper bound on $|\wh{Z}_0^r|^2$ arising from \eqref{eqq:14} also extends to other $j \ne 0$; summing over all $j \notin \Spm$ gives
    %\begin{align}
    %\sum_{j \notin \Spm} \sum_{r \in [2k_1]} |\wh{Z}^r_{j}|^2 \le 2k_1 \sum_{j \notin \Spm} \bigg( (1 + 4\delta) \|\wh{X}_{I_{j} \cup I_{j-1} \cup I_{j+1}}\|_2^2 + 3\delta \sum_{j'\in[\frac{n}{k_1}]\backslash \{j\}} \frac{\|\wh{X}_{I_{j'}}\|_2^2}{|j'-j|^{F-1}} \bigg). \label{eq:upper_Spm_sum}
    %\end{align}
    %Note that because $j \notin \Spm$, we also have $j-1, j+1 \notin S$, and hence $\sum_{j \notin \Spm} \|\wh{X}_{I_{j} \cup I_{j-1} \cup I_{j+1}}\|_2^2 \le 3 \sum_{j \ne S} \|\wh{X}_{I_j}\|_2^2$.  Hence, and trivially expanding the sum $j \notin \Spm$ to $j \notin S$ to handle the second term, this expression is upper bounded by that given in the second claim of the Lemma.
    By following the same steps as those used to handle \eqref{equation:15}, we obtain the following analog of \eqref{eq:combined_abs} with $|\wh{G}_f|^2$ in place of $|\wh{G}_f|$:
    \begin{equation}
    \sum_{f=1}^{n} |\wh{G}_{f-k_1j}|^2 \cdot |\wh{X}_f|^2  \leq \|\wh{X}_{I_j \cup I_{j-1} \cup I_{j+1}}\|_2^2 + \delta \cdot \sum_{j'\in[\frac{n}{k_1}]\backslash \{j\}}\frac{\|\wh{X}_{I_{j'}}\|_2^2}{|j' - j|^{2(F-1)}}. \label{eqq:18}
    \end{equation}
    Combining \eqref{eqq:14}, \eqref{eq:combined_abs}, and \eqref{eqq:18}, we obtain
    \begin{align*}
        \frac{\sum_{r\in[2k_1]} |\wh{Z}^r_{j}|^2 }{2k_1} &\le \|\wh{X}_{I_j \cup I_{j-1} \cup I_{j+1}}\|_2^2 + \delta \cdot \sum_{j'\in[\frac{n}{k_1}]\backslash \{j\}}\frac{\|\wh{X}_{I_{j'}}\|_2^2}{|j'|^{2(F-1)}} \\
        &\quad + 3\delta \cdot \bigg( \|\wh{X}_{I_j \cup I_{j-1} \cup I_{j+1}}\|_2^2 + \delta \sum _{j'\in[\frac{n}{k_1}]\backslash \{j\}} \frac{\|\wh{X}_{I_{j'}}\|_2^2}{|j'-j|^{F-1}}\bigg),
    \end{align*}
    and summing over $j \in [n]$ gives
    \begin{align}
        \frac{1}{2k_1}\sum_{r \in [2k_1]} \|\wh{Z}^r\|^2 
            &\le \sum_{j \in [n]} \bigg( (1 + 3\delta) \|\wh{X}_{I_{j} \cup I_{j-1} \cup I_{j+1}}\|_2^2 + (3\delta^2+\delta) \sum_{j'\in[\frac{n}{k_1}]\backslash \{j\}} \frac{\|\wh{X}_{I_{j'}}\|_2^2}{|j'-j|^{F-1}} \bigg) \nn \\
            &= 3(1 + 3\delta) \|\wh{X}\|_2^2 + (3\delta^2+\delta) \sum_{j \in [n]} \sum_{j'\in[\frac{n}{k_1}]\backslash \{j\}} \frac{\|\wh{X}_{I_{j'}}\|_2^2}{|j'-j|^{F-1}} \label{eq:Zstep2}
    \end{align}
    The double summation is upper bounded by $\sum_{j' \in [n]} \|\wh{X}_{I_{j'}}\|_2^2 \cdot 2\sum_{\Delta=1}^{\infty} \frac{1}{\Delta^{F-1}} = 2 \|\wh{X} \|_2^2 \cdot \sum_{\Delta=1}^{\infty} \frac{1}{\Delta^{F-1}}$, which in turn is upper bounded by $3 \|\wh{X}\|_2^2 $ for $F \ge 4$, a condition guaranteed by Definition \ref{def:downsampling}.  We can therefore upper bound \eqref{eq:Zstep2} by $ \|\wh{X}\|_2^2 (3(1+3\delta) + 3(3\delta^2 + \delta))$, which is further upper bounded by $6 \|\wh{X}\|_2^2$ for $\delta \le \frac{1}{20}$, as is assumed in Definition \ref{def:downsampling}.
    
    For the lower bound, we sum the first part of the lemma over all $j$, yielding
    \begin{equation}
    \begin{split}
        \sum_{r \in [2k_1]} \|\wh{Z}^r\|^2 
            &\ge \frac{\sum_{r\in[2k_1]} |\wh{Z}^r_{j}|^2 }{2k_1}\ge ( 1-\delta) \|\wh{X}\|_2^2 - 3\delta \cdot \bigg( 3\|\wh{X}\|_2^2 + \delta \sum _{j\in[\frac{n}{k_1}]}\sum _{j'\in[\frac{n}{k_1}]\backslash \{j\}} \frac{\|\wh{X}_{I_{j'}}\|_2^2}{|j'-j|^{F-1}}\bigg). 
    \end{split} \nn
    \end{equation}
    We showed above that the double summation is upper bounded by $3\|\wh{X}\|_2^2$, yielding an lower bound of $(1-\delta - 9\delta - 3\delta^2) \|\wh{X}\|_2^2$.  This is lower bounded by $(1-12\delta)\|\wh{X}\|_2^2$ for $\delta \le \frac{1}{20}$.
\end{proofof}

\section{Omitted Proofs from Section \ref{sec:ep_sampling}} 

\subsection{Proof of Lemma \ref{lemm:11}} \label{sec:pf_active}

Note that for any $j \in \big[\frac{n}{k_1}\big]$, solving the first part of Lemma \ref{claim:1} for $\|\wh{X}_{I_j}\|_2^2$ gives
\begin{equation}
     \|\wh{X}_{I_{j}}\|_2^2 \le \frac{1}{1-\delta}\Bigg(\frac{1}{2k_1} \sum_{r \in [2k_1]} |\wh{Z}^r_{j}|^2 +  3\delta \cdot \bigg( \|\wh{X}_{I_j \cup I_{j-1} \cup I_{j+1}}\|_2^2 + \delta \sum _{j'\in[\frac{n}{k_1}]\backslash \{j\}} \frac{\|\wh{X}_{I_{j'}}\|_2^2}{|j'-j|^{F-1}} \bigg) \Bigg). \label{eq:after_lem_app}
\end{equation}
We will sum both sides over $j \in S^* \backslash \tilde{S}$; we proceed by analyzing the resulting terms.

\textbf{Second term in \eqref{eq:after_lem_app} summed over $j \in S^* \backslash \tilde{S}$:} We have
\begin{align}
    &\sum_{j \in S^* \backslash \tilde{S}} 3\delta \cdot \bigg( \|\wh{X}_{I_j \cup I_{j-1} \cup I_{j+1}}\|_2^2 + \delta \sum _{j'\in[\frac{n}{k_1}]\backslash \{j\}} \frac{\|\wh{X}_{I_{j'}}\|_2^2}{|j'-j|^{F-1}} \bigg) \nn \\
    &\qquad \le 9\delta \|\wh{X}\|_2^2 + 3\delta^2  \sum_{j \in S^* \backslash \tilde{S}} \sum _{j'\in[\frac{n}{k_1}]\backslash \{j\}} \frac{\|\wh{X}_{I_{j'}}\|_2^2}{|j'-j|^{F-1}} \nn  \\
    &\qquad \le 9\delta  \|\wh{X}\|_2^2 + 10\delta^2 \|\wh{X}\|_2^2 \le 10\delta \|\wh{X}\|_2^2, \label{eq:active_lem_final1.5}
\end{align}
where the last line follows by expanding the double summation to all $j,j' \in \big[\frac{n}{k_1}\big]$ with $j \ne j'$, noting that $2\sum_{\Delta=1}^{\infty} \frac{1}{\Delta^{F-1}} \le 2.5$ for $F \ge 4$ (a condition guaranteed by Definition \ref{def:downsampling}), and then applying the assumption $\delta \le \frac{1}{20}$. 

\textbf{First term in \eqref{eq:after_lem_app} summed over $j \in S^* \backslash \tilde{S}$:} We first rewrite the sum of squares in terms of a weighted sum of fourth moments:
\begin{align}
\sum_{j \in S^* \backslash \tilde{S}} \frac{1}{2k_1} \sum_{r \in [2k_1]} |\wh{Z}^r_{j}|^2 
    &= \frac{1}{2k_1} \sum_{r \in [2k_1]} \|\wh{Z}^r_{S^* \backslash \tilde{S}}\|_2^2 = \sum_{r \in [2k_1]} \|\wh{Z}^r\|_2 \cdot \frac{\|\wh{Z}^r_{S^* \backslash \tilde{S}}\|_2^2}{\|\wh{Z}^r\|_2} \nn \\
    &\le \frac{1}{2k_1} \sqrt{ \bigg( \sum_{r \in [2k_1]} \|\wh{Z}^r\|_2^2 \bigg) \bigg( \sum_{r \in [2k_1]} \frac{\|\wh{Z}^r_{S^* \backslash \tilde{S}}\|_2^4}{\|\wh{Z}^r\|_2^2} \bigg)}, \label{eq:active_lem_final2}
\end{align}
by Cauchy-Schwarz applied to the length-$2k_1$ vectors containing entries $\|\wh{Z}^r\|_2$ and $\frac{\|\wh{Z}^r_{S^* \backslash \tilde{S}}\|_2^2}{\|\wh{Z}^r\|_2}$.

The second summation inside the square root is upper bounded as
\begin{align}
    \sum_{r \in [2k_1]} \frac{\|\wh{Z}^r_{S^* \backslash \tilde{S}}\|_2^4}{\|\wh{Z}^r\|_2^2} 
        &\le \sum_{r \in [2k_1]} \|\wh{Z}^r_{S^*}\|_2^2 \cdot \frac{\|\wh{Z}^r_{S^* \backslash \tilde{S}}\|_2^2 }{\|\wh{Z}^r\|_2^2} \nn \\
        &\le \sum_{r \in [2k_1]} \gamma^r \cdot \frac{\|\wh{Z}^r_{S^* \backslash \tilde{S}}\|_2^2}{\|\wh{Z}^r\|_2^2} +  \sum_{r \in [2k_1]} \Big| \|\wh{Z}^r_{S^*}\|_2^2 - \gamma^r \Big|_+ \cdot \frac{\|\wh{Z}^r_{S^* \backslash \tilde{S}}\|_2^2}{\|\wh{Z}^r\|_2^2}, \label{eq:second_sum_bound}
\end{align}
where the first inequality follows since $\|\wh{Z}^r_{S^* \backslash \tilde{S}}\|_2^2\leq \|\wh{Z}^r_{S^*}\|_2^2$ and the second inequality uses $\|\wh{Z}^r_{S^*}\|_2^2 \le \gamma^r  + \big| \|\wh{Z}^r_{S^*}\|_2^2 - \gamma^r \big|_+$.  

Now observe that by definition of $\tilde{S}$ (Definition~\ref{eq:s-tilde}), for every $j \notin \tilde{S}$, we have
    $$\sum_{r \in [2k_1]} \Big( |\wh{Z}^r_j|^2 \cdot \frac{\gamma^r}{\|\wh{Z}^r\|_2^2} \Big) \leq \delta \cdot \frac{\sum_{r \in [2k_1]} \|\wh{Z}^r\|_2^2}{k_0},$$
and summing both sides over all $j \in S^* \backslash \tilde{S}$ gives
\begin{equation}
    \sum_{r \in [2k_1]} \gamma^r \cdot \frac{\|\wh{Z}^r_{S^* \backslash \tilde{S}}\|_2^2}{\|\wh{Z}^r\|_2^2} \le \frac{\delta |S^* \backslash \tilde{S}|}{k_0} \sum_{r \in [2k_1]} \|\wh{Z}^r\|_2^2 \le 10\delta \sum_{r \in [2k_1]} \|\wh{Z}^r\|_2^2, \nn
\end{equation}
since $|S^*| \le 10k_0$ by assumption.
Applying this to the first term in \eqref{eq:second_sum_bound}, as well as $\frac{\|\wh{Z}^r_{S^* \backslash \tilde{S}}\|_2^2}{\|\wh{Z}^r\|_2^2} \le 1$ for the second term, we obtain 
\begin{align}
    \sum_{r \in [2k_1]} \frac{\|\wh{Z}^r_{S^* \backslash \tilde{S}}\|_2^4}{\|\wh{Z}^r\|_2^2}  
        &\le 10\delta \sum_{r \in [2k_1]} \|\wh{Z}^r\|_2^2 + \sum_{r \in [2k_1]} \Big| \|\wh{Z}^r_{S^*}\|_2^2 - \gamma^r \Big|_+ \nn \\
        &\le 50\delta \sum_{r \in [2k_1]} \|\wh{Z}^r\|_2^2, \label{eq:b2to23bto23t}
\end{align}
where we have applied the assumption (*) of the lemma.

Finally, substituting \eqref{eq:b2to23bto23t} into \eqref{eq:active_lem_final2} yields
\begin{align}
    \sum_{j \in S^* \backslash \tilde{S}} \frac{1}{2k_1} \sum_{r \in [2k_1]} |\wh{Z}^r_{j}|^2     
        &\le \frac{1}{2k_1} \sqrt{ \bigg( \sum_{r \in [2k_1]} \|\wh{Z}^r\|_2^2 \bigg) \bigg( \sum_{r \in [2k_1]} \frac{\|\wh{Z}^r_{S^* \backslash \tilde{S}}\|_2^4}{\|\wh{Z}^r\|_2^2} \bigg)} \nn \\
     &\le \frac{1}{2k_1} \sqrt{ \bigg( \sum_{r \in [2k_1]} \|\wh{Z}^r\|_2^2 \bigg) \bigg( 50\delta \sum_{r \in [2k_1]} \|\wh{Z}^r\|_2^2 \bigg)}\text{~~~~~(by ~\eqref{eq:b2to23bto23t})} \nn \\
    &\le \frac{ \sqrt{50\delta} }{2k_1} \sum_{r \in [2k_1]} \|\wh{Z}^r\|_2^2 \le 43\sqrt{\delta}\|\wh{X}\|_2^2, \label{eq:4huiaehsiutbs}
\end{align}
where the last inequality uses the fact that $\frac{\sum_{r \in [2k_1]} \|\wh{Z}^r\|_2^2 }{2k_1}\le 6 \|X\|_2^2$ by the second part of Lemma \ref{claim:1}.  The proof is concluded by substituting \eqref{eq:active_lem_final1.5} and \eqref{eq:4huiaehsiutbs} into \eqref{eq:after_lem_app}, and using the assumption $\delta \le \frac{1}{20}$ to deduce that $\frac{1}{1-\delta}\big( 43\sqrt{\delta} + 10\delta\big) \le 100\sqrt{\delta}$.

\section{Omitted Proofs from Section \ref{sec:prelim}} 

\subsection{Proof of Lemma \ref{lem:uhat}} \label{sec:pf_uhat}

The (exact) Fourier transform of $U$, denoted by $\wh{U}^*$, is given by
\begin{align}
    \wh{U}^*_j 
        &= \frac{1}{B} \sum_{b \in [B]} U_b \omega_{B}^{-bj} \nn \\
        &= \frac{1}{n} \sum_{b \in [B]} \sum_{i'\in [\frac{n}{B}]} X_{\sigma( \Delta + b + B \cdot i')} G_{b + B \cdot i'} \omega_{B}^{-bj}  \nn \\
        &= \frac{1}{n} \sum_{i \in [n]} X_{\sigma( \Delta + i)} G_{i} \omega_{n}^{-ij n/B}, \label{eq:uhat_init}
\end{align} 
where we used the fact that $\omega_B^{(\cdot)}$ is periodic with period $B$, and then applied $\omega_B = \omega_n^{n/B}$.  We see that \eqref{eq:uhat_init} is the Fourier transform of the signal $\{X_{\sigma( \Delta + i)} G_{i}\}_{i \in [n]}$ evaluated at frequency $jn/B$, and hence, since multiplication and convolution are dual under the Fourier transform, we obtain
\begin{equation}
    \wh{U}^*_j = ( \wh{Y} \star \wh{G} )_{jn/B}, \label{eq:uhat_conv}
\end{equation}
where $Y_i = X_{\sigma( \Delta + i)}$.  Now, by standard Fourier transform properties, we have $\wh{Y}_f = \wh{X}_{\sigma^{-1} f} \omega_n^{\Delta f}$, and substitution into \eqref{eq:uhat_conv} gives
\begin{equation}
\begin{split}
    \wh{U}^*_j &= \sum_{f \in [n]} \wh{X}_{\sigma^{-1} f} \wh{G}_{j\frac{n}{B} - f } \omega_n^{\Delta f } \\
        &=\sum_{f \in [n]} \wh{X}_f \wh{G}_{\sigma f - \frac{n}{B} j } \omega_n^{\sigma \Delta f },
\end{split} \nn
\end{equation}
where we have used the assumed symmetry of $G$ about zero.

\subsection{Proof of Lemma \ref{lem:perm_property}} \label{sec:pf_perm}

For brevity, let $\Psi = \sum_{f' \ne f} | \wh{X}_{f'} |^2 \EE_{\pi}\big[ |G_{o_f(f')}|^2 \big]$ denote the left-hand side of \eqref{eq:filter_exp}.  Following the approach of \cite[Lemma 3.3]{IK14a}, we define the intervals $\Fc_t = \big(\pi(f) -\frac{n}{B}2^{t}, \pi(f) + \frac{n}{B}2^{t} \big]$ for $t = 1,\dotsc,\log_2 b$, and write
\begin{align}
        \Psi &\le \sum_{f' \ne f} |\wh{X}_{f'}|^2 \sum_{t=1}^{\log_2 B} \PP[ \pi(f') \in \Fc_t \backslash \Fc_{t-1} ] \max_{f'' \,:\, \pi(f'') \in \Fc_t \backslash \Fc_{t-1}} |\wh{G}_{o_f(f'')}|^2 \nn \\ 
            &\le \frac{4}{B} \sum_{f' \ne f} |\wh{X}_{f'}|^2 \bigg( 2 + \sum_{t=2}^{\log_2 B} 2^t \max_{f'' \,:\, \pi(f'') \in \Fc_t \backslash \Fc_{t-1}} |\wh{G}_{o_f(f'')}|^2 \bigg),  \label{eq:Psi_init}
\end{align}
where the second line follows by (i) upper bounding $\PP[ \pi(f') \in \Fc_t \backslash \Fc_{t-1} ] \le \PP[ \pi(f') \in \Fc_t ]$ and applying the approximate pairwise independence property (\emph{cf.}, Definition \ref{def:permutation}); (ii) using the fact that there are at most $\frac{n}{B} \cdot 2^{t+1}$ integers within $\Fc_t$, and applying $|\wh{G}_f| \le 1$ for the case $t=1$.

To handle the term containing $|\wh{G}_{o_f(f'')}|^2$, we use the triangle inequality to write
\begin{equation}
    \begin{split}
        |o_f(f'')| &\ge |\pi(f) - \pi(f'')| - \Big|\pi(f) - \frac{n}{b}\mathrm{round}\big(\pi(f)\frac{B}{n}\big)\Big| \\
            &\ge |\pi(f) - \pi(f'')| - \frac{n}{B}.
    \end{split} \nn
\end{equation}
For any $f''$ with $\pi(f'') \notin \Fc_{t-1}$, we have $|\pi(f) - \pi(f'')| \ge \frac{n}{B} 2^{t-1}$, and hence $|o_f(f'')| \ge \frac{n}{B} (2^{t-1} - 1)$. As a result, for $t \ge 2$, the third property in Definition \ref{def:filterG} gives 
$$\wh{G}_{o_f(f'')} \le \Big(\frac{1}{4}\Big)^{F-1} \Big( \frac{1}{2^{t-1}-1} \Big)^{F-1} \le \Big(\frac{1}{4}\Big)^{F-1} \Big(\frac{1}{2^{t-2}} \Big)^{F-1}  = \Big( \frac{1}{2^t} \Big)^{F-1},$$ 
and hence
\begin{equation}
    \begin{split}
        \sum_{t=2}^{\log_2 B} 2^t \max_{f'' \,:\, \pi(f'') \in \Fc_t \backslash \Fc_{t-1}} |\wh{G}_{o_f(f'')}|^2
            &\le \sum_{t=2}^{\infty} \Big( \frac{1}{2^t} \Big)^{2F - 1}. %\\
            %&\le \frac{1}{2^{12}} \sum_{t=2}^{\infty} (2^{-7})^t, \\
    \end{split} \nn
\end{equation}
This sum is less than $\frac{1}{2}$ for all $F \ge 2$, and hence substitution into \eqref{eq:Psi_init} gives $\Psi \le \frac{10}{B} \|\wh{X}\|_2^2$, as desired.

\subsection{Proof of Lemma \ref{lem:semi_equi}} \label{sec:pf_semi}

We use techniques resembling those used for a $(k_1,\e)$-downsampling in Section \ref{sec:overview}, but with the notable difference of using a more rapidly-decaying filter with bounded support in {\em frequency} domain.

\textbf{Choice of filter:} We let $G \in \RR^n$ be the filter used in \cite{IKP} (as opposed to that used in Definition \ref{def:filterG}), satisfying the following:
\begin{itemize}
    \item There exists an ideal filter ${G}'$ satisfying $G^{\prime}_f \in [0,1]$ for all $f$, and
    \begin{equation}
        G^{\prime}_f = 
        \begin{cases}
            1 & |f| \le \frac{n}{2k_1}\\
            0 & |f| \ge \frac{n}{k_1},
        \end{cases} \label{eq:G'}
    \end{equation}
    such that $\| {G} - {G}' \|_2 \le n^{-c}$;
    \item $\wh{G}$ is supported on a window of length $O(c k_1 \log n)$ centered at zero;
    \item Each entry of $\wh{G}$ can be computed in time $O(1)$.
\end{itemize}

\textbf{Intuition behind the proof.} Before giving the details, we provide the intuition for the proof. Recall that our goal is to compute $X^r_j$ for $|j|\leq k_0/2$, for all $r\in [2k_1]$. To do this, we first note that  $X^0_j$, for $|j|\leq k_0/2$ (i.e., for only one value of $r$, namely $0$), can be computed via a reduction to standard semi-equispaced FFT (Lemma~\ref{lem:semi_equi_std}) on an input signal of length $2n/k_1$. To achieve this, consider the signal $X\cdot G$ aliased to length $2n/k_1$, which is close to $X$ on all points $j$ such that $|j|\leq n/(2k_1)$. In order to compute $X^0_j$ for $|j|\leq k_0/2$, it essentially suffices (modulo boundary issues; see below) to calculate $(X\cdot G)_j$ for $|j|\leq k_0/2$. We show below that this can be achieved using Lemma~\ref{lem:semi_equi_std}, because multiplication followed by aliasing are dual to convolution and subsampling: the input $(k_0, k_1)$-block sparse signal of length $n$ can be naturally mapped to an $O(k_0 \log n)$-sparse signal in a reduced space with $\approx n/k_1$ points, in which standard techniques (Lemma~\ref{lem:semi_equi_std}) can be applied. 

This intuition only shows how to compute the values of $X^r_j$ for $r=0$ and $|j|\leq k_0/2$, but we need the values for all $r\in [2k_1]$. As we show below, the regular structure of the set of shifts that we are interested in allows us use the standard FFT on a suitably defined set of length-$2k_1$ signals, without increasing the runtime by a $k_1$ factor. It is interesting to note that our runtime is $O(\log n)$ worse than the runtime of Lemma~\ref{lem:semi_equi_std} due to the two-level nature of our scheme; this is for reasons similar to the $\log^d n$ scaling of runtime of high-dimensional semi-equispaced FFT (e.g.~\cite{GhaziHIKPS13, K16}). 

We now give the formal proof of the lemma.

\textbf{Computing a convolved signal:} Here we show that we can efficiently compute the values $\wh{Y}_j^r = (\wh{X}^r \star \wh{G})_{\frac{k_1}{2} j}$ at all $j \in \big[ \frac{2n}{k_1} \big]$ where it is non-zero, for all values of $r\in [2k_1]$.   We will later show that applying Lemma \ref{lem:semi_equi_std} to these signals (as a function of $j$) gives accurate estimates of the desired values of ${X}$.

Note that in the definition of $\wh{Y}_j^r$, each non-zero block is convolved with a filter of support $O(c k_1 \log n)$, and so contributes to at most $O(c \log n)$ values of $j$.  Since there are $k_0$ non-zero blocks, there are $O(c k_0\log n)$ values of $j$ for which the result is non-zero. 

The procedure is as follows:
\begin{enumerate}
    \item For all $j$ such that $\wh{Y}_j^r$ may be non-zero ($O(c k_0 \log n)$ in total), compute $\widetilde{Y}_j^b = \frac{k_1}{2}\sum_{l=1}^{\frac{n}{2k_1}} \wh{X}_{b+2k_1 l} \wh{G}_{\frac{k_1}{2} j - (b+2k_1 l) }$ for $b \in [2k_1]$.  That is, alias the signal $\{\wh{X}_f \wh{G}_{\frac{k_1}{2} j - f}\}$ down to length $2k_1$, and normalize by $\frac{2}{k_1}$ (for later convenience). Since $\wh{G}$ is supported on an interval of length $O(c k_1\log n)$, this can be done in time $O(c \log n)$ per entry, for a total of $O(c k_1 \log n)$ per $j$ value, or $O(c^2 k_0 k_1 \log^2 n)$ overall.
    \item Compute the length-$2k_1$ inverse FFT of $\widetilde{Y}_j = (\widetilde{Y}_j^1,\dotsc,\widetilde{Y}_j^{2k_1})$ to obtain $\wh{Y}_j \in \CC^{2k_1}$. This can be done in time $O(k_1 \log(1+k_1))$ per $j$ value, or $O(k_0 k_1 \log(1+k_1)  \log n)$ overall.
\end{enumerate}

We now show that $\wh{Y}_j^r =\frac{k_1}{2} ({X}^r \star {G})_{\frac{k_1}{2} j}$ for $r=1,\ldots, 2k_1$. By the definition of the inverse Fourier transform, we have
\begin{equation}
\begin{split}
    \wh{Y}_j^r &= \sum_{b=1}^{2k_1} \widetilde{Y}_j^b \omega_{2k_1}^{rb} \\
        &= \frac{k_1}{2}\sum_{b=1}^{2k_1} \sum_{l=1}^{\frac{n}{2k_1}} \wh{X}_{b+2k_1 l} \wh{G}_{\frac{k_1}{2} j - (b+2k_1 l) } \omega_{2k_1}^{rb} \\
        &= \frac{k_1}{2}\sum_{f=1}^{n} \wh{X}_{f} \wh{G}_{\frac{k_1}{2}j - f}\omega_{2k_1}^{rf} \\
        &= \frac{k_1}{2}\sum_{f=1}^{n} \wh{X}_{f} \wh{G}_{\frac{k_1}{2}j - f}\omega_n^{ rf\cdot \frac{n}{2k_1}} \\
        &= \frac{k_1}{2}\sum_{f=1}^{n} \wh{X}_{f}^r \wh{G}_{\frac{k_1}{2}j - f} \text{~~~(since $X^r_{(\cdot)}=X_{(\cdot) +\frac{n r}{2k_1}}$ by Definition~\ref{def:downsampling})},
\end{split} \nn
\end{equation}
where the second line is by the definition of $\widetilde{Y}^b$, the third by the periodicity of $\omega_{2k_1}$, and the fifth since translation and phase shifting are dual under the Fourier transform.  Hence, $\wh{Y}_j^r =\frac{k_1}{2} (\wh{X}^r \star \wh{G})_{\frac{k_1}{2} j}$.

\textbf{Applying the standard semi-equispaced FFT:} For $r \in [2k_1]$, define $\wh{Y}^r = (\wh{Y}^r_1,\dotsc,\wh{Y}^r_{n/k_1})$.  We have already established that the support of each $\wh{Y}^r$ is a subset of a set having size at most $k' = O(c k_0 \log n)$.  We can therefore apply Lemma \ref{lem:semi_equi_std} with $\zeta=n^{-(c+1)}$ to conclude that we can evaluate ${Y}^{r}_j$ for $|j| \le \frac{k'}{2}$ satisfying
    \begin{equation}
        |{Y}^{r}_j - {Y}^{* r}_j| \le n^{-(c+1)} \|{Y}^r\|_2,  \label{eq:semi_appr1}
    \end{equation}
where ${Y}^{* r}$ is the exact inverse Fourier transform of $\wh{Y}^r$.  Moreover, this can be done in time $O\big(k' \log \frac{n/k_1}{n^{-(c+1)}}\big) = O(c^2 k_0 \log^2 n)$ per $r$ value, or $O(c^2 k_0 k_1 \log^2 n)$ overall.

\noindent \textbf{Proof of first part of lemma:} It remains to show that the above procedure produces estimates of the desired $X$ values of the form \eqref{eq:semi_equi}.  

Recall that $\wh{Y}_j^r = \frac{k_1}{2} (\wh{X}^r \star \wh{G})_{\frac{k_1}{2} j}$. By the convolution theorem and the fact that subsampling and aliasing are dual (e.g., see Appendix \ref{sec:pf_uhat}), the inverse Fourier transform of $\wh{Y}^r$ satisfies the following when $|j| \le \frac{n}{k_1}$:
\begin{align}
    {Y}^r_j 
        &= \sum_{i \in [\frac{k_1}{2}]} ({G} \cdot {X}^r)_{j+\frac{2n}{k_1} i} \nn \\
        &= {G}_j {X}^r_j + \sum_{i \in [\frac{k_1}{2}], \, i \ne 0} ({G} \cdot {X}^r)_{j+\frac{2n}{k_1} i} \nn \\
        &=  \bigg( {G}^{\prime}_j {X}^r_j + \sum_{i \in [\frac{k_1}{2}], \, i \ne 0} ({G}^{\prime} \cdot {X}^r)_{j+\frac{2n}{k_1} i} \bigg) \pm \|G - G'\|_2 \|X\|_2 \nn\\
        &= {X}^r_j \pm n^{-c} \|{X}\|_2, \label{eq:semi_appr2}
\end{align}
where the last line follows from the definition of $G'$ in \eqref{eq:G'} and the assumption $\|G - G'\|_2 \le n^{-c}$.

Combining \eqref{eq:semi_appr1} and \eqref{eq:semi_appr2} and using the triangle inequality, we obtain
\begin{equation}
    |{Y}^{r}_j - {X}^{r}_j| \le n^{-(c+1)} \|{Y}^r\|_2 + n^{-c} \|{X}\|_2. \nn
\end{equation}
Since we have already shown that we can efficiently compute ${Y}^{r}_j$ for $|j| \le \frac{k'}{2}$ with $k' = O(c k_0 \log n)$, it only remains to show that $\|{Y}^r\|_2 \le n \|{X}\|^2$.  To do this, we use the first line of \eqref{eq:semi_appr2} to write
\begin{align}
    |{Y}_j^r| 
        &\le \sum_{i \in [\frac{k_1}{2}]} |{G}_{j+\frac{2n}{k_1} i}| \cdot |{X}^r_{j+\frac{2n}{k_1} i}| \nn \\
        &\le 2\sum_{i \in [\frac{k_1}{2}]} |{X}^r_{j+\frac{2n}{k_1} i}| \nn \\
        &\le \sqrt{2k_1 \sum_{i \in [\frac{k_1}{2}]} |{X}^r_{j+\frac{2n}{k_1} i}|^2 }, \label{eq:se_abs_bound}
\end{align}
where the first line is the triangle inequality, the second line follows since the first filter assumption above ensures that $|G_j| \le 2$ for all $j$, and the third line follows since the squared $\ell_1$-norm is upper bounded by the squared $\ell_2$-norm times the vector length.

Squaring both sides of \eqref{eq:se_abs_bound} and summing over all $j$ gives $\|{Y}^r\|_2^2 \le 2k_1 \|{X}\|^2 \le n^2 \|{X}\|^2$ (under the trivial assumption $n \ge 2$), thus completing the proof.

\textbf{Proof of second part of lemma:} 
In the proof of the first part, we applied Lemma \ref{lem:semi_equi_std} to signals of length $\frac{2n}{k_1}$.  It follows directly from the arguments in \cite[Cor.~12.2]{IKP} that since we can approximate the entries of $X^r_j$ for all $|j| \le \frac{k_0}{2}$, we can do the same for all $j$ equaling $\sigma j' + b$ modulo-$\frac{2n}{k_1}$ for some $|j'| \le \frac{k_0}{2}$.  Specifically, this follows since the multiplication by $\sigma$ and shift by $b$ simply amounts to a phase shift and a linear change of variables $f\to \sigma^{-1} f$ in frequency domain, both of which can be done in constant time.
    
However, the second part of the lemma regards indices modulo-$\frac{n}{k_1}$, as opposed to modulo-$\frac{2n}{k_1}$.  To handle the former, we note that for any integer $a$, we either have $a\text{ mod } \frac{n}{k_1} = a\text{ mod } \frac{2n}{k_1}$ or  $a\text{ mod } \frac{n}{k_1} = \big(a + \frac{n}{k_1}\big) \text{ mod } \frac{2n}{k_1}$.  Hence, we obtain the desired result by simply performing two calls to the first part, one with a universal shift of $\frac{n}{k_1}$.

\subsection{Proof of Lemma \ref{lem:hash2bins}} \label{sec:pf_hash2bins}

    \subsubsection{First Part}

    Since $U_X$ is computed according to $X$ itself in Algorithm \ref{alg:hash2bins}, we only need to compute the error in $U_{\chi}$.  
    
    In the definition of hashing in Definition \ref{def:hashing}, since $G$ has support $O(FB)$, we see that the values of $X$ used correspond to a permutation of an interval having length $k' = O(FB)$.  We can therefore apply the second part of Lemma \ref{lem:semi_equi_std} with sparsity $k'$ and parameter $\zeta = n^{-c'}$ for some $c' > 0$, ensuring an $\ell_{\infty}$-guarantee of $n^{-c'} \| \chi \|_2$ for the signal $\chi$.  
    
    Since $\wh{U}$ is computed from these values using \eqref{eq:hashing} followed by the FFT, we readily obtain via the relation $\|v\|_{\infty} \le \|v\|_2 \le \sqrt{m}\|v\|_{\infty}$ (for $v \in \CC^m$) and Parseval's theorem that $\wh{U}$ has an $\ell_{\infty}$-guarantee of $n^{-(c' - O(1))} \|\chi\|_2$, which can be made to equal $n^{-c}\|\wh{\chi}\|_2$ by choosing $c' = c + O(1)$.
    
    \paragraph{Sample complexity and runtime:} The only operation that consumes samples from the signal $X$ is the hashing operation applied to $X$. From the definition of hashing in Definition \ref{def:hashing}, and the fact that the filter $G$ has support $O(FB)$, we find that the sample complexity is also $O(FB)$.  
    
    The runtime is dominated by the application of the semi-equispaced FFT, which is  $O(c F (\|\wh{\chi}\|_0+B) \log n)$ by Lemma \ref{lem:semi_equi_std}.  In particular, this dominates the $O(B\log B)$ time to perform the FFT in Algorithm \ref{alg:hash2bins}, and the hashing operation, whose time complexity is the same as the sample complexity.

    \subsubsection{Second Part}
%    
%    We first determine which values of the downsampled signals are used in the hashing on line \ref{line:hash} of Algorithm \ref{alg:hash2bins}.  Recall the definition of a $(k_1,\delta)$-downsampling of a signal $X$ from \eqref{equation:4}:
%    \begin{equation}
%        Z^r_j= \frac{1}{k_1} \sum_{i\in [k_1]}(G \cdot X^r)_{j+\frac{n}{k_1} \cdot i}, \, j \in \Big[\frac{n}{k_1}\Big]. \nonumber 
%    \end{equation}
%    We see that the values of $X$ used are those at the indices where we use $Z^r$, plus all multiples of $k_1$.  However, from Lemma \ref{lem:semi_equi}, when we call \textsc{SemiEquiInverseBlockFFT}, we get the same permuted indices of $X^r$ for all $r \in [2k_1]$, and hence the multiples of $k_1$ come for free!  Thus, it suffices to call \textsc{SemiEquiInverseBlockFFT} with block sparsity $(O(\max\{k_0,B\}),k_1)$, since $\wh{\chi}$ is $(O(k_0),k_1)$-block sparse by assumption. 
    
    Recall the definition of a $(k_1,\delta)$-downsampling of a signal $X$ from \eqref{equation:4}:
    \begin{equation}
        Z^r_j= \frac{1}{k_1} \sum_{i\in [k_1]}(G \cdot X^r)_{j+\frac{n}{k_1} \cdot i}, \quad j \in \Big[\frac{n}{k_1}\Big], \, r\in[2k_1]. \nonumber 
    \end{equation}
    In order to compute the $\big(\frac{n}{k_1},B^r,G^r,\sigma,\Delta\big)$-hashing of $\wh{Z}^r$ (\emph{cf.}, Definition \ref{def:hashing}), we use the samples of $Z^r_j$ at the locations $j=\sigma(j'+\Delta) \mod \frac{n}{k_1}$ for $|j'| \le FB^r$; this is because $G^r$ is supported on $[-FB^r,+FB^r]$. Note that $FB^r$ is further upper bounded by $O(F\Bmax)$.  
    
    We claim that in the second part of Lemma \ref{lem:semi_equi}, it suffices to set the sparsity level to $O(F\Bmax+k_0)$.  To see this, first note that $k_0$ is added in accordance with Remark \ref{rem:k0} and the fact that $\wh{\chi}$ is $(k_0,k_1)$-block sparse.  Moreover, note that $\wh{Z}^r$ has length $\frac{n}{k_1}$, and one sample of $Z^r_j$ can be computed from $X^r_{j+i\frac{n}{k_1}}=X^{r+2i}_j$ for $|i|\le F$ as per Definition~\ref{def:downsampling} and the fact that the filter $G$ is supported on $[-F\frac{n}{k_1},+F\frac{n}{k_1}]$. Therefore, all we need is $X^{r'}_{j'}$ for each $r' \in [2k_1]$ and for all $j'=\sigma(j+\Delta) \mod \frac{n}{k_1}$ with $|j| \le F\Bmax$.
    
    Applying the second part of Lemma \ref{lem:semi_equi} with sparsity $O(F \Bmax + k_0)$ and parameter $\zeta = n^{-c'}$ for some $c' > 0$, ensuring an $\ell_{\infty}$-guarantee of $2n^{-c'} \| \chi \|_2$ on the computed values of $\chi$.  By an analogous argument to the first case, this implies an $\ell_{\infty}$-guarantee of $n^{-c} \| \chi \|_2$ on the FFT $\wh{U}^r$ of the hashing of $Z_{\chi}^r$, with $c = c' + O(1)$.

    \paragraph{Sample complexity and runtime:} We take $O(FB^r)$ samples of the $r$-th downsampled signal each time we do the hashing, separately for each $r \in [2k_1]$. By Lemma \ref{lem:downsamp-cost-unit-access}, accessing a single sample of $Z^r_{X}$ costs us $O(\log \frac{1}{\delta})$ samples of $X$.  Hence, the  sample complexity is $O\big(F \sum_{r\in[2k_1]}B^r \log \frac{1}{\delta}\big)$.  
    
    We now turn to the runtime.  By Lemma \ref{lem:semi_equi}, the call to $\textsc{SemiEquiInverseBlockFFT}$ with  $O(F\Bmax+k_0)$ in place of $k_0$ takes time $O\big(c^2 (F\Bmax + k_0) k_1 \log^2 n )$.  The hashing operation's runtime matches its sample complexity, and since we have assumed $\delta \ge \frac{1}{n}$, its contribution is dominated by the preceding term.

\subsection{Proof of Lemma \ref{lemm:14}} \label{sec:pf_energy_est}

\textbf{First part of lemma:} We start with the following bound on the expression inside the expectation:
\begin{equation}
\Big| \|\wh{Y}_S\|_2^2 - \|\wh{U}^*\|_2^2 \Big|_+ \le \Big| \|\wh{Y}_S\|_2^2 - \|\wh{U}^*_{h(S \backslash S_{coll})}\|_2^2 \Big|_+ \nn
\end{equation}
where $h(S)=\{ h(j) : j \in S\}$ with $h(j) = \mathrm{round}\big(\pi(j)\frac{B}{m}\big)$, denoting the bucket into which element $j$ hashes.  We define $\Scoll$ to be a subset of $S$ containing the elements that collide with each other, i.e., $\Scoll = \{ j \in S \, | \, h(j) \cap h(S\backslash \{j\}) \neq \emptyset \}$, yielding
\begin{align}
    \Big| \|\wh{Y}_S\|_2^2 - \|\wh{U}^*\|_2^2 \Big|_+ 
    &\le \Big| \sum_{j \in S} |\wh{Y}_j|^2 - \sum_{b \in h(S \backslash \Scoll)}|\wh{U}^*_b|^2 \Big|_+ \nn \\
    &= \Big| \sum_{j \in \Scoll} |\wh{Y}_j|^2 + \sum_{j \in S\backslash \Scoll} \Big( |\wh{Y}_j|^2-|\wh{U}^*_{h(j)}|^2 \Big) \Big|_+ \nn  \\
    &\le \sum_{j \in \Scoll} |\wh{Y}_j|^2 + \sum_{j \in S} \Big| |\wh{Y}_j|^2 - |\wh{U}^*_{h(j)}|^2 \Big|_+,  \label{eq94}
\end{align}
where the final line follows from the inequality $[a + b]_+ \le |a| + [b]_+$.

\textbf{Bounding the first term in \eqref{eq94}:} We start by evaluating the expected value of the term corresponding to $\Scoll$ over the random permutation $\pi$:
\begin{align}
\EE_{\pi} \Big[ \sum_{j\in \Scoll} |\wh{Y}_j|^2 \Big]
	&\le \EE_{\pi} \Big[\sum_{j\in S} |\wh{Y}_j|^2 \Ic[j \in \Scoll] \Big] \nn  \\ 
	&\le \sum_{j\in S} |\wh{Y}_j|^2 \sum_{j'\in S \backslash \{j\}} \PP[h(j)=h(j')] \nn  \\    
    &\le \sum_{j\in S} \sum_{j'\in S} |\wh{Y}_j|^2 \frac{4}{B}  \nn \\
    &= \frac{4|S|}{B} \sum_{j\in S} |\wh{Y}_j|^2, \label{eq:Scoll_term}
\end{align}
where the second line follows from the union bound, and the third line follows since $\pi$ is approximately pairwise independent as per Definition \ref{def:permutation}. 

\textbf{Bounding the second term in \eqref{eq94}:} We apply Lemma \ref{lem:uhat} to obtain $\wh{U}^*_{h(j)} = \sum_{j' \in [m] } \wh{Y}_{j'} \wh{H}_{o_j(j')} \omega_m^{\sigma \Delta j'}$ with $o_j(j')=\pi(j')-h(j)\frac{m}{B}$.  We write this as $\wh{U}^*_{h(j)} = \wh{Y}_{j} \wh{H}_{o_j(j)} \omega_m^{\sigma \Delta j} + \err_j$ with $\err_j := \sum_{j' \in [m] \backslash \{j\}} \wh{Y}_{j'} \wh{H}_{o_j(j')} \omega_m^{\sigma \Delta j'}$, yielding
\begin{align}
    \sum_{j\in S} \Big|  |\wh{Y}_j|^2 - |\wh{U}^*_{h(j)}|^2 \Big|_+  
        &\le \sum_{j\in S} \Big| |\wh{Y}_j|^2  - |\wh{Y}_{j} \wh{H}_{o_j(j)} \omega_m^{\sigma \Delta j} + \err_j|^2 \Big| \nn \\
        &\le  \sum_{j\in S} \bigg( \big| |\wh{Y}_j|^2  - |\wh{Y}_j H_{o_j(j)}|^2\big| + |\err_j|^2 + 2|\err_j|\cdot |\wh{Y}_j \wh{H}_{o_j(j)}| \bigg)  \label{eq:err_j_init} 
\end{align}
by $|\xi|_+ \le |\xi|$ and the triangle inequality.  We have by definition that $|o_j(j)| \le \frac{m}{2B}$, and hence item 2 in Definition \ref{def:filterG} yields $H_{o_j(j)} \ge 1 - \big(\frac{1}{4}\big)^{F'-1}$, which in turn implies $H_{o_j(j)}^2 \ge 1 - 2\big(\frac{1}{4}\big)^{F'-1}$.  Combining this with $H_f \le 1$ from item 1 in Definition \ref{def:filterG}, we can weaken \eqref{eq:err_j_init} to
\begin{align}
    \sum_{j\in S} \Big|  |\wh{Y}_j|^2 - |\wh{U}^*_{h(j)}|^2 \Big|_+ 
    &\le \sum_{j\in S} \Big( 2|\err_j| \cdot |\wh{Y}_j| + |\err_j|^2 + 2\Big(\frac{1}{4}\Big)^{F'-1} |\wh{Y}_j|^2 \Big). \label{eqq:93}
\end{align}
We proceed by bounding the expected value of $|\err_j|^2$.  We first take the expectation over $\Delta$, using Parseval's theorem to write
\begin{align}
    \EE_\Delta [|\err_j|^2] 
    &= \sum_{j' \in [m] \backslash \{j\}} |\wh{Y}_{j'}|^2 |\wh{H}_{o_j(j')}|^2. \nn
\end{align}
Taking the expectation over $\pi$, we obtain
\begin{align}
    \EE_{\Delta,\pi} [|\err_j|^2] 
    &= \EE_{\pi}\Big[\sum_{j' \in [m] \backslash \{j\}} |\wh{Y}_{j'}|^2 |\wh{H}_{o_j(j')}|^2 \Big] \nn\\
    &= \sum_{j' \in [m] \backslash \{j\}} |\wh{Y}_{j'}|^2 \EE_{\pi}[ |\wh{H}_{o_j(j')}|^2\nn\\
    &\le \frac{10}{B} \sum_{j' \in [m] \backslash \{j\}} |\wh{Y}_{j'}|^2 \le \frac{10}{B} \|\wh{Y}\|_2^2. \nn
\end{align}
where the final line follows from Lemma \ref{lem:perm_property}.

Substituting into \eqref{eqq:93}, and using Jensen's inequality to write $\EE[|\err_j|] \le \sqrt{ \EE[|\err_j|^2] }$, we obtain
\begin{align}
    \EE_{\Delta,\pi} \Big[ \sum_{j\in S} \Big|  |\wh{Y}_j|^2 - |\wh{U}^*_{h(j)}|^2 \Big|_+ \Big] 
    &\le 2 \sum_{j\in S} \sqrt{\frac{10}{B} \|\wh{Y}\|_2^2} \cdot |\wh{Y}_j| + \frac{10}{B} \sum_{j\in S} \|\wh{Y}\|_2^2 + 2\Big(\frac{1}{4}\Big)^{F'-1} \sum_{j\in S}|\wh{Y}_j|^2 \nn \\
    &\le 10\sqrt{\frac{|S|}{B}}\|\wh{Y}\|_2^2 + \bigg(10\frac{|S|}{B} + 2\delta^2 \bigg) \|\wh{Y}\|_2^2, \label{eqq:102}
\end{align}
where the second line follows from the fact that $\|v\|_1 \le \sqrt{|S|}\|v\|_2$ for any $v \in \CC^{|S|}$, as well as $\Big(\frac{1}{4}\Big)^{F'-1} \le \delta$ by the choice of $F'$.  The claim follows by substituting \eqref{eq:Scoll_term} and \eqref{eqq:102} into \eqref{eq94}.

\textbf{Second part of lemma:} By the definition of $\wh{U}^*$ ({\em cf.}, Definition \ref{def:hashing}), we have
\begin{align}
\EE_{\Delta} \Big[ \|\wh{U}^*\|_2^2 \Big]
&= \EE_{\Delta} \Big[ \sum_{b \in [B]} \Big| \sum_{j \in [m]} \wh{Y}_{j} \wh{H}_{\pi(j)-b\frac{m}{B}} \omega_m^{ \Delta j} \Big|^2 \Big] \nn \\
&= \sum_{b \in [B]} \sum_{j \in [m]} |\wh{Y}_{j}|^2 |\wh{H}_{\pi(j)-b\frac{m}{B}}|^2 \nn
\end{align}
by Parseval. Taking the expectation with respect to $\pi$, we obtain
\begin{align}
\EE_{\Delta,\pi} \Big[ \|\wh{U}^*\|_2^2 \Big] 
&= \sum_{b \in [B]} \EE_{\pi}\Big[\sum_{j \in [m]} |\wh{Y}_{j}|^2 |\wh{H}_{\pi(j)-b\frac{m}{B}}|^2 \Big] \nn \\
&= \sum_{b \in [B]} \sum_{j \in [m]} |\wh{Y}_{j}|^2 \EE_{\pi}\Big[ |\wh{H}_{\pi(j)-b\frac{m}{B}}|^2 \Big] \nonumber\\
&\le \sum_{b \in [B]} \frac{3}{B} \sum_{j \in [m]} |\wh{Y}_{j}|^2 = 3\|\wh{Y}\|_2^2. \nn
\end{align}
where the final line follows by noting that $\pi(j)-b\frac{m}{B}$ is uniformly distributed over $[m]$, and applying the second part of Lemma \ref{lem:filter_properties}.

\section{Omitted Proofs from Section \ref{sec:full_alg}}

\subsection{Proof of Lemma \ref{lem:reduceSNR}} \label{sec:pf_reduce_snr}

\paragraph{Note on $\frac{1}{\poly(n)}$ assumptions in lemmas:} Throughout the proof, we apply Lemmas \ref{lem:multi_locate}, \ref{lem:prune}, and \ref{lem:estimate}.  The first of these assumes that $\wh{\chi}_0$ uniformly distributed over an arbitrarily length-$\Omega\big(\frac{\|\wh{\chi}\|^2}{\poly(n)}\big)$ interval, and the latter two use the assumption $\|\wh{X} - \wh{\chi}\|_2 \ge \frac{1}{\poly(n)}\|\wh{\chi}\|_2^2$.  

We argue that these assumptions are trivial and can be ignored.  To see this, we apply a minor technical modification to the algorithm as follows.  Suppose the implied exponent to the $\poly(n)$ notation is $c'$.  By adding a noise term to $\wh{\chi}_0$ on each iteration uniform in $[-n^{-c'+10} \|\wh{\chi}\|^2,n^{c'+10} \|\wh{\chi}\|^2]$, we immediately satisfy the first assumption above, and we also find that the probability of $\|\wh{X} - \wh{\chi}\|_2 < \frac{1}{\poly(n)}\|\wh{\chi}\|_2^2$  is at most $O(n^{-10})$, and the additional error in the estimate is $O(n^{c'+10} \|\wh{\chi}\|^2)$.  Since we only do $O(\log\SNR') = O(n)$ iterations (by the assumption $\SNR' = O(\poly (n))$), this does not affect the result because the accumulated noise added to $\wh{\chi}_0$ which we denote by $\err(\wh{\chi}_0)$, does not exceed  $\frac{\|\wh{X}\|_2^2}{\poly(n)}$ which by the final assumption of the lemma implies that $\err(\wh{\chi}_0) \le \nu^2$.

\paragraph{Overview of the proof:} We introduce the \emph{approximate support} set of the input signal $\wh{X}$, given by the union of the top $k_0$ blocks of the signal and the blocks whose energy is more than the tail noise level:
\begin{equation}
S_0 := \Big \{ j \in \Big[\frac{n}{k_1}\Big] \, : \, \|\wh{X}_{I_j}\|_2^2 \geq  \mu^2 \Big \} \cup \bigg( \argmin_{\substack{S \subset [\frac{n}{k_1}] \\  |S| = k_0}} \sum_{j \in [\frac{n}{k_1}] \backslash S} \|\wh{X}_{I_j}\|_2^2\bigg). \label{eq:def_S0}
\end{equation}
From the definition of $\mu^2$ in Definition \ref{def:SNR}, we readily obtain $|S_0| \le 2k_0$. For each $t = 1, 2, ..., T$, define the set $S_t$ as
\begin{equation}
S_t = S_{t-1} \cup L'_t, \nonumber
\end{equation}
where $L'_t$ is the output of \textsc{PruneLocation} at iteration $t$ of \textsc{ReduceSNR}. $S_t$ contains the set of the head elements of $\wh{X}$, plus every element that is modified by the algorithm so far. % which can be both potentially large in the residual. 

We prove by induction on the iteration number $t=1,\dotsc,T$ that there exist events $\Ec_0 \supseteq \Ec_{1} \supseteq ... \supseteq \Ec_{T}$ such that conditioned on $\Ec_t$, the following conditions hold true:
\begin{enumerate}[label={\bf \alph*.}]
	\item $|S_t| \le 2k_0 + \frac{tk_0}{T}$;
	\item $\|\wh{\chi}^{(t)}_{I_j}\|_2^2 = 0$ for all $j \in \big[\frac{n}{k_1}\big] \backslash S_t$;
	\item $\|\wh{X} - \wh{\chi}^{(t)}\|_2^2 \le 99 \cdot \SNR' (k_0 \nu^2) / 2^{t}$;
\end{enumerate}
and for each $t \le T$, we have $\PP [ \Ec_{t+1} | \Ec_t ] \ge 1-\frac{1}{10T}$.

\paragraph{Base case of the induction:} We have already deduced that $|S_0| \le 2k_0$ and defined $\wh{\chi}^{(0)}=0$, and we have $\|\wh{X} - \wh{\chi}^{(0)}\|_2^2 = \|\wh{X}\|_2^2 \le \SNR' \cdot (k_0 \mu^2) / 2^{0}$ by the two assumptions of the lemma. Hence, we can let $\Ec_0$ be the trivial event satisfying $\PP[\Ec_0] = 1$.

\paragraph {Inductive step:} We seek to define an event $\Ec_{t+1}$ that occurs with probability at least $1-\frac{1}{10T}$ conditioned on $\Ec_{t}$, and such that the induction hypotheses \textbf{a}, \textbf{b}, and \textbf{c} are satisfied for $t+1$ conditioned on $\Ec_{t+1}$. To do this, we will introduce three events $\Eloct$, $\Eprunet$, and $\Eestt$, and set $\Ec_{t+1} = \Eloct \cap \Eprunet \cap \Eestt \cap \Ec_t$. Throughout the following, we let $\delta$, $\theta$, and $p$ be chosen as in Algorithm \ref{alg:final}

\paragraph {Success event associated with  \textsc{MultiBlockLocate}:} Let $\Eloct$ be the event of having a successful run of \textsc{MultiBlockLocate}($X, \wh{\chi}^{(t)}, n, k_1, k_0 , \delta, p$) at iteration $t+1$ of the algorithm, meaning the following conditions on the output $L$:
\begin{gather}
      	|L| \le C \cdot T \cdot \frac{k_0}{\delta}  \log \frac{k_0}{\delta}  \log^3 \frac{1}{\delta p} \label{eq:list-size-before-pruning} \\
      	\sum_{j \in S_t \backslash L} \| (\wh{X}-\wh{\chi}^{(t)})_{I_j} \|_2^2 \leq  0.1 \|\wh{X}-\wh{\chi}^{(t)}\|_2^2, \label{eqq:129}
\end{gather}
where $C$ is a constant to be specified shortly, for a small enough $\delta$. To show this, we invoke Lemma \ref{lem:multi_locate} with $S^*=S_t$. Note that inductive hypothesis (\textbf{a}) implies $|S_t| \le (2+\frac{t}{\log\SNR'})k_0 \le 3k_0$. By the first part of Lemma \ref{lem:multi_locate}, we have $\EE \big[ |L| \big] \le C' \frac{k_0}{\delta} \log \frac{k_0}{\delta}\log\frac{1}{ p} \log^2 \frac{1}{\delta p}$ for an absolute constant $C'$, and hence \eqref{eq:list-size-before-pruning} follows with $C = 100C'$ and probability at least $1 - \frac{1}{100T}$, by Markov's inequality.

By the second part of Lemma \ref{lem:multi_locate}, \eqref{eqq:129} holds with probability at least $1-p$ provided that $\delta \le \big(\frac{0.1}{200}\big)^2$, so by the union bound, the event $\Eloct$ occurs with probability at least $\PP[ \Eloct | \Ec_{t}] \ge 1- p - \frac{1}{100T}$.

\paragraph {Success event associated with \textsc{PruneLocation}:} Let $\Eprunet$ be the event of having a successful run of \textsc{PruneLocation}($X, \wh{\chi}^{(t)}, L, k_0, k_1 , \delta, p, n, \theta$) at iteration $t+1$ of the algorithm, meaning the following conditions on the output $L'$:
    \begin{gather}
   		|L' \backslash S_t| \le \frac{k_0}{T} \label{eq:131} \\
   		\sum_{j \in [\frac{n}{k_1}] \backslash L'} \| (\wh{X}-\wh{\chi}^{(t)})_{I_j} \|_2^2
       		\leq 0.2 \|\wh{X}-\wh{\chi}^{(t)}\|_2^2 + k_0( \mu^2 + 33\nu^2 \SNR'/2^{t+1}). \label{eq:132}
	\end{gather}
\paragraph {The probability of \eqref{eq:131} holding:} In order to bound $|L' \backslash S_t|$, first recall that the set $\Stail$, defined in Lemma \ref{lem:prune} part (\textbf{a}), has the following form:
   	$$\Stail = \Big \{ j \in \Big[\frac{n}{k_1}\Big] \, : \, \|(\wh{X} - \wh{\chi}^{(t)})_{I_j}\|_2 \leq \sqrt{\theta} - \sqrt{\frac{\delta}{k_0}}\|\wh{X} - \wh{\chi}^{(t)}\|_2 \Big \}.$$
By substituting $\theta=10 \cdot 2^{-(t+1)}\cdot \nu^2 (\SNR')$ and using $\|\wh{X} - \wh{\chi}^{(t)}\|_2^2 \le 99 \cdot \SNR' (k_0 \nu^2) / 2^{t}$ from part \textbf{c} of the inductive hypothesis, we have 
\begin{align*}
    & \sqrt{\theta} - \sqrt{\frac{\delta}{k_0}}\|\wh{X} - \wh{\chi}^{(t)}\|_2 \\
    &\qquad \ge \sqrt{10 \cdot 2^{-(t+1)}\cdot \nu^2 (\SNR')} - \sqrt{\frac{\delta}{k_0}} \sqrt{ 99 \cdot \SNR' (k_0 \nu^2) / 2^{t} } \\
    &\qquad \ge \sqrt{9 \cdot \nu^2 (\SNR')/ 2^{t+1}},
\end{align*}
where the last inequality holds when $\delta$ is sufficiently small.  Hence,
\begin{equation}
    \Stail \supseteq \Big \{ j \in \Big[\frac{n}{k_1}\Big] \, : \, \|(\wh{X} - \wh{\chi}^{(t)})_{I_j}\|_2^2 \leq 9 \cdot \nu^2 (\SNR')/ 2^{t+1} \Big \}. \label{eq:Stail_bound}
\end{equation}
Now, to prove that \eqref{eq:131} holds with high probability, we write
\begin{equation}
|L' \backslash S_t| = |(L' \cap \Stail) \backslash S_t| + |L' \backslash (\Stail \cup S_t)|. \label{eq:listsize}
\end{equation}
To upper bound the first term, note that by the first part of Lemma \ref{lem:prune}, we have
   	$$\EE \Big[ \big| L' \cap \Stail \big| \Big] \le  \delta p \cdot |L|,$$
and hence by Markov's inequality, the following holds with probability at least $1-\frac{1}{100T}$:
\begin{align}
    \big| (L' \cap \Stail) \backslash S_t \big| 
        &\le \big| L' \cap \Stail \big|   \nn \\
        &\le 100 T \delta p \cdot |L|  \nn \\
        &\le 100 T \delta p \cdot C  T  \frac{k_0}{\delta}  \log \frac{k_0}{\delta}  \log^3 \frac{1}{\delta p}  \nn \\
        &= 100 C T^2  p \cdot k_0 \log \frac{k_0}{\delta}  \log^3 \frac{1}{\delta p}  \nn \\
        &= \frac{100 C \delta \cdot k_0  \log^3 \frac{1}{\delta p}}{\log \frac{k_0}{\delta} \cdot \log^2 \SNR'}, \nn
\end{align}
where the third line follows from \eqref{eq:list-size-before-pruning} (we condition on $\Eloct$), and the fifth line follows from $T = \log\SNR'$ and the choice $p = \frac{\delta}{\log^2 \frac{k_0}{\delta} \log^4 \SNR'}$ in Algorithm \ref{alg:final}.  Again using this choice of $p$, we claim that $\frac{100 C \delta \log^3 \frac{1}{\delta p}}{\log \frac{k_0}{\delta} \cdot \log \SNR'} \le 1$ for sufficiently small $\delta$ regardless of the values $(k_0,\SNR')$; this is because the dependence of $1/p$ on $k_0$ and $\SNR'$ is logarithmic, so in the numerator contains $\log^3\log k_0$ and $\log^3\log \SNR'$ while the denominator contains $\log k_0$ and $\log^2 \SNR'$. Hence,
\begin{equation}
   	\big| (L' \cap \Stail) \backslash S_t \big| \le \frac{k_0}{\log \SNR'} \nn
\end{equation}
with probability at least $1-\frac{1}{100T}$.

We now show that the second term in \eqref{eq:listsize} is zero, by showing that $\Stail \cup S_t = [\frac{n}{k_1}]$. To see this, note that the term $\nu^2 (\SNR')/ 2^{t+1}$ in the bound on $\Stail$ in \eqref{eq:Stail_bound} satisfies
\begin{equation}
   	\nu^2 (\SNR')/ 2^{t+1} \ge \frac{1}{2} \nu^2 \ge \frac{1}{2}\mu^2, \label{eq:nu_bound}
\end{equation}
by applying $t \le T = \log\SNR'$, followed by the first assumption of the lemma.  Hence,
\begin{equation}
\Stail \backslash S_t \supseteq \bigg\{ j \in \Big[\frac{n}{k_1}\Big] \big\backslash S_t \, : \, \|(\wh{X} - \wh{\chi}^{(t)})_{I_j}\|_2^2 \leq 4 \mu^2 \bigg\}. \nn
\end{equation}
By part \textbf{b} of the inductive hypothesis, we have $\|(\wh{X} - \wh{\chi}^{(t)})_{I_j}\|_2^2 = \|\wh{X}_{I_j}\|_2^2$ for all $j \notin S_t$, and hence
\begin{equation}
\Stail \backslash S_t \supseteq \bigg\{ j \in \Big[\frac{n}{k_1}\Big] \big\backslash S_t \, : \, \|\wh{X}_{I_j}\|_2^2 \leq 4 \mu^2 \bigg\}. \nn
\end{equation}
But from \eqref{eq:def_S0}, we know that $S_0$ (and hence $S_t$) contains all $j$ with $\|\wh{X}_{I_j}\|_2^2 > 4 \mu^2$, so we obtain $\Stail \backslash S_t \supset \big[\frac{n}{k_1}\big] \backslash S_t$, and hence $\Stail \cup S_t = (\Stail \backslash S_t) \cup S_t = \big[\frac{n}{k_1}\big] $, as required.

Therefore, the probability of \eqref{eq:131} holding conditioned on $\Ec_{t}$ and $\Eloct$ is at least $1-\frac{1}{100T}$.

\paragraph {The probability of \eqref{eq:132} holding:} To show \eqref{eq:132}, we use the second part of Lemma \ref{lem:prune}.  The set $\Shead$ therein is defined as
   	$$\Shead = \Big \{ j \in \Big[\frac{n}{k_1}\Big] \, : \, \|(\wh{X} - \wh{\chi}^{(t)})_{I_j}\|_2 \geq \sqrt{\theta} + \sqrt{\frac{\delta}{k_0}} \|\wh{X} - \wh{\chi}^{(t)}\|_2 \Big \}.$$
By substituting $\theta=10 \cdot 2^{-(t+1)}\cdot \nu^2 (\SNR')$ and using $\|\wh{X} - \wh{\chi}^{(t)}\|_2^2 \le 99 \cdot \SNR' (k_0 \nu^2) / 2^{t}$ from part \textbf{c} of the inductive hypothesis, we have
\begin{align*}
    &\sqrt{\theta} + \sqrt{\frac{\delta}{k_0}}\|\wh{X} - \wh{\chi}\|_2 \\
    &\qquad = \sqrt{10 \cdot 2^{-(t+1)}\cdot \nu^2 (\SNR')} + \sqrt{\frac{\delta}{k_0}} \sqrt{99 \cdot \SNR' (k_0 \nu^2) / 2^{t}} \\
    & \qquad \le \sqrt{11 \cdot \nu^2 (\SNR')/ 2^{t+1}}
\end{align*}
for sufficiently small $\delta$, and hence
\begin{equation}
   	\Shead \supseteq \Big \{ j \in \Big[\frac{n}{k_1}\Big] \, : \, \|(\wh{X} - \wh{\chi})_{I_j}\|_2^2 \geq 11 \cdot \nu^2 (\SNR')/ 2^{t+1} \Big \}. \label{shead}
\end{equation}
Next, we write
\begin{align}
\sum_{j \in [\frac{n}{k_1}] \backslash L'} \| (\wh{X}-\wh{\chi}^{(t)})_{I_j} \|_2^2 
%&= \| (\wh{X}-\wh{\chi}^{(t)})_{S_t \backslash L'} \|_2^2 + \| (\wh{X}-\wh{\chi}^{(t)})_{[\frac{n}{k_1}] \backslash (L' \cup S_t)} \|_2^2 \\
%&= \| (\wh{X}-\wh{\chi}^{(t)})_{(S_t \cap \Shead) \backslash L'} \|_2^2 + \| (\wh{X}-\wh{\chi}^{(t)})_{S_t \backslash (\Shead \cup L') } \|_2^2 \\
%&+ \| (\wh{X}-\wh{\chi}^{(t)})_{[\frac{n}{k_1}] \backslash (L' \cup S_t)} \|_2^2 \\
&= \sum_{j \in (S_t \cap \Shead \cap L) \backslash L'}  \| (\wh{X}-\wh{\chi}^{(t)})_{I_j}\|_2^2 +  \sum_{j \in (S_t \cap \Shead) \backslash (L' \cup L)} \| (\wh{X}-\wh{\chi}^{(t)})_{I_j} \|_2^2 \nonumber \\
&\quad + \sum_{j \in S_t \backslash (\Shead \cup L') } \| (\wh{X}-\wh{\chi}^{(t)})_{I_j} \|_2^2 + \sum_{j \in [\frac{n}{k_1}] \backslash (L' \cup S_t)} \| (\wh{X}-\wh{\chi}^{(t)})_{I_j} \|_2^2, \label{eq:four_terms}
\end{align}
and we proceed by upper bounding the four terms.

\paragraph {Bounding the first term in \eqref{eq:four_terms}:} By the second part of Lemma \ref{lem:prune} and the use of Markov, we have
\begin{equation}
\sum_{j \in (S_t \cap \Shead \cap L) \backslash L'}\| (\wh{X}-\wh{\chi}^{(t)})_{I_j} \|_2^2 \leq  \delta \sum_{j \in L \cap \Shead}\|(\wh{X}-\wh{\chi}^{(t)})_{I_j}\|_2^2 \le \delta \|\wh{X}-\wh{\chi}^{(t)}\|_2^2 \nn
\end{equation}
with probability at least $1-p$.

\paragraph {Bounding the second term in \eqref{eq:four_terms}:} Conditioned on $\Eloct$, we have
\begin{equation}
\begin{split}
\sum_{j \in (S_t \cap \Shead) \backslash (L \cup L') } \| (\wh{X}-\wh{\chi}^{(t)})_{I_j} \|_2^2 
\le \sum_{j \in S_t \backslash L} \| (\wh{X}-\wh{\chi}^{(t)})_{I_j} \|_2^2
\le 0.1 \|\wh{X}-\wh{\chi}^{(t)}\|_2^2,
\end{split} \nn
\end{equation}
where we have applied \eqref{eqq:129}.

\paragraph {Bounding the third term in \eqref{eq:four_terms}:} We have
\begin{equation}
\begin{split}
\sum_{j \in S_t \backslash (\Shead \cup L') } \| (\wh{X}-\wh{\chi}^{(t)})_{I_j} \|_2^2 
&\le \sum_{j \in S_t \backslash \Shead }\| (\wh{X}-\wh{\chi}^{(t)})_{I_j} \|_2^2 \\
&\le | S_t \backslash \Shead | \cdot \max_{j \in  S_t \backslash \Shead } \|(\wh{X}-\wh{\chi}^{(t)})_{I_j}\|^2_{2} \\
&\le |S_t| (11 \cdot \nu^2 \SNR'/ 2^{t+1})
\end{split} \nn
\end{equation}
by \eqref{shead}. Part \textbf{a} of the inductive hypothesis implies that $|S_t| \le 3k_0$, and hence
\begin{equation}
\begin{split}
\sum_{j \in (L \cap S_t) \backslash (\Shead \cup L') }\| (\wh{X}-\wh{\chi}^{(t)})_{I_j} \|_2^2 
&\le 33k_0 \cdot \nu^2 \SNR'/ 2^{t+1}.
\end{split} \nn
\end{equation}

\paragraph {Bounding the fourth term in \eqref{eq:four_terms}:}
\begin{equation}
\begin{split}
\sum_{j \in [\frac{n}{k_1}] \backslash (L' \cup S_t)}\| (\wh{X}-\wh{\chi}^{(t)})_{I_j} \|_2^2
& \le \sum_{[\frac{n}{k_1}] \backslash S_t} \| (\wh{X}-\wh{\chi}^{(t)})_{I_j} \|_2^2\\
&= \sum_{[\frac{n}{k_1}] \backslash S_t} \| \wh{X}_{I_j} \|_2^2 \le k_0\mu^2,
\end{split} \nn
\end{equation}
where the equality follows from part \textbf{b} of the inductive hypothesis, and the final step holds since $S_t$ contains all elements with $\|\wh{X}_{I_j}\|_2^2 \geq  \mu^2$ (\emph{cf.}, \eqref{eq:def_S0}).

Adding the above four contributions and applying the union bound, we find that conditioned on $\Ec_{t}$ and $\Eloct$, \eqref{eq:132} holds with probability at least $\PP[\Eprunet | \Ec_{t} \cap \Eloct] \ge 1 - p - \frac{1}{100T}$, provided that $\delta$ is a sufficiently small constant ($\delta \le 0.1$).

\paragraph {Success event associated with \textsc{EstimateValues}:} Let $\Eestt$ be the event of having a successful run of \textsc{EstimateValues}($X, \wh{\chi}^{(t)}, L', k_0, k_1 , \delta, p$) at iteration $t+1$ of the algorithm conditioned on $\Ec_{t}$, meaning the following conditions on the output signal $W$:
\begin{gather}
W_f = 0 \text{ for all }f \notin \Fc  \nonumber  \\
\sum_{j \in L'} \|(\wh{X}-\wh{\chi}^{(t)} - W)_{I_j}\|_2^2 \le \delta \| \wh{X}-\wh{\chi}^{(t)} \|_2^2, \label{eq58xx}
\end{gather}
where $\Fc$ contains the frequencies within the blocks indexed by $L'$.
By Lemma \ref{lem:estimate} and the fact that $|L'| \le 3k_0$ conditioned on $\Eprunet$, $\Eloct$, and $\Ec_{t}$, it immediately follows that $\Eestt$ occurs with probability at least $\PP[\Eestt | \Eprunet \cap \Eloct \cap \Ec_{t}] \ge 1- p$.

\paragraph{Combining the events:} We can now wrap everything up as follows:
\begin{equation}
\begin{split}
\PP \big[ \Ec_{t+1} \big| \Ec_t \big] 
&= \PP \big[ \Eloct \cap \Eprunet \cap \Eestt \cap \Ec_t \big| \Ec_t \big]\\
&= \PP \big[ \Eestt \big| \Eloct \cap \Eprunet \cap \Ec_t \big] \PP \big[ \Eprunet \big| \Eloct \cap \Ec_t \big] \PP \big[ \Eloct \big| \Ec_t \big].
\end{split} \nn
\end{equation}
% Note that by assumption, $\Ec_T \subseteq \Ec_{T-1} \subseteq \dotsc \subseteq \Ec_1 \subset \Ec_0$, and hence when we condition on $\Ec_t$, we also condition on every $\Ec_{t'}$, $t'<t$.  
Substituting the probability bounds into the above equation, we have
\begin{equation}
\PP \big[ \Ec_{t+1} \big| \Ec_t \big] \ge 1- 3p-\frac{2}{100T} \ge 1-\frac{1}{20T}, \nn
\end{equation}
by the choice of $p$ in Algorithm \ref{alg:final} along with $T = \log \SNR$.

Now we show that the event $\Ec_{t+1} = \Eloct \cap \Eprunet \cap \Eestt \cap \Ec_t$ implies the induction hypothesis. Conditioned on $\Eprunet \cap \Ec_t$, we have \eqref{eq:131}, which immediately gives part \textbf{a}. Conditioned on $\Eloct \cap \Eestt$, from the definition $S_{t+1} = S_t \cup L'$, part \textbf{b} of the inductive hypothesis follows from the fact that only elements in $L'$ are updated. Finally, conditioned on $\Ec_t \cap \Eprunet \cap \Eestt$, we have
\begin{align}
\|\wh{X} - \wh{\chi}^{(t+1)}\|_2^2
&= \sum_{j \in L'} \|(\wh{X} - \wh{\chi}^{(t+1)})_{I_j}\|_2^2 + \sum_{j \in [\frac{n}{k_1}] \backslash L'}\|(\wh{X} - \wh{\chi}^{(t+1)})_{I_j}\|_2^2 \nn \\
&= \sum_{j \in L'} \|(\wh{X} - \wh{\chi}^{(t)} - W)_{I_j}\|_2^2 + \sum_{j \in [\frac{n}{k_1}] \backslash L'} \|(\wh{X} - \wh{\chi}^{(t+1)})_{I_j}\|_2^2 \nn \\
&\le (0.2 + \delta) \|\wh{X}-\wh{\chi}^{(t)}\|_2^2 + k_0( \mu^2 + 33\nu^2 \SNR'/2^{t+1}) \nn \\
&\le 99 \nu^2 k_0 \SNR'/2^{t+1}, \nn
\end{align} 
where the second line holds since $W$ is non-zero only for the blocks indexed by $L'$, the third line follows from \eqref{eq:132} and \eqref{eq58xx}, and the last line holds for sufficiently small $\delta$ from part \textbf{c} of the induction hypothesis, and the upper bound $\mu^2 \le 2\nu^2(\SNR')/2^{t+1}$ given in \eqref{eq:nu_bound}.

 The first part of the lemma now follows from a union bound over the $T$ iterations and the fact that $\err(\wh{\chi}_0) \le \nu^2$, and by noting that the three parts of the induction hypothesis immediately yield the two claims therein. We conclude by analyzing the sample complexity and runtime.
 
\paragraph{Sample complexity:} We have from Lemma \ref{lem:multi_locate} that the \emph{expected} sample complexity of \textsc{MultiBlockLocate} in a given iteration is $O^*\big(\frac{k_0}{\delta} \log(1+k_0) \log n + \frac{k_0 k_1}{\delta^2} \big)$, and multiplying by the number  $T = O(\log \SNR')$ of iterations gives a total of $O^*\big(\log\SNR' \big(\frac{k_0}{\delta} \log(1+k_0) \log n + \frac{k_0 k_1}{\delta^2} \big)\big)$.  Hence, by Markov's inequality, this is also the total sample complexity across all calls to \textsc{MultiBlockLocate} with probability at least $1 - \frac{1}{100}$; this probability can be combined with the union bound that we applied above.  Since $\delta = \Omega(1)$, the above sample complexity simplifies to $O^*(k_0 \log(1+k_0) \log\SNR' \log n+ k_0 k_1\log\SNR')$.

By Lemma \ref{lem:prune}, the sample complexity of \textsc{PruneLocation} is $O(\frac{k_0 k_1}{\delta} \log \frac{1}{\delta p} \log \frac{1}{\delta})$, and by Lemma \ref{lem:estimate}, the sample complexity of \textsc{EstimateValues} is $O(\frac{k_0 k_1}{\delta} \log \frac{1}{p} \log\frac{1}{\delta})$.  Substituting the choices of $\delta$ and $p$, these behave as $O^*(k_0 k_1)$ per iteration, or  $O^*(k_0 k_1 \log\SNR')$ overall. 

\paragraph{Runtime:} 
By Lemma \ref{lem:multi_locate}, the expected runtime of \textsc{MultiBlockLocate} in a given iteration is $O^*\big( \frac{k_0}{\delta} \log(1+k_0)  \log^2 n + \frac{k_0 k_1}{\delta^2} \log^2 n + \frac{k_0 k_1}{\delta} \log^3 n \big)$.  Moreover, by Lemma \ref{lem:prune}, the runtime of  \textsc{PruneLocation} as a function of $|L|$ is $O(\frac{k_0 k_1}{\delta} \log \frac{1}{\delta p} \log \frac{1}{\delta} \log n + k_1 \cdot |L| \log\frac{1}{\delta p})$, and substituting $\EE \big[ |L| \big] = O \big( \frac{k_0}{\delta} \log \frac{k_0}{\delta}  \log\frac{1}{ p} \log^2 \frac{1}{\delta p} \big) $ from Lemma \ref{lem:multi_locate}, this becomes  $O^*\big(\frac{k_0 k_1}{\delta} \log n\big)$ in expectation, by absorbing the $\log\frac{1}{\delta}$ and $\log\frac{1}{p}$ factors into the $O^*(\cdot)$ notation.

Summing the preceding per-iteration expected runtimes, multiplying by the number of iterations $T$, and substituting the choices of $T$, $p$ and $\delta$, we find that their combined expectation is $O^*( k_0 \log k_0 \log\SNR' \log^2 n  + k_0 k_1 \log\SNR' \log^3 n  )$.  Hence, by Markov's inequality, this is also the total sample complexity across all calls to \textsc{MultiBlockLocate} and  \textsc{PruneLocation} with probability at least $1 - \frac{1}{100}$. Applying the union bound over this failure event, $\bar \Ec_T$, and the $1/\poly(n)$ probability failure event arising from random perturbations of $\wh{\chi}_0$, we obtain the required bound of $0.9$ on the success probability.

By Lemma \ref{lem:estimate}, the runtime of \textsc{EstimateValues} is $O(\frac{k_0 k_1}{\delta} \log \frac{1}{p} \log\frac{1}{\delta} \log n + k_1 \cdot |L'|\log \frac{1}{p})$, which behaves as $O(\frac{k_0 k_1}{\delta} \log \frac{1}{p} \log\frac{1}{\delta} \log n)$ conditioned on $\Eprunet \cap \Ec_t$ (see \eqref{eq:131} and recall that $|S_t| \le 3k_0$).  By our choices of $p$ and $\delta$, this simplifies to $O^*(k_0 k_1 \log n)$ per iteration, or  $O^*(k_0 k_1 \log\SNR' \log n)$ overall.

\subsection{Proof of Lemma \ref{lem:constSNR}} \label{sec:pf_constant_snr}

The proof resembles that of Lemma \ref{lem:reduceSNR}, but is generally simpler, and has some differing details.  We provide the details for completeness.

\paragraph{Note on $\frac{1}{\poly(n)}$ assumptions in lemmas:} We use an analogous argument to the start of Section \ref{sec:pf_reduce_snr} to handle the assumptions $|\wh{X}_0 - \wh{\chi}_0|_2 \ge \frac{1}{\poly(n)}\|\wh{\chi}\|_2^2$ and $\|\wh{X} - \wh{\chi}\|_2 \ge \frac{1}{\poly(n)}\|\wh{\chi}\|_2^2$ in Lemmas \ref{lem:multi_locate}, \ref{lem:prune}, and \ref{lem:estimate}.  Specifically, we add a noise term to $\wh{\chi}_0$ uniform in $[-n^{-c'+10} \|\wh{\chi}\|^2,n^{c'+10} \|\wh{\chi}\|^2]$. This does not affect the result, since the noise added to $\wh{\chi}_0$ which we denote by $\err(\wh{\chi}_0)$, does not exceed  $\frac{\|\wh{X}\|_2^2}{\poly(n)}$ which by the assumptions of the lemma implies that $\err(\wh{\chi}_0) \le \e \nu^2$.

\paragraph{Overview of the proof:} We first introduce the \emph{approximate support} set of the input signal $\wh{X} - \wh{\chi}$, given by the top $10k_0$ blocks of the signal:
\begin{align}
	S_0 & := \argmin_{\substack{S \subset [\frac{n}{k_1}] \\ |S| = 10 k_0}} \sum_{j \notin S} \|(\wh{X} - \wh{\chi})_{I_j}\|_2^2 \label{eq:S0}
\end{align}
We also introduce another set indexing blocks whose energy is sufficiently large:
\begin{equation}
S_{\e} = \Big \{ j \in \Big[\frac{n}{k_1}\Big] \, : \, \|(\wh{X}-\wh{\chi})_{I_j}\|_2^2 \geq  \e \frac{\Err^2 (\wh{X} - \wh{\chi} , 10k_0 , k_1)}{k_0} \Big \} \cup S_0. \label{eq:S_eps}
\end{equation}
It readily follows from this definition and Definition \ref{def:SNR} that $|S_{\e} \backslash S_0| \le k_0/\e$.

The function calls three other primitives, and below, we show that each of them succeeds with high probability by introducing suitable success events.  Throughout, we let $\theta$, $p$, and $\eta$ be as chosen in Algorithm \ref{alg:final}

\paragraph {Success event of the location primitive:} Let $\Eloc$ be the event of having a successful run of \textsc{MultiBlockLocate}($X, \wh{\chi}, k_1, k_0 , n, \e^2, p$), meaning the following conditions on the output $L$:
\begin{gather}
	|L| \le C \frac{k_0}{\e^2} \log \frac{k_0}{\e^2}  \log^3 \frac{1}{\e^2 p} \label{eq:list-size-before-pruning-constsnr} \\
	\sum_{j \in S_0 \backslash L} \| (\wh{X}-\wh{\chi})_{I_j} \|_2^2 \leq  200 \e \|\wh{X}-\wh{\chi}\|_2^2, \label{eqq:67}
\end{gather}
where $C$ is a constant to be specified shortly. To verify these conditions, we invoke Lemma \ref{lem:multi_locate} with $S^*=S_0$. By the first part of Lemma \ref{lem:multi_locate}, we have $\EE \big[ |L| \big] \le C' \frac{k_0}{\e^2} \log \frac{k_0}{\e} \log^3 \frac{1}{\e p}$ for an absolute constant $C'$, and hence \eqref{eq:list-size-before-pruning-constsnr} follows with $C = 100C'$ and probability at least $1 - \frac{1}{100}$, by Markov's inequality.

By the second part of Lemma \ref{lem:multi_locate} with $\delta = \epsilon^2$, \eqref{eqq:67} holds with probability at least $1-p$, so by the union bound, the event $\Eloc$ occurs with probability at least $1- p - \frac{1}{100}$.

\paragraph {Success event of the pruning primitive:} 
Let $\Eprune$ be the event of having a successful run of \textsc{PruneLocation}($X, \wh{\chi}, L, n, k_0, k_1 , \e, p, \theta$), meaning the following conditions on the output $L'$:
\begin{gather}
	|L' \backslash S_0| \le \frac{2k_0}{\e} \label{eq:68} \\
	\sum_{j \in [\frac{n}{k_1}] \backslash L'} \| (\wh{X}-\wh{\chi})_{I_j} \|_2^2
	\leq 300\e \|\wh{X}-\wh{\chi}\|_2^2 + 6000\e \nu^2 k_0 + \Err^2 (\wh{X} - \wh{\chi} , 10k_0 , k_1). \label{eq:69}
\end{gather}

\paragraph {Bounding the probability of \eqref{eq:68}:} In order to bound $|L' \backslash S_0|$, first note that the set $\Stail$, defined in Lemma \ref{lem:prune} part (\textbf{a}), has the following form:
$$\Stail = \Big \{ j \in \Big[\frac{n}{k_1}\Big] \, : \, \|(\wh{X} - \wh{\chi})_{I_j}\|_2 \leq \sqrt{\theta} - \sqrt{\frac{\e}{k_0}}\|\wh{X} - \wh{\chi}\|_2 \Big \}.$$
By substituting $\theta=200 \cdot \e \nu^2$ (\emph{cf.}, Algorithm \ref{alg:final}) and using the assumption $\|\wh{X}- \wh{\chi}\|_2^2 \le 100k_0 \nu^2$ in the lemma, we have 
\begin{align}
	\sqrt{\theta} - \sqrt{\frac{\e}{k_0}}\|\wh{X} - \wh{\chi}\|_2 &= \sqrt{200 \cdot \e \nu^2 } - \sqrt{\frac{\e}{k_0}}\|\wh{X} - \wh{\chi}\|_2 \nonumber\\
	&\ge \sqrt{200 \cdot \e \nu^2 } -\sqrt{100 \cdot \e \nu^2 } \nonumber \\
	&\ge \sqrt{16 \cdot \e \nu^2}.
\end{align}
   	
Hence,
$$\Stail \supseteq \Big \{ j \in \Big[\frac{n}{k_1}\Big] \, : \, \|(\wh{X} - \wh{\chi}^{(t)})_{I_j}\|_2^2 \leq 16 \e \cdot \nu^2 \Big \}.$$
Now, to prove that \eqref{eq:68} holds with high probability, we write
\begin{equation}
\begin{split}
|L' \backslash S_0| 
&= |(L' \cap S_{\e}) \backslash S_0| + |L' \backslash (S_0 \cup S_{\e})|\\ 
&\le |(L' \cap S_{\e}) \backslash S_0| + |(L' \cap \Stail) \backslash S_{\e}| + |L' \backslash (\Stail \cup S_{\e})|. \label{eqq:70}
\end{split}
\end{equation}
We first upper bound the first term as follows:
$$|(L' \cap S_{\e}) \backslash S_0| \le |S_{\e} \backslash S_0| \le k_0 / \e,$$ 
which follows directly from the definition of $S_{\e}$.
To upper bound the second term in \eqref{eqq:70}, note that by Lemma \ref{lem:prune} part (\textbf{a}) with $\delta = \epsilon$,
   	$$\EE \Big[ \big| L' \cap \Stail \big| \Big] \le  \e p \cdot |L|,$$
and hence by Markov's inequality, the following holds with probability at least $1-\frac{1}{100}$:
	\begin{align}
	\big| (L' \cap \Stail) \backslash S_{\e} \big|
	&\le \big| L' \cap \Stail \big|  \nn \\
	&\le 100 \e p \cdot |L| \nn \\
	&\le 100  \e p \cdot C   \frac{k_0}{\e}  \log \frac{k_0}{\e}  \log^3 \frac{1}{\e p} \nn \\
	&= 100 C p \cdot k_0 \log \frac{k_0}{\e}  \log^3 \frac{1}{\e p} \nn \\
	&= \frac{100 C \eta \e \cdot k_0  \log^3 \frac{1}{\e p}}{\log \frac{k_0}{\e} } ,\nn
	\end{align}
	where the third line follows from \eqref{eq:list-size-before-pruning-constsnr} (we condition on $\Eloc$), and the fifth line follows from and the choice $p = \frac{ \eta \e } {\log^2 \frac{k_0}{\e} }$ in Algorithm \ref{alg:final}.  Again using this choice of $p$, we claim that $\frac{100 C \eta \e \log^3 \frac{1}{\e p}}{\log \frac{k_0}{\delta}} \le 1$ for sufficiently small $\eta$ regardless of the value of $k_0$; this is because the dependence of $1/p$ on $k_0$ is logarithmic, so the numerator contains $\log^3\log k_0$, while the denominator contains $\log k_0$ which means that the ratio is upper bounded and can be made arbitrarily small by choosing a small enough constant $\eta$. Hence
	\begin{equation}
	\big| (L' \cap \Stail) \backslash S_{\e} \big| \le k_0 \nonumber
	\end{equation}
	with probability at least $1-\frac{1}{100}$.

We now show that the third term in \eqref{eqq:70} is zero, by showing that $\Stail \cup S_{\e} = [\frac{n}{k_1}]$. To see this, note that the term $\nu^2$ in the definition of $\Stail$ is more than $\frac{\Err^2 (\wh{X} - \wh{\chi} , 10k_0 , k_1)}{k_0}$ by the first assumption of the lemma, and hence
\begin{equation}
\Stail \backslash S_{\e} \supset \Big\{ j \in \Big[\frac{n}{k_1}\Big] \big\backslash S_{\e} \, : \, \|(\wh{X} - \wh{\chi})_{I_j}\|_2^2 \leq 16 \e \frac{\Err^2 (\wh{X} - \wh{\chi} , 10k_0 , k_1)}{k_0} \Big\} \nn
\end{equation}
However, the definition of $S_{\e}$ in \eqref{eq:S_eps} reveals that the condition upper bounding $\|(\wh{X} - \wh{\chi})_{I_j}\|_2^2$ is redundant, and $\Stail \backslash S_{\e} \supset \big[\frac{n}{k_1}\big] \backslash S_{\e}$, and hence $\Stail \cup S_{\e} = (\Stail \backslash S_{\e}) \cup S_{\e} = \big[\frac{n}{k_1}\big] $.

\paragraph {Bounding the probability of \eqref{eq:69}:} To show \eqref{eq:69}, we use the second part of Lemma \ref{lem:prune}. The set $\Shead$ therein is defined as
$$\Shead = \Big \{ j \in \Big[\frac{n}{k_1}\Big] \, : \, \|(\wh{X} - \wh{\chi})_{I_j}\|_2 \geq \sqrt{\theta} + \sqrt{\frac{\e}{k_0}} \|\wh{X} - \wh{\chi}\|_2 \Big \}.$$
   	By substituting $\theta=200 \e \cdot \nu^2 $ (\emph{cf.}, Algorithm \ref{alg:final}) and using the assumption $\|\wh{X}- \wh{\chi}\|_2^2 \le 100k_0 \nu^2$ in the lemma, we have 
       	$$\sqrt{\theta} + \sqrt{\frac{\e}{k_0}}\|\wh{X} - \wh{\chi}\|_2 = \sqrt{200 \e \cdot \nu^2} + \sqrt{\frac{\e}{k_0}}\|\wh{X} - \wh{\chi}\|_2 \le \sqrt{600 \e \cdot \nu^2},$$
and hence
\begin{equation}
\Shead \supseteq \Big \{ j \in \Big[\frac{n}{k_1}\Big] \, : \, \|(\wh{X} - \wh{\chi})_{I_j}\|_2^2 \geq 600\e \cdot \nu^2 \Big \}. \label{shead-constsnr}
\end{equation}
Next, we write
\begin{align}
\sum_{j \in [\frac{n}{k_1}] \backslash L'} \| (\wh{X}-\wh{\chi})_{I_j} \|_2^2 
%&= \| (\wh{X}-\wh{\chi}^{(t)})_{S_t \backslash L'} \|_2^2 + \| (\wh{X}-\wh{\chi}^{(t)})_{[\frac{n}{k_1}] \backslash (L' \cup S_t)} \|_2^2 \\
%&= \| (\wh{X}-\wh{\chi}^{(t)})_{(S_t \cap \Shead) \backslash L'} \|_2^2 + \| (\wh{X}-\wh{\chi}^{(t)})_{S_t \backslash (\Shead \cup L') } \|_2^2 \\
%&+ \| (\wh{X}-\wh{\chi}^{(t)})_{[\frac{n}{k_1}] \backslash (L' \cup S_t)} \|_2^2 \\
&= \sum_{j \in (S_0 \cap \Shead \cap L) \backslash L'} \| (\wh{X}-\wh{\chi})_{I_j} \|_2^2 + \sum_{j \in (S_0 \cap \Shead) \backslash (L' \cup L)}\| (\wh{X}-\wh{\chi})_{I_j} \|_2^2 \nonumber \\
&\quad+ \sum_{j \in S_0 \backslash (\Shead \cup L') }\| (\wh{X}-\wh{\chi})_{I_j} \|_2^2 + \sum_{j \in [\frac{n}{k_1}] \backslash (L' \cup S_0)}\| (\wh{X}-\wh{\chi})_{I_j} \|_2^2, \label{eq:four_terms2}
\end{align}
and we proceed by upper bounding the four terms.

\paragraph {Bounding the first term in \eqref{eq:four_terms2}:} By part \textbf{b} of Lemma \ref{lem:prune}, the choice $\delta = \e$, and the use of Markov, we have
\begin{equation}
\sum_{j \in (S_0 \cap \Shead \cap L) \backslash L'}\| (\wh{X}-\wh{\chi})_{I_j} \|_2^2 \leq  \e \sum_{j \in L \cap \Shead}\|(\wh{X}-\wh{\chi})_{I_j}\|_2^2 \le \e \|\wh{X}-\wh{\chi}\|_2^2 \nn
\end{equation}
with probability at least $1-p$.

\paragraph {Bounding the second term in \eqref{eq:four_terms2}:} Conditioned on $\Eloc$, we have
\begin{align}
\sum_{j \in (S_0 \cap \Shead) \backslash (L \cup L') }\| (\wh{X}-\wh{\chi})_{I_j} \|_2^2  \nn
&\le  \sum_{j \in S_0 \backslash L} \| (\wh{X}-\wh{\chi})_{I_j} \|_2^2 \nn\\
&\le 200 \e \|\wh{X}-\wh{\chi}\|_2^2, \nn
\end{align}
where we have applied \eqref{eqq:67}.

\paragraph {Bounding the third term in \eqref{eq:four_terms2}:} We have
\begin{equation}
\begin{split}
\sum_{j \in S_0 \backslash (\Shead \cup L') }\| (\wh{X}-\wh{\chi})_{I_j} \|_2^2 
&\le \sum_{j \in S_0 \backslash \Shead }\| (\wh{X}-\wh{\chi})_{I_j} \|_2^2 \\
&\le | S_0 \backslash \Shead | \cdot \max_{j \in S_0 \backslash \Shead} \|(\wh{X}-\wh{\chi})_{I_j}\|^2_{2} \\
&\le |S_0| (600 \e \cdot \nu^2 )
\end{split} \nn
\end{equation}
by \eqref{shead-constsnr}. We have by definition that $|S_0| = 10k_0$ (\emph{cf.}, \eqref{eq:S0}), and hence
\begin{equation}
\begin{split}
\sum_{j \in (L \cap S_0) \backslash (\Shead \cup L')}\| (\wh{X}-\wh{\chi})_{I_j} \|_2^2 
&\le 6000 k_0 \e  \cdot \nu^2. 
\end{split} \nn
\end{equation}

\paragraph {Bounding the fourth term in \eqref{eq:four_terms2}:} We have
\begin{equation}
\begin{split}
\sum_{j \in [\frac{n}{k_1}] \backslash (L' \cup S_0)}\| (\wh{X}-\wh{\chi})_{I_j} \|_2^2
& \le \sum_{j \in [\frac{n}{k_1}] \backslash S_0}\| (\wh{X}-\wh{\chi})_{I_j} \|_2^2\\
&= \Err^2 (\wh{X} - \wh{\chi} , 10k_0 , k_1),
\end{split} \nn
\end{equation}
which follows from the definition of $S_0$ in \eqref{eq:S0}, along with Definition \ref{def:SNR}.

Hence by the union bound, it follows that $\Eprune$ holds with probability at least $1 - p - \frac{1}{1000}$ conditioned on $\Eloc$.

\paragraph {Success event of estimation primitive:} Let $\Eest$ be the event of having a successful run of \textsc{EstimateValues}($X, \wh{\chi}, L, n, 3k_0/\e, k_1 , \e, p$), meaning the following conditions on the output, $W$:
\begin{gather}
	W_f = 0 \text{ for all } f \notin \Fc \nn \\
	\sum_{j \in L'} \|(\wh{X}-\wh{\chi} - W)_{I_j}\|_2^2 \le \e \| \wh{X}-\wh{\chi} \|_2^2, \label{eq79}
\end{gather}
where $\Fc$ contains the frequencies within the blocks indexed by $L'$. Since the assumption of the theorem implies that $\|\wh{\chi}\|_0 = O(k_0k_1)$, by Lemma \ref{lem:estimate} (with $\delta = \epsilon$ and $3k_0/\e$ in place of $k_0$) and the fact that conditioned on $\Eprune$ we have $|L'| \le 3k_0/\e$ (\emph{cf.}, \eqref{eq:68}), it follows that $\Eest$ occurs with probability at least $1- p$.

\paragraph{Combining the events:} We can now can wrap everything up.  

Letting $\Ec$ denote the overall success event corresponding to the claim of the lemma, we have
\begin{equation}
\begin{split}
\PP [ \Ec ] 
&= \PP \big[ \Eloc \cap \Eprune \cap \Eest \big]\\
&= \PP \big[ \Eest \Big| \Eloc \cap \Eprune \big] \PP \big[ \Eprune \big| \Eloc \big] \PP \big[ \Eloc \big].
\end{split} \nn
\end{equation}
By the results that we have above, along with the union bound, it follows that
\begin{equation}
\PP \big[ \Ec  \big] \ge 1- 2/100 - 3p \ge 0.95 \nn
\end{equation}
for sufficiently small $\eta$ in Algorithm \ref{alg:final}. Applying the union bound over $\bar \Ec$ and the $1/\poly(n)$ probability failure event arising from random perturbations of $\wh{\chi}_0$, we obtain the required bound of $0.9$ on the success probability.

What remains is to first show that the statement of the lemma follows from $\Eloc \cap \Eprune \cap \Eest$. To do this, we observe that, conditioned on these events,
\begin{equation}
\begin{split}
\|\wh{X} - \wh{\chi}'\|_2^2
&=  \sum_{j \in L'} \|(\wh{X} - \wh{\chi}')_{I_j}\|_2^2 + \sum_{j \in [\frac{n}{k_1}] \backslash L'}\|(\wh{X} - \wh{\chi}')_{I_j}\|_2^2 + \err(\wh{\chi}_0)\\
&= \sum_{j \in L'} \|(\wh{X} - \wh{\chi} - W)_{I_j}\|_2^2 + \sum_{j \in [\frac{n}{k_1}] \backslash L'}\|(\wh{X} - \wh{\chi})_{I_j}\|_2^2 + \err(\wh{\chi}_0)\\
&\le \e \| \wh{X}-\wh{\chi} \|_2^2 + 300\e \|\wh{X}-\wh{\chi}\|_2^2 + 6000\e \nu^2 k_0 + \Err^2 (\wh{X} - \wh{\chi} , 10k_0 , k_1)+ \e \nu^2\\
&\le (4 \cdot 10^5) \e \nu^2 k_0+ \Err^2 (\wh{X} - \wh{\chi} , 10k_0 , k_1),
\end{split} \nn
\end{equation}
where the second line follows since $\wh{\chi}' = \wh{\chi} + W$ and $W$ is non-zero only within the blocks indexed by $L'$, the third line follows from \eqref{eq:69} and \eqref{eq79}, and the final line follows from the assumption $\| \wh{X}-\wh{\chi} \|_2^2 \le 100k_0 \nu^2$ in the lemma.

	\paragraph{Sample complexity:} We condition on $\Eloc \cap \Eprune \cap \Eest$ and calculate the number of samples used. 
	
	By Lemma \ref{lem:multi_locate} with $\delta = \e^2$, the sample complexity of \textsc{MultiBlockLocate} is $O^*(|L| \cdot \log n + \frac{k_0 k_1}{\e^4} )$ which behaves $ O^*( \frac{k_0}{\e^2} \log (1+k_0) \log n + \frac{k_0 k_1}{\e^4} \log \frac{1}{p})$ conditioned on $\Eloc$ (see \eqref{eqq:67}).
	Moreover, by Lemma \ref{lem:prune} with $\delta = \e$, the sample complexity of \textsc{PruneLocation} is $O(\frac{k_0 k_1}{\e} \log \frac{1}{\e p} \log \frac{1}{\e})$. Moreover, by Lemma \ref{lem:estimate} with $\delta = \e$, the sample complexity of \textsc{EstimateValues} is $O(\frac{k_0 k_1}{\e^2} \log \frac{1}{p} \log\frac{1}{\e})$. 
	
	The claim follows by summing the three terms and noting from the choice of $p$ in Algorithm \ref{alg:final} that, up to $\log\log\frac{k_0}{\e}$ factors, we can replace $p$ by $\e$ in the above calculations.
	
	\paragraph{Runtime:} As in sample complexity analysis, we condition on $\Eloc \cap \Eprune \cap \Eest$. 
	
	First, by Lemma \ref{lem:multi_locate} with $\delta = \e^2$, the runtime of \textsc{MultiBlockLocate} is $O^*\big( |L| \cdot\log^2 n + \frac{k_0 k_1}{\e^4} \log^2 n + \frac{k_0 k_1}{\e^2} \log^3 n \big)$, which behaves as  $O^*\big( \frac{k_0}{\e^2}  \log \frac{k_0}{\e}\cdot\log^2 n + \frac{k_0 k_1}{\e^4} \log^2 n + \frac{k_0 k_1}{\e^2} \log^3 n  \big)$ conditioned on $\Eloc$  (see \eqref{eq:list-size-before-pruning-constsnr}).
	Second, by Lemma \ref{lem:prune} with $\delta = \e$, the runtime of \textsc{PruneLocation} is $O(\frac{k_0 k_1}{\e} \log \frac{1}{\e p} \log \frac{1}{\e} \log n + k_1 \cdot |L| \log\frac{1}{\e p})$, which behaves as  $O(\frac{k_0 k_1}{\e} \log \frac{1}{\e p} \log \frac{1}{\e} \log n + k_1  \cdot \frac{k_0}{\e^2}  \log \frac{k_0}{\e}  \log^4 \frac{1}{\e p})$ conditioned on $\Eloc$ (see \eqref{eqq:67}).
	Finally, by Lemma \ref{lem:estimate} with $\delta = \e$, the runtime of \textsc{EstimateValues} is $O(\frac{k_0 k_1}{\e} \log \frac{1}{p} \log\frac{1}{\e} \log n + k_1 \cdot |L'|\log \frac{1}{p})$, which behaves as $O(\frac{k_0 k_1}{\e} \log \frac{1}{p} \log\frac{1}{\e} \log n)$ conditioned on $\Eprune$ (see \eqref{eq:68} and recall that $|S_0| = 10k_0$).
	
	The claim follows by summing the above terms and replacing $p$ by $\e$, with the $\log \log n$, $\log \log \SNR'$ and $\log\frac{1}{\epsilon}$ terms absorbed into the $O^*(\cdot)$ notation.

\section{Discussion on Energy-Based Importance Sampling} \label{sec:discussion}

Here we provide further discussing on the adaptive energy-based importance sampling scheme described in Sections \ref{sec:overview}--\ref{sec:ep_sampling}.  Recall from Definition \ref{def:downsampling} that given the signal $X$ and filter $G$, we are considering downsampled signals of the form $\wh{Z}^r_j = (\wh{X}^r \star \wh{G} )_{j k_1}$ with $X^r_i = X_{i + \frac{nr}{2k_1}}$ for $r \in [2k_1]$, and recall from \eqref{eq:budget_alloc} that the goal of energy-based importance sampling is to approximately solve the covering problem
\begin{equation}
    \text{Minimize}_{ \{s^r\}_{r\in[2k_1]} } ~~~ \sum_{r \in [2k_1]} s^r ~~~  \text{ subject to } ~~~  \sum_{\substack{j \,:\, |\wh{Z}^r_j|^2 \geq \frac{\|\wh{Z}^r\|_2^2}{s^r} \\ \text{ for some }r\in[2k_1]}} \|\wh{X}_{I_j}\|_2^2 \ge (1 - \alpha) \|\wh{X}^*\|_2^2 \label{eq:budget_alloc2}
\end{equation}
for suitable $\alpha \in (0,1)$, where $\wh{X}^*$ is the best $(k_0.k_1)$-block sparse approximation of $\wh{X}$.

To ease the discussion, we assume throughout this appendix that the filter $G$ is a width-$\frac{n}{k_1}$ rectangle in time domain, corresponding to a sinc pulse of ``width'' $k_1$ in frequency domain.  Such a filter is less tight than the one we use (see the proof of Lemma \ref{lem:filter_properties}), but similar enough for the purposes of the discussion.

\subsection{Examples -- Flat vs.~Spiky Energies} \label{sec:sampling_examples}

We begin by providing two examples for the $1$-block sparse case, demonstrating how the energies can vary with $r$. An illustration of the energy in each $\wh{Z}^r$ is illustrated in Figure \ref{fig:sinc_example} in two different cases -- one in which $X$ is a sinc pulse (i.e., rectangular in frequency domain), and one in which $X$ is constant (i.e., a delta function in frequency domain).  Both of the signals are $(1,k_1)$-block sparse with $k_1 = 16$, and the signal energy is the same in both cases.  However, the sinc pulse gives significantly greater variations in $|\wh{Z}^r_j|^2$ as a function of $r$.  In fact, these examples demonstrate two extremes that can occur -- in one case, the energy exhibits no variations, and in the other case, the energy is $O(k_1)$ times its expected value for an $O\big(\frac{1}{k_1}\big)$ fraction of the $r$ values.

The second example above is, of course, an extreme case of a $(1,k_1)$-block sparse signal, because it is also $(1,1)$-block sparse.  Nevertheless, one also observes a similar flatness in time domain for other signals; e.g., one could take the first example above and randomize the signs, as opposed to letting them all be positive.

\begin{figure}
	\centering
    \includegraphics[width=1\textwidth]{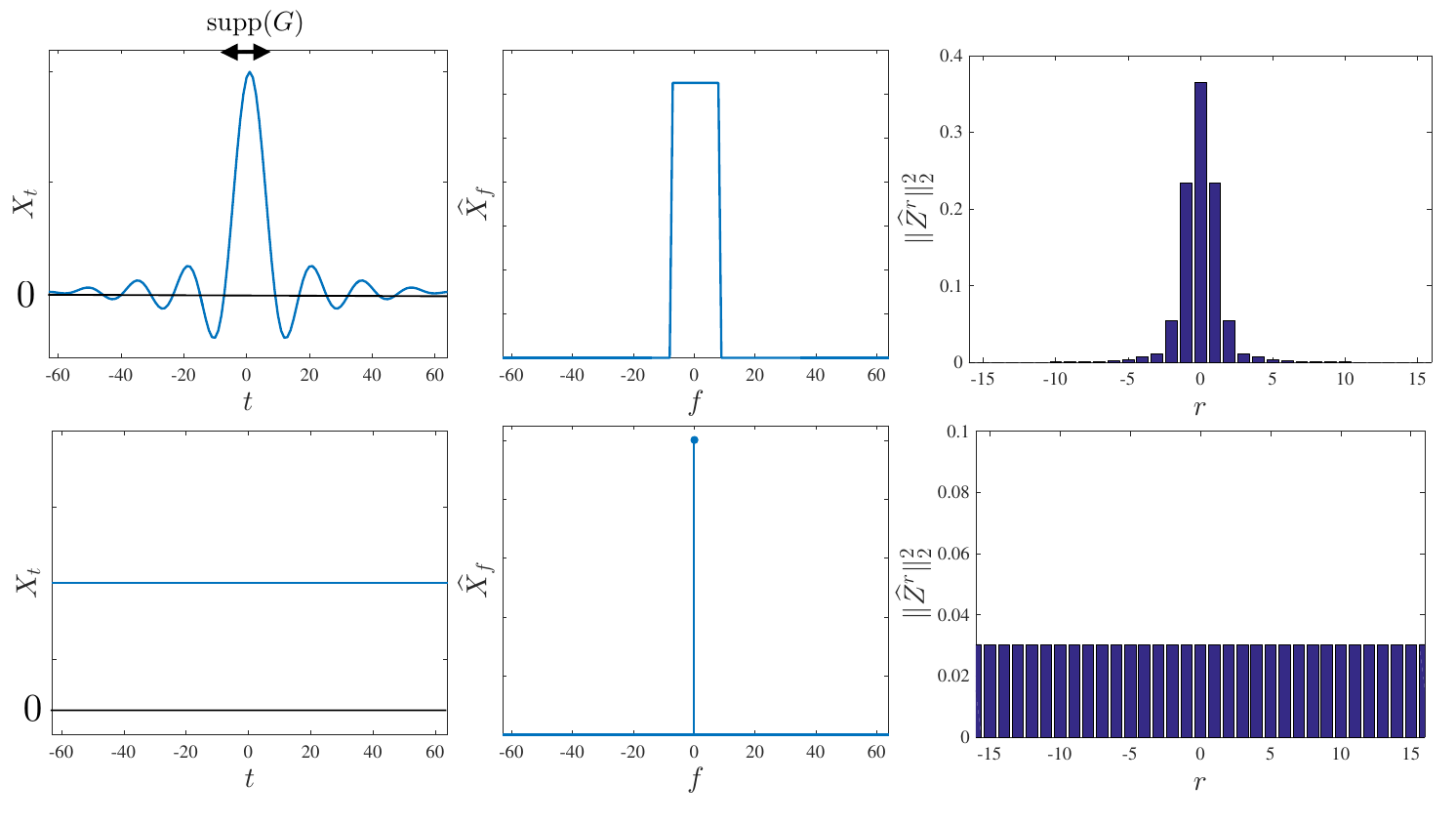}
    \par
    
	\caption{Behavior of $\|\wh{Z}^r\|_2^2$ as a function of $r$ for a sinc function (top) and a rectangular function (bottom), both of which are $(1,16)$-block sparse. \label{fig:sinc_example}} \vspace*{-2ex}
\end{figure}

\subsection{The $\log(1+k_0)$ factor} \label{app:logk0}

Here we provide an example demonstrating that, as long as we rely solely on frequencies being covered according to Definition \ref{def:covered}, after performing the budget allocation, the extra $\log(1+k_0)$ factor in our analysis is unavoidable. Specifically, we argue that for a certain signal $X$, the \emph{optimal} solution to \eqref{eq:budget_alloc} satisfies $\sum_{r \in [2k_1]} s^r = \Omega(k_0 \log(1+k_0))$.  However, we do not claim that this $\log(1+k_0)$ factor is unavoidable for \emph{arbitrary} sparse FFT algorithms.

 We consider a scenario where $k_0 = \Theta(k_1) = o(n)$, and for concreteness, we let both $k_0$ and $k_1$ behave as $\Theta(n^{0.1})$; hence, $\log(1+k_0) = O(\log k_0)$

\begin{figure}
    \centering
    \includegraphics[width=0.7\textwidth]{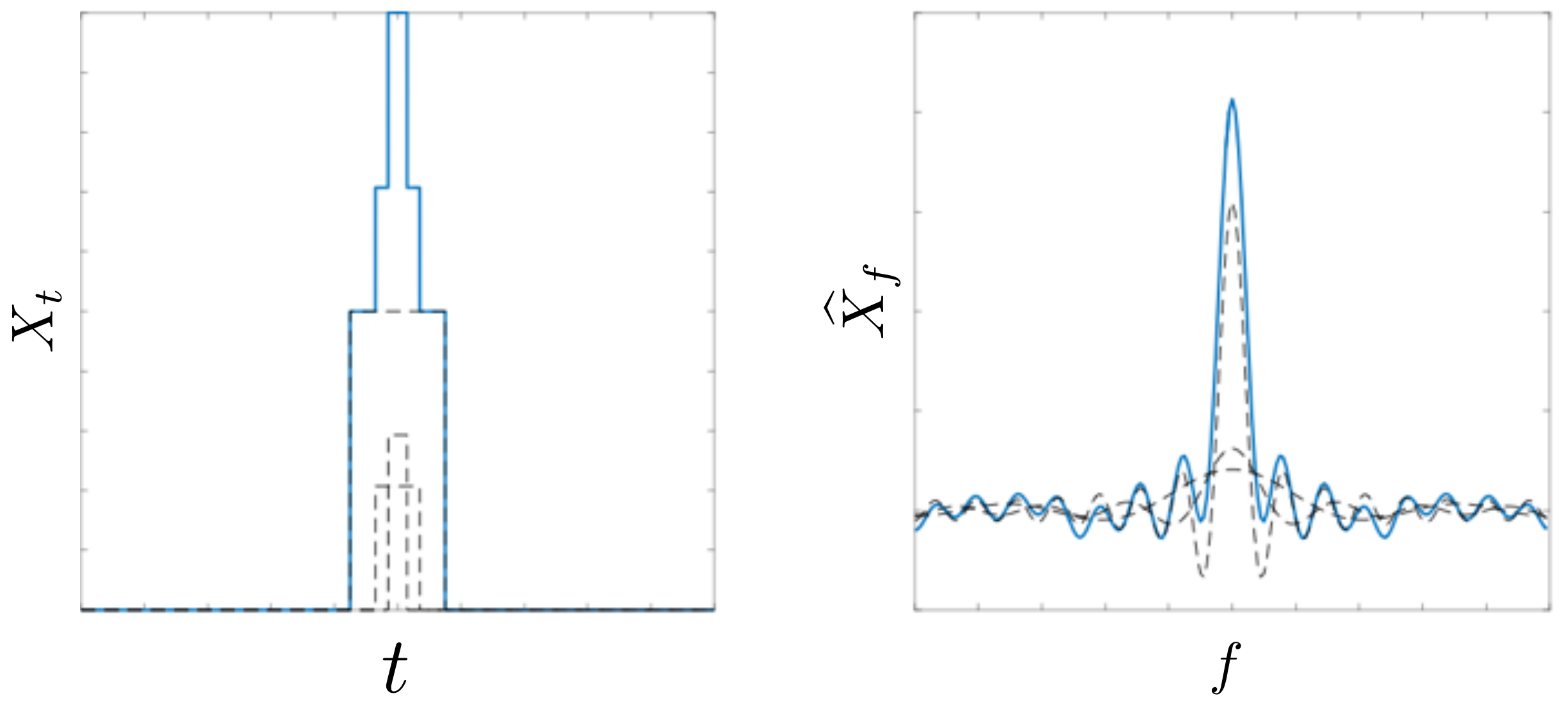}
    \par
    
    \caption{Base signal and its Fourier transform, for constructing a signal where a $\log k_0$ loss is unavoidable with our techniques. \label{fig:logk0}} \vspace*{-2ex}
\end{figure}

\textbf{Constructing a base signal:} We first specify a base signal $W \in \CC^{n}$ that will be used to construct the approximately $(k_0,k_1)$-block sparse signal.  Specifically, we fix the integers $C$ and $L$, and set
\begin{equation}
    W_t = 
    \begin{cases}
        \sqrt{2^{L-1}} & |t| \le \frac{Cn}{2 k_1} \\
        \sqrt{2^{L-2}} & \frac{Cn}{2 k_1} < |t| \le \frac{3Cn}{2 k_1} \\
        % \sqrt{2^{L-3}} & \frac{3n}{2C k_1} \le |t| < \frac{7n}{2C k_1} \\
        \vdots & \vdots \\
        \sqrt{2^{L - \ell}} & \frac{(2^{\ell-1} - 1)Cn}{2 k_1} < |t| \le \frac{(2^{\ell} - 1)Cn}{2 k_1} \\
        \vdots & \vdots \\
        1 & \frac{(2^{L-1} - 1)Cn}{2 k_1} < |t| \le \frac{(2^{L} - 1)Cn}{2 k_1}  \\
        0 & |f| > \frac{(2^{L +1} - 1)Cn}{2 k_1}.  \\
    \end{cases} \label{eq:logk0_W}
\end{equation}
Hence, the signal contains $L$ regions of exponentially increasing width but exponentially decreasing magnitude.  See Figure \ref{fig:logk0} for an illustration ($L = 3$), and observe that we can express this function as a sum of rectangles having geometrically decreasing magnitudes.  Hence, we can specify its Fourier transform as a sum of sinc functions.  

The narrowest of the rectangles has width $\frac{Cn}{k_1}$, and hence the widest of the sinc pulses has width $\frac{k_1}{C}$.  This means that by choosing $C$ to be sufficiently large, we can ensure that an arbitrarily high proportion of the energy lies in a window of length $k_1$ in frequency domain, meaning $W$ is approximately $1$-block sparse.  

\textbf{Constructing a block-spare signal:} We construct a $k_0$-block sparse signal by adding multiple copies of $W$ together, each shifted by a different amount in time domain, and also modulated by a different frequency (i.e., shifted by a different amount in frequency domain).  We choose $L$ such that the cases in \eqref{eq:logk0_W} collectively occupy the whole time domain, yielding $L = \Theta(\log k_1) = \Theta(\log k_0)$.

Then, we set $k_0 = \frac{k_1}{C}$ and let each copy of $W$ be shifted by a multiple of $\frac{Cn}{k_1}$, so that the copies are separated by a distance equal to the length of the thinnest segment of $W$, and collectively these thin segments cover the whole space $[n]$.  As for the modulation, we choose these so that the resulting peaks in frequency domain are separated by $\Omega(k_1^{4})$, so that there the tail of the copy of $W$ corresponding to one block has a negligible effect on the other blocks.  This is possible within $n$ coefficients, since we have chosen $k_1 = O(n^{0.1})$.

\textbf{Evaluating the values of $|Z_j^r|^2$} Recall that we are considering $G$ in \eqref{eq:budget_alloc2} equaling a rectangle of width $\frac{n}{k_1}$.  Because of the above-mentioned separation of the blocks in frequency domain, each copy of $W$ can essentially be treated separately.  By construction, within a window of length $\frac{n}{k_1}$, we have one copy of $W$ at magnitude $\sqrt{2^{L-1}}$, two copies at magnitude $\sqrt{2^{L-2}}$, and so on.  Upon subsampling by a factor of $k_1$, the relative magnitudes remain the same; there is no aliasing, since we let $G$ be rectangular.  Hence, the dominant coefficients in the spectrum of the subsampled signal exhibit this same structure, having energies of a form such as $(8,4,4,2,2,2,2,1,1,1,1,1,1,1,1)$ when sorted and scaled (up to negligible leakage effects).  Moreover, the matrix of $|\wh{Z}_j^r|^2$ values (\emph{cf.}, Figure \ref{fig:MatrixExample}) essentially amounts to circular shifts of a vector of this form -- the structure of any given $\wh{Z}^r$ maintains this geometric structure, but possibly in a different order.

\textbf{Lower bounding the sum of budgets allocated:} We now turn to the allocation problem in \eqref{eq:budget_alloc2}.  Allocating a sparsity budget $s$ to a signal $Z^r$ covers all coefficients $j$ for which $|\wh{Z}^r_j|^2 \ge \frac{ \|\wh{Z}\|^2 }{ s }$.  For the signal we have constructed, the total energy $E$ is equally spread among the $L$ geometric levels: The $\ell$-th level consists of $2^{\ell-1}$ coefficients having energy $2^{1-\ell} \frac{E}{L}$, and hence covering that level requires that $s \ge L \cdot 2^{\ell}$.

Hence, setting $s = L \cdot 2^{\ell-1}$ covers the top $\ell$ levels, for a total of $2^{\ell} - 1$ coefficients.  That is, covering some number of coefficients requires letting $s$ be $\Omega(L)$ times that number, and hence covering a constant fraction of the $k_0$ coefficients requires the sum of sparsity budgets to be $\Omega( L k_0 )$.  Moreover, we have designed every block to have the same energy, so accounting for a constant fraction of the energy amounts to covering a constant fraction of the $k_0$ coefficients.

Since we selected $L = \Theta(\log k_0)$, this means that the sum of sparsity budgets is $\Omega(k_0 \log k _0)$, so that the $\log k_0$ factor must be present in any solution to \eqref{eq:budget_alloc2}.

%!TEX root = BlockSparseFT-new.tex
\section{Location of Reduced Signals} \label{sec:loc_k}
\newcommand{\one}{{\bf 1}}
\newcommand{\zero}{{\bf 0}}
\newcommand{\h}{{\bf w}}
\renewcommand{\H}{{\bf W}}
\newcommand{\A}{{\mathcal A}}
\newcommand{\quant}{\text{quant}}
\newcommand{\E}{{\mathcal{E}}}

In Algorithm \ref{alg:loc_k}, we provide a location primitive that, given a sequence of budgets $s^r$, locates dominant frequencies in the sequence of reduced signals $Z^r$ using $O\big( \sum_{r\in[2k_1]} s^r\log n \big)$ samples. The core of the primitive is a simple $k$-sparse recovery scheme, where $k$ frequencies are hashed into $B=Ck$ buckets for a large constant $C>1$, and then each bucket is decoded individually. Specifically, for each bucket that is approximately $1$-sparse (i.e., dominated by a single frequency that hashed into it) the algorithm accesses the signal at about a logarithmic number of locations and decodes the bit representation of the dominant frequency bit by bit. More precisely, to achieve the right sample complexity we decode the frequencies in blocks of $O(\log\log n)$ bits. Such schemes or versions thereof have been used in the literature (e.g., \cite{GMS, Has12, K16}).

A novel aspect of our decoding scheme is that it receives access to the input signal $X$, but must run a basic sparse recovery scheme as above on each reduced signal $Z^r$. Specifically, for each $r$ it must hash $Z^r$ into $s^r$ buckets (the {\em budget} computed in \textsc{MultiblockLocate} and passed to \textsc{LocateSignal} as input). This would be trivial since $Z^r$ can be easily accessed given access to $X$ (\emph{cf.}, Lemma \ref{lem:downsamp-cost-unit-access}), but the fact that we need to operate on the {\em residual signal} $X-\chi$ (where $\wh{\chi}$ is block sparse and given explicitly as input) introduces difficulties.

The difficulty is that we would like to compute $\chi$ on the samples that individual invocations of sparse recovery use, for each $r\in [2k_1]$, but computing this directly would be very costly. Our solution consists of ensuring that all invocations of sparse recovery use the same random permutation $\pi$, and therefore all need to access $X-\chi$ on a set of shifted intervals after a change of variables given by $\pi$ (crucially, this change of variables is shared across all $r$). The lengths of the intervals are different, and given by $s^r$, but it suffices to compute the values of $\chi$ on the shifts of the largest of these intervals, which is done in \textsc{HashToBinsReduced} (see Lemma~\ref{lem:hash2bins}). We present the details below in Algorithm~\ref{alg:loc_k}.

\if 0  While the hashing techniques themselves we use here are fairly standard (e.g., see \cite{Has12,Has12a}), we details are slightly different, and we provide a proof for completeness.  A notable difference to \cite{Has12,Has12a} is that we use less sharp filters, so as to avoid any dependence on $\log n$ in the sample complexity. 

Moreover, there is a non-trivial addition of computing the required values of a $(k_1,\delta)$-downsampling using \textsc{HashToBinsReduced} in Algorithm  \ref{alg:hash2bins}.  In order to compute these values efficiently, we must ensure that we use \emph{the same permutations} for each downsampled signal.  
% Note that not all of the parameters to \textsc{SemiEquiInverseFFT} are shown in Algorithm \ref{alg:loc_k}, but they are made clear in the runtime analysis in the proof of Lemma \ref{lem:loc_k} below.
\fi

For convenience, throughout this section, we use $m$ to denote the reduced signal length $n/k_1$.

\begin{algorithm}[H]
    \caption{Location primitive: Given access to the input signal $X$,  a partially recovered signal $\wh{\chi}$, a budget $k$ and bound on failure probability $p$, identifies any given $j \in [n/k_1]$ with $|\wh{Z}_j^r|^2\geq \|\wh{Z}^r\|_2^2/k$ for some $r \in [2k_1]$, in the $(k_1,\delta)$-downsampling of $X-\chi$.}\label{alg:location}
    \begin{algorithmic}[1]
    \Procedure{LocateReducedSignals}{$X, \wh{\chi}, n, k_0, k_1, \{s^r\}_{r \in [2k_1]}, \delta, p$} 
        \\ \Comment{Uses large absolute constants $C_1,C_2,C_3 > 0$}
    \State $B^r \gets C_2 s^r$ for each $r \ in [2k_1]$
    \State $H^r \gets (m, B, F')$-flat filter for each $r \in [2k_1]$, for sufficiently large $F'\geq 2$
    \State $\Bmax \gets B^r$
    \State $\{Z_X^r\}_{r \in [2k_1]} \gets (k_1,\delta)$-downsampling of $X$ \Comment See Definition \ref{def:downsampling}
    \State $m \gets n/k_1$
    \State $L\gets \emptyset$
    
    \For {$t=\{1,\dotsc,C_1\log(2/p)\}$}
    \State $\sigma \gets $ uniformly random odd integer in $[m]$
    \State $\A\gets $ $C_3\log\log m$ uniformly random elements in $[m]\times [m]$
    \State $\Lambda\gets 2^{\lfloor \frac1{2}\log_2\log_2 m\rfloor}$, $N\gets \Lambda^{\lceil \log_\Lambda m\rceil}$ \Comment{Implicitly extend $X$ to an $m$-periodic length-$N$ signal}
    
    \For{$(\alpha,\beta)\in \A$} \Comment Hashing with common randomness
            \State $\h\gets N\Lambda^{-g} $
        \State $\Delta \gets \alpha + \h\cdot\beta$
        \State $\mathbf{H} \gets \{H^r\}_{r \in [2k_1]}$
        \State $\mathbf{B} \gets \{B^r\}_{r\in[2k_1]}$
        \State $\wh{U}^r(\alpha+\h \cdot \beta) \gets \textsc{HashToBinsReduced}(\{Z_X^r\}_{r \in [2k_1]},\wh{\chi}, \mathbf{H}, n,k_1,\mathbf{B},\sigma,\Delta)$
    \EndFor
    
%    \State Compute $\{\chi_i\}$ using $\textsc{SemiEquiInverseBlockFFT}$ w/ block sparsity $(O(k_0+\Bmax),k_1)$  \\
%         \Comment Performed for all $\pi(j) = \sigma j + (\alpha+ N\Lambda^{-g} \cdot \beta)$ with $(\alpha,\beta)\in \A$, $g \in \{1,\dotsc,\log_{\Lambda} N\}$ \\
%         \Comment See proof of Lemma \ref{lem:loc_k} for precise parameters
    
    \For {$r \in [2k_1]$}
        \State $B\gets C_2 s^r$ \label{line:est_B} \Comment Rounded up to a power of two
        
%        \For {$g=\{1,\dotsc,\log_\Lambda N\}$}
%            \State $\h\gets N\Lambda^{-g}  $ % \Comment{Note that $\h=0$ modulo $N$ when $g=0$}
%        \State $\wh{U}^r(\alpha+\h \cdot \beta) \gets$ FFT of $(m,B,H,\sigma,\alpha+\h \cdot \beta)$-hashing of $Z^r$, for all $(\alpha,\beta)\in \A$ \label{line:loc_hash}
%        \EndFor

        \For{$b \in [B]$}\Comment{Loop over all hash buckets}
            \State ${\bf f}\gets {\bf 0}$
            \For {$g=\{1,\dotsc,\log_\Lambda N\}$} 
            \State $\h\gets N\Lambda^{-g} $ \label{line:g_loop_start}
            \State {\bf If}~there exists a unique $\lambda \in \{0,1,\dotsc,\Lambda-1\}$ such that  \label{line:r_test}
            \State~~~~$\left|\omega_{\Lambda}^{-\lambda \cdot \beta}\cdot \omega^{-(N\cdot \Lambda^{-g} {\bf f})\cdot \beta}\cdot \frac{\wh{U}^r_b(\alpha+\h\cdot \beta)}{\wh{U}^r_b(\alpha)}-1\right|<\frac1{3}$ for at least $\frac{3}{5}$ fraction of $(\alpha, \beta)\in \A$
            \State {\bf then}  ${\bf f}\gets {\bf f}+\Lambda^{g-1}\cdot \lambda$ \label{line:g_loop_end}
        \EndFor
        \State $L \gets L \cup \{\sigma^{-1}{\bf f}\cdot \frac{m}{N}\}$ \Comment{Add recovered element to output list}
        \EndFor
        \EndFor
    \EndFor 
    
    \State \textbf{return} $L$
    \EndProcedure 
    \end{algorithmic} \label{alg:loc_k}
\end{algorithm}

\medskip
\noindent \textbf{Lemma \ref{lem:loc_k}} (\textsc{LocateReducedSignal} guarantees -- formal version) {\em
Fix $(n,k_0,k_1)$, the signals $X, \wh{\chi}\in \C^{n}$ with $\wh{\chi}_0$ uniformly distributed over an arbitrarily length-$\frac{\|\wh{\chi}\|^2}{\poly(n)}$ interval, the sparsity budgets $\{s^r\}_{r \in [2k_1]}$ with $s^r = O\big(\frac{k_0}{\delta}\big)$ for all $r\in[2k_1]$, and the parameters $\delta \in \big(\frac{1}{n},\frac{1}{20}\big)$ and $p \in \big(\frac{1}{n^3},\frac{1}{2} \big)$, and let $\{Z^r\}_{r \in [2k_1]}$ be the $(k_1,\delta)$-downsampling of $X - \chi$.

Letting $L$ denote the output of  $\Call{LocateReducedSignals}{X, \wh{\chi}, n, k_0, k_1, \{s^r\}_{r\in[2k_1]}, \delta, p}$, we have that for any $j\in \big[\frac{n}{k_1}\big]$ such that $|Z_j^r|^2\geq \|Z^r\|_2^2/s^r$ for some $r \in [2k_1]$, one has $j \in L$ with probability at least $1-p$.  The list size satisfies $|L| = O\big( \sum_{r \in [2k_1]} s^r \log\frac{1}{\delta}\big)$.  Moreover, if $\wh{\chi}$ is $(O(k_0),k_1)$-block sparse, the sample complexity is $O\big(\sum_{r \in [2k_1]}s^r \log\frac{1}{p} \log\frac{1}{\delta} \log n\big)$, and the runtime is $O\big( \sum_{r \in [2k_1]} s^r  \log\frac{1}{p} \log\frac{1}{\delta} \log^2 n + \frac{k_0 k_1}{\delta}\log\frac{1}{p} \log^3 n\big)$. 
} 
\begin{proof}
    We first note that the claim on the list size follows immediately from the fact that $B = O(s^r)$ entries are added to the list for each $t$ and $r$, and the loop over $t$ is of length $O\big(\log\frac{1}{p}\big)$.
    
    In order to prove the main claim of the lemma, it suffices to show that for any single value of $r$, if we replace the loop over $r$ by that single value, then $L$ contains any given $j\in \big[\frac{n}{k_1}\big]$ such that $|Z_j^r|^2\geq \|Z^r\|_2^2/s^r$, with probability at least $1-p$.  Since this essentially corresponds to a standard sparse recovery problem, we switch to simpler notation throughout the proof: We let $Y$ denote a generic signal $Z^r$, we write its length as $m =n/k_1$, we index its entries in frequency domain as $\wh{Y}_f$, and we define $k = s^r$.
    
      The proof now consists of two steps. First, we show correctness of the location algorithm assuming that the \textsc{SemiEquiInverseBlockFFT} computation in line~11 computes all the required values for the computation of $\wh{U}$ in line~19. We then prove that \textsc{SemiEquiInverseBlockFFT} indeed computes all the required values of $\chi$, and conclude with sample complexity and runtime bounds.

{\bf Proving correctness of the location process}	We show that each element $f$ with $|\wh{Y}_f|^2\geq \|\wh{Y}\|_2^2/k$ is reported in a given iteration of the outer loop over $t=1,\ldots, C_1\log(2/p)$, with probability at least $9/10$. Since the loops use independent randomness, the probability of $f$ not being reported in any of the iterations is bounded by $(1/10)^{C_1\log(2/p)}\leq p/2$ if $C_1$ is sufficiently large.
	
	Fix an iteration $t$. We first show that the random set $\A$ chosen in \textsc{LocateReducedSignals} has useful error-correcting properties with high probability.
	Specifically, we let $\Ebal$ denote the event that for every $\lambda\in [\Lambda], \lambda\neq 0$ at least a fraction $49/100$ of the numbers $\{\omega_\Lambda^{\lambda\cdot \beta}\}_{(\alpha, \beta)\in \A}$ have non-positive real part; in that case, we say that $\A$ is {\em balanced}. We have for fixed $\lambda\in [\Lambda], \lambda\neq 0$ that since the pair $(\alpha, \beta)$ was chosen uniformly at random from $[m]\times [m]$, the quantity $\omega_\Lambda^{\lambda\cdot \beta}$ is uniformly distributed on the set of roots of unity of order $2^s$ for some $s>0$ (since $\lambda\neq 0$). At least half of these roots have non-positive real part, so for every fixed $\lambda\in [\Lambda], \lambda\neq 0$ one has
	$\PP_\beta[\text{Re}(\omega_{\Lambda}^{\lambda\cdot \beta})\leq 0]\geq 1/2$. It thus follows by standard concentration inequalities that for every fixed $\lambda$ at least $49/100$ of the numbers  $\{\omega_\Lambda^{\lambda\cdot \beta}\}_{(\alpha, \beta)\in \A}$ have non-positive real part with probability at least $1-e^{-\Omega(|\A|)}=1-\exp(-\Omega(C_3 \log\log m)))\geq 1-1/(100\log_2 m)$ as long as $C_3$ is larger than an absolute constant. A union bound over $\Lambda\leq \log_2 m$ values of $\lambda$ shows that $\PP[\Ebal]\geq 1-(\log_2 m)\cdot /(100\log_2 m) = 1- 1/100$ for sufficiently large $m$ (recall from Section \ref{sec:intro} that $\frac{n}{k_1}$ exceeds a large absolute constant by assumption). We henceforth condition on $\Ebal$.

	Fix any $f$  such that $|\wh{Y}_f|^2\geq ||\wh{Y}||_2^2/k$, and let $q=\sigma i$ for convenience. We show by induction on $g=1,\ldots, \log_\Lambda N$ that before the $g$-th iteration of lines \ref{line:g_loop_start}--\ref{line:g_loop_end} of Algorithm~\ref{alg:location}, we have that ${\bf f}$ coincides with ${\bf q}$ on the bottom $g\cdot \log_2 \Lambda$ bits, i.e., ${\bf f}-{\bf q}= 0 \mod \Lambda^{g-1}$. 
	
	The {\bf base} of the induction is trivial and is provided by $g=1$.  
	We now show the {\bf inductive step}. Assume by the inductive hypothesis that ${\bf f}-{\bf q}= 0 \mod \Lambda^{g-1}$, so that
	${\bf q}={\bf f}+\Lambda^{g-1}(\lambda_0+\Lambda \lambda_1+\Lambda^2 \lambda_2+\ldots)$. Thus,  $(\lambda_0, \lambda_1,\ldots)$ is the expansion of $({\bf q}-{\bf f})/\Lambda^{g-1}$ base $\Lambda$, and $\lambda_0$ is the least significant digit. We now show that $\lambda_0$ is the unique value of $\lambda$ that satisfies the condition of line \ref{line:r_test} of Algorithm~\ref{alg:location}, with high probability 
	
	In the following, we use the definitions of $\pi(f)$, $h(f)$, and $o_f(f')$ from Definition \ref{def:hashing} with $\Delta = 0$.  First, we have for each $a=(\alpha, \beta)\in \A$ and $\h\in \H$ that
	\begin{equation*}
		\begin{split}
			\wh{H}_{o_f(f)}^{-1}\wh{U}_{h(f)}(\alpha+\h\cdot \beta)-  \wh{Y}_f \omega_N^{(\alpha+\h\cdot \beta) {\bf q}}
			&= \wh{H}_{o_f(f)}^{-1}\wh{U}^*_{h(f)}(\alpha+\h\cdot \beta) -  \wh{Y}_f \omega_N^{(\alpha+\h\cdot \beta) {\bf q}} + \err_w\\
			&=\wh{H}_{o_f(f)}^{-1} \sum_{f'\in [m]\setminus \{f\}} \wh{H}_{o_f(f')} \wh{Y}_{f'} \omega_N^{\sigma f' \cdot (\alpha+\h\cdot \beta)}+\err_w=:E'(\h),
		\end{split}
	\end{equation*}
	where $\err_w = \wh{H}_{o_f(f)}^{-1}(\wh{U}_{h(f)} - \wh{U}^*_{h(f)})(\alpha+\h\cdot \beta)$.
	
	And similarly 
	\begin{equation*}
		\begin{split}
			\wh{H}_{o_f(f)}^{-1}\wh{U}_{h(f)}(\alpha)-  \wh{Y}_f \omega_N^{\alpha {\bf q}}
			&= \wh{H}_{o_f(f)}^{-1}\wh{U}^*_{h(f)}(\alpha)-  \wh{Y}_f \omega_N^{\alpha {\bf q}} + \err\\
			&= \wh{H}_{o_f(f)}^{-1}\sum_{f'\in [m]\setminus \{f\}} \wh{H}_{o_f(f')} \wh{Y}'_{f'} \omega_N^{\sigma f' \cdot \alpha} + \err=:E''.
		\end{split}
	\end{equation*}
	where $\err = \wh{H}_{o_f(f)}^{-1}(\wh{U}_{h(f)} - \wh{U}^*_{h(f)})(\alpha)$.

	We will show that $f$ is recovered from bucket $h(f)$ with high (constant) probability. The bounds above imply that 
	\begin{equation}\label{eq:gergergre}
	\begin{split}
	\frac{\wh{U}_{h(f)}(\alpha+\h \beta))}{\wh{U}_{h(f)}(\alpha)}=\frac{\wh{Y}_f \omega_N^{(\alpha+\h \beta) {\bf q}}+E'(\h)}{\wh{Y}_f \omega_N^{\alpha {\bf q}}+E''}.
	\end{split}
	\end{equation}
	
	The rest of the proof consists of two parts. We first show that with high probability over the choice of $\pi$, the error terms $E'(\h)$ and $E''$ are small in absolute value for most $a=(\alpha, \beta)\in \A$ with extremely high probability. We then use this assumption to argue that $f$ is recovered.
	
	\paragraph{Bounding the error terms $E'(\h)$ and $E''$ (part (i)).} We have by Parseval's theorem that
	\begin{equation}\label{eq:go4ng43g34g}
	\begin{split}
	\EE_{a}[|E'(\h)|^2]
	&\leq \wh{H}_{o_f(f)}^{-2} \sum_{f'\in [m]\setminus \{f\}} \wh{H}_{o_f(f')}^2 |Y_{f'}|^2 + |\err_w|^2 + 2 |\err_w| \wh{H}_{o_f(f)}^{-1} \sum_{f'\in [m]\setminus \{f\}} \wh{H}_{o_f(f')} |Y_{f'}|,
	\end{split}
	\end{equation}
	and
	\begin{equation}
	\begin{split}
	\EE_a[|E''|^2]
	&\leq \wh{H}_{o_f(f)}^{-2}\sum_{f'\in [m]\setminus \{f\}} \wh{H}_{o_f(f')}^2 |\wh{Y}_{f'}|^2 + |\err|^2 + 2 |\err| \wh{H}_{o_f(f)}^{-1}\sum_{f'\in [m]\setminus \{f\}} \wh{H}_{o_f(f')} |\wh{Y}_{f'}|,
	\end{split} \nonumber
	\end{equation}
	where we used the fact that $\alpha+\h \beta$ is uniformly random in $[m]$ (due to $\alpha$ being uniformly random in $[m]$ and independent of $\beta$ by definition of $\A$ in line~6 of Algorithm~\ref{alg:location}).
	
	Taking the expectation of the term $\wh{H}_{o_f(f)}^{-2} \sum_{f'\in [m]\setminus \{f\}} \wh{H}_{o_f(f')}^2 |Y_{f'}|^2$ with respect to $\pi$, we obtain
	\begin{align*}
		&\EE_\pi\bigg[\wh{H}_{o_f(f)}^{-2} \sum_{f'\in [m]\setminus \{f\}} \wh{H}_{o_f(f')}^2 |Y_{f'}|^2\bigg]=O(\|Y\|_2^2/B)=O(||X'||_2^2/(C_2k))
	\end{align*}
	by Lemma~\ref{lem:perm_property} (note that $F' \geq 2$, so the lemma applies) and the choice $B=C_2\cdot k$ (line~\ref{line:est_B} of Algorithm~\ref{alg:location}). We thus have by Markov's inequality together with the assumption that $|\wh{Y}_f|^2\geq ||\wh{Y}||_2^2/k$  that 
	$$
	\PP_\pi\bigg[\wh{H}_{o_f(f)}^{-2} \sum_{f'\in [m]\setminus \{f\}} \wh{H}_{o_f(f')}^2 |\wh{Y}_{f'}|^2>|\wh{Y}_f|^2/1700\bigg]<O(1/C_2)<1/40
	$$
	and 
	$$
	\PP_\pi\bigg[\wh{H}_{o_f(f)}^{-2}\sum_{f'\in [m]\setminus \{f\}} \wh{H}_{o_f(f')}^2 |\wh{Y}_{f'}|^2 >|\wh{Y}_f|^2/1700\bigg]<O(1/C_2)<1/40
	$$
	since $C_2$ is larger than an absolute constant by assumption. 
	
	\textbf{Bounding $\err$ and $\err_w$ (numerical errors from semi-equispaced FFT computation):} Recall that we have by assumption that $\wh{\chi}_0$ uniformly distributed over an arbitrarily length-$\frac{\|\wh{\chi}\|^2}{\poly(n)}$ interval, and that $\wh{Y} = \wh{Z}^r$ for some $\wh{Z}^r$ in the $(k_1,\delta)$-downsampling of $X - \chi$.  By decomposing $\wh{Z}^r_j = ((\wh{X}^r - \wh{\chi}^r) \star \wh{G} )_{j k_1}$ into a deterministic part and a random part (in terms of the above-mentioned uniform distribution), we readily obtain for some $c' > 0$ that
	\begin{equation}
        \|\wh{Y}\|^2 \ge \frac{\|\wh{\chi}\|_2^2}{n^{c'}} \label{eq:chi_to_Z}
	\end{equation}
	with probability at least $1-\frac{1}{n^4}$.  Since $p \ge \frac{1}{n^3}$ by assumption, we deduce that this also holds with probability at least $1 - p/2$.  By the the accuracy of the $\wh{\chi}$ values stated in Lemma \ref{lem:semi_equi}, along with the argument used in \eqref{eq:semi_appr2}--\eqref{eq:se_abs_bound} its proof in Appendix \ref{sec:pf_semi} to convert \eqref{eq:chi_to_Z} to accuracy on hashed values, we know that $|\wh{U}_{h(f)}-\wh{U}^*_{h(f)}| \le \|\wh{U}-\wh{U}^*\|_{\infty} \le n^{-c+1}\|\wh{\chi}\|_2$.  Hence, by using $|\wh{H}_{o_f(f)}|^{-2} \le 2$ and $\alpha ,  \alpha+\h \cdot \beta \le m$, we find that 
	\begin{gather}
	|\err| \le 2n^{-c+1} \|\wh{\chi}\|_2 \le 2n^{-c+c'+1} \|\wh{Y}\|_2 \nonumber \\
	|\err_w| \le 2n^{-c+1} \|\wh{\chi}\|_2 \le 2n^{-c+c'+1} \|\wh{Y}\|_2,  \nonumber 
	\end{gather}
	where the second inequality in each equation holds for some $c' > 0$ by \eqref{eq:chi_to_Z}.
	
	Note also that $\wh{H}_{o_f(f)}^{-1}\sum_{f'\in [m]\setminus \{f\}} \wh{H}_{o_f(f')} |\wh{Y}_{f'}| \le 2\|\wh{Y}\|_1 \le 2\sqrt{m}\|\wh{Y}\|_2$, since we have $\wh{H}_{o_f(f)}^{-1} \le 2$ and $\wh{H}_{f'}|\le 1$ for all $f'$. We can thus write
	\begin{align}
	|\err|^2 + 2 |\err| \cdot \wh{H}_{o_f(f)}^{-1}\sum_{f'\in [m]\setminus \{f\}} \wh{H}_{o_f(f')} |\wh{Y}_{f'}|
	&\le |\err|^2 + 4 \sqrt{m} |\err| \cdot \|\wh{Y}\|_2 \nn  \\
	&\le 4n^{-2c+2c'+2} \|\wh{Y}\|_2^2 + 8n^{-c+c'+3/2} \|\wh{Y}\|_2^2 \nn \\
	&= n^{\Omega(-c+c')} \|\wh{Y}\|_2^2, \label{eq:est_poly_error}
	\end{align}
	which can thus be made to behave as $\frac{1}{\poly(n)} \|\wh{Y}\|_2^2$ by a suitable choice of $c$.
	
	\paragraph{Bounding the error terms $E'(\h)$ and $E''$ (part (ii)).} By the union bound, we have $|E'(\h)|^2\leq |\wh{Y}_f|^2/1600$ and  $|E''|^2\leq |\wh{Y}_f|^2/1600$ simultaneously with probability at least $1-1/20$ -- denote the success event by $\E^t_{f, \pi} (\h)$. Conditioned on $\E^t_{f, \pi}(\h)$, we thus have by~\eqref{eq:go4ng43g34g} and \eqref{eq:est_poly_error}, along with the fact that $\A$ is independent of $\pi$, that 
	$$
	\EE_{a}[|E'(\h)|^2|]\leq |Y_f|^2/1600\text{~~~and~~~}\EE_{a}[|E''|^2]\leq |Y_f|^2/1600.
	$$
	Another application of Markov's inequality gives
	$$
	\PP_a[|E'(\h)|^2\geq  |\wh{Y}_f|^2/40]\leq 1/40 \text{~~~and~~~}\PP_a[|E''|^2\geq  |\wh{Y}_f|^2/40]\leq 1/40.
	$$
	This means that conditioned on $\E^t_{f, \pi}(\h)$, with probability at least $1-e^{-\Omega(|\A|)}\geq 1-1/(100\log_2 m)$ over the choice of $\A$, both events occur for all but $2/5$ fraction of $a\in \A$; denote this success event by $\E^t_{f, \A}(\h)$.  We condition on this event in what follows. Let $\A^*(\h)\subseteq \A$ denote the set of values of $a\in \A$ that satisfy the bounds above.
	
	In particular, we can rewrite ~\eqref{eq:gergergre} as
	\begin{align}
	\frac{\wh{U}_{h(f)}(\alpha+\h\beta)}{\wh{U}_{h(f)}(\alpha)}&=\frac{\wh{Y}_f \omega_N^{(\alpha+\h\beta) {\bf q}}+E'(\h)}{\wh{Y}_f \omega_N^{\alpha {\bf q}}+E''} \nn \\
	&=\frac{\omega_N^{(\alpha+\h\beta) {\bf q}}}{\omega_N^{\alpha {\bf q}}}\cdot\xi\text{~~~$\bigg($where~~}\xi=\frac{1+\omega^{-(\alpha+\h\beta) {\bf q}}E'(\h)/\wh{Y}_f'}{1+\omega_N^{-\alpha{\bf q}} E''/\wh{Y}_f'}\bigg) \nn \\
	&=\omega_N^{(\alpha+\h\beta) {\bf q}-\alpha {\bf q}}\cdot\xi \nn \\
	&=\omega_N^{\h\beta {\bf q}}\cdot\xi. \nonumber % \label{eq:gergergre-2}
	\end{align}
	We thus have for  $a\in \A^*(\h)$ that
	\begin{equation}\label{eq:bound-1-oigb344tg32t}
	|E'(\h)|/|\wh{Y}_f'|\leq  1/40\text{~~~and~~~~}|E''|/|\wh{Y}_f'|\leq  1/40.
	\end{equation}

	\paragraph{Showing that $\A^*(\h)\subseteq \A$ suffices for recovery.}
	By the above calculations, we get
	$$
	\frac{\wh{U}_{h(f)}(\alpha+\h \beta)}{\wh{U}_{h(f)}(\alpha)}=\omega_N^{\h \beta {\bf q}}\cdot\xi=\omega_N^{N\Lambda^{-g} \beta {\bf q}}\cdot\xi=\omega_N^{N\Lambda^{-g} \beta {\bf q}}+\omega_N^{N\Lambda^{-g} \beta {\bf q}}(\xi-1). 
	$$
	We proceed by analyzing the first term, and we will later show that the second term is small. Since ${\bf q}={\bf f}+\Lambda^{g-1}(\lambda_0+\Lambda \lambda_1+\Lambda^2 \lambda_2+\ldots)$, by the inductive hypothesis, we have
	\begin{equation*}
		\begin{split}
			\omega_\Lambda^{-\lambda\cdot \beta}\cdot \omega_N^{-N\Lambda^{-g}{\bf f}\cdot \beta}\cdot \omega_N^{N\Lambda^{-g} \beta {\bf q}}&=\omega_\Lambda^{-\lambda\cdot \beta}\cdot \omega_N^{N\Lambda^{-g}({\bf q}-{\bf f})\cdot \beta}\\
			&=\omega_\Lambda^{-\lambda\cdot \beta}\cdot \omega_N^{N\Lambda^{-g}(\Lambda^{g-1}(\lambda_0+\Lambda \lambda_1+\Lambda^2 \lambda_2+\ldots))\cdot \beta}\\
			&=\omega_\Lambda^{-\lambda\cdot \beta}\cdot \omega_N^{(N/\Lambda)\cdot (\lambda_0+\Lambda \lambda_1+\Lambda^2 \lambda_2+\ldots)\cdot \beta}\\
			&=\omega_\Lambda^{-\lambda\cdot \beta}\cdot \omega_{\Lambda}^{\lambda_0\cdot \beta}\\
			&=\omega_\Lambda^{(-\lambda+\lambda_0)\cdot \beta},
		\end{split}
	\end{equation*}
	where we used the fact that $\omega_N^{N/\Lambda}=e^{2\pi f (N/\Lambda)/N}=e^{2\pi f/\Lambda}=\omega_\Lambda$. Thus, we have
	\begin{equation} % \label{eq:92hg34grggggdds}
	\omega_{\Lambda}^{-\lambda\cdot \beta}\omega^{-(N\Lambda^{-g}{\bf f})\cdot \beta}\frac{\wh{U}_{h(f)}(\alpha+\h \beta)}{\wh{U}_{h(f)}(\alpha)}=\omega_\Lambda^{(-\lambda+\lambda_0)\cdot \beta}\xi. \nonumber
	\end{equation}
	
	We now consider two cases. First suppose that $\lambda=\lambda_0$. Then $\omega_\Lambda^{(-\lambda+\lambda_0)\cdot \beta}=1$, and it remains to note that  by~\eqref{eq:bound-1-oigb344tg32t} we have $|\xi-1|\leq \frac{1+1/40}{1-1/40}-1< 1/3$.
	Thus, every $a\in \A^*(\h)$ passes the test in line~\ref{line:r_test} of Algorithm~\ref{alg:location}. Since $|\A^*(\h)|\geq (3/5)|\A|$ by the argument above, we have that $\lambda_0$ passes the test in line~\ref{line:r_test}. It remains to show that $\lambda_0$ is the unique element in $0,\ldots, \Lambda-1$ that passes this test.
	
	Suppose that $\lambda \neq \lambda_0$. Then, by conditioning on $\Ebal$, at least a $49/100$ fraction of $\omega_\Lambda^{(-\lambda+\lambda_0)\cdot \beta}$ have negative real part.  This means that for at least $49/100$ of $a\in \A$, we have 
	$$
	|\omega_\Lambda^{(-\lambda+\lambda_0)\cdot \beta}\xi-1|\geq |\mathbf{i}\cdot |\xi|-1|\geq |(7/9)\mathbf{i}-1|> 1/3,
	$$
	and hence the condition in line~16 of Algorithm~\ref{alg:location} is not satisfied for any $\lambda \neq \lambda_0$. 
	
	We thus get that conditioned on $\Ebal$ and the intersection of $\E^t_{f, \pi}(\h)$ for all $\h\in \H$ and $\E^t_{f, \A}$, recovery succeeds for all values of $g=1,\ldots, \log_\Lambda N$. By a union bound over the failure events, we get that 
	$$
	\PP\bigg[\Ebal\cap \E^t_{f, \A} \cap \Big( \bigcap_{\h\in \H} \E^t_{f, \pi}(\h) \Big) \bigg]\geq 1-1/100-(\log_\Lambda N)\cdot\frac1{100\log_2 m}\geq 98/100.
	$$
	This shows that location is successful for $f$ in a single iteration $t$ with probability at least $98/100\geq 9/10$, as required.
	
	\textbf{Sample complexity and runtime.} 
%	Recall the definition of hashing from \eqref{eq:hashing}:
%    \begin{equation}
%        U_f = \frac{B}{n} \sum_{i\in [\frac{m}{B}]} Z^r_{\sigma( \Delta + f + B \cdot i)} H_{j + B \cdot i}, \quad f \in [B]. \label{eq:hashing2}
%    \end{equation} 
%    Since $H$ has support $O(B)$, computing these for a given $Z^r$ requires $O(B)$ samples of $Z^r$, or $O\big(B\log\frac{1}{\delta}\big)$ samples of $X$ by Lemma \ref{lem:downsamp-cost-unit-access}.  Thus, if we replace the loop over $r$ in Algorithm \ref{alg:loc_k} by a single value, along with its corresponding choice $B = C_2 s^r$ in line \ref{line:est_B}, then the sample complexity of the hashing is $O(B |\A| \log_\Lambda m\cdot \log(1/p)\log(1/\delta)=O(B\log m \log(1/p) \log(1/\delta))$. Substituting $B = C_2 s^r$ and summing over $r \in [2k_1]$ gives the desired result.
%	
%	The runtime corresponding to the hashing on line \ref{line:loc_hash} matches the sample complexity, and the subsequent decoding step takes $O(B\log_\Lambda m\cdot |\A|\cdot \Lambda\log(1/p))=O(B\log^2 N \log(1/p))$ time.  Analogously, for a given $r$ value, performing the FFT on line \ref{line:loc_hash} for all $(\alpha,\beta,g,t)$ takes $O\big(B\log B \log n \log\frac{1}{p}\big)$.  Setting $B = O(s^r)$, summing these bounds over all $r$, and upper bounding $\log s^r$ by $\log n$ in the latter, their combined contribution behaves as $O\big(\sum_{r} s^r \log^2 n \log\frac{1}{p}\big)$.
	We first consider the calls to \textsc{HashToBinsReduced}.  This is called for $\log\log m$ values of $(\alpha,\beta)$ and $\log_{\Lambda} N = O\big( \frac{\log m}{\log \log m} \big)$ values of $g$ in each iteration, the product of which is $O(\log m) = O(\log n)$.  Moreover, the number of iterations is $O\big(\log\frac{1}{p}\big)$.  Hence, using Lemma \ref{lem:hash2bins}, we find that the combination of all of these calls costs $O\big(F \sum_{r\in[2k_1]}B^r \log \frac{1}{\delta} \log n\big)$ samples, with a runtime of $O\big((\Bmax F + k_0) k_1 \log^3 n )$, where $\Bmax = O(\max_{r}s^r)$, and $k_0$ is such that $\wh{\chi}$ is $(O(k_0),k_1)$-block sparse.   By the assumption $\max_{r} s^r = O\big( \frac{k_0}{\delta}\big)$, the runtime simplifies to $O\big( \frac{k_0k_1}{\delta} \log^3 n )$

\end{proof}

\section{Pruning the Location List} \label{sec:prune}

\begin{algorithm}
	\caption{Prune a location list via hashing and thresholding techniques.}
	\begin{algorithmic}[1]
    \Procedure{PruneLocation}{$X, \wh{\chi}, L, n, k_0 , k_1 , \delta , p, \theta$}
    
    \State $B \gets 160\frac{k_0 k_1}{\delta} $
    \State $F \gets 10 \log \frac{1}{\delta}$
    \State $G \gets (n, B, F)$-flat filter
    
    \State $T \gets 10 \log \frac{1}{\delta p}$
    
    \For{\texttt{$t \in \{1,\dotsc,T\}$}}
    \State $\Delta \gets$ uniform random sample from $[\frac{n}{k_1}]$
    \State $\sigma \gets$ uniform random sample from odd numbers in  $[\frac{n}{k_1}]$
    \State $\wh{U} \gets \textsc{HashToBins}(X,\wh{\chi},G,n,B,\sigma,\Delta)$
    \State $W^{(t)}_j \gets \sum_{f \in I_j} \big| \wh{G}^{-1}_{o_f(f)} \wh{U}_{h(f)} \omega_n^{-\sigma \Delta f} \big|^2$ for all $j \in L$ \Comment $h(f),o_f(f)$ in Definition \ref{def:hashing} 
    
    \EndFor
    
    \State $W_j \gets \Median_{\, t}(W^{(t)}_j)$ for all $j \in L$
    \State $L' \gets \{j \in L \,:\, W_j \ge \theta\}$
    
    %\For{\texttt{$j \in L$}}
    %\If{$ W_j \ge \theta$} 
    %\State $L' \gets L' \cup \{j\}$
    %\EndIf
    %
    %\EndFor
    
    \State \textbf{return} $L'$
    \EndProcedure
		
	\end{algorithmic}
	\label{alg:prune}
	
\end{algorithm}

The pruning procedure is given in Algorithm \ref{alg:prune}. Its goal is essentially to reduce the size of the list returned by \textsc{MultiBlockLocate} (\emph{cf.}, Algorithm \ref{alg:2}) from $O(k_0 \log(1+k_0))$ to $O(k_0)$.  More formally, the following lemma shows that with high probability, the pruning algorithm retains most of the energy in the head elements, while removing most tail elements.

\medskip
\noindent \textbf{Lemma \ref{lem:prune}}~(\textsc{PruneLocation} guarantees -- re-stated from Section \ref{sec:utilities}) {\em 
Given $(n,k_0,k_1)$, the parameters $\theta > 0$, $\delta \in \big(\frac{1}{n}, \frac{1}{20}\big)$ and $p \in(0,1)$, and the signals $X \in \CC^n$ and $\wh{\chi} \in \CC^n$ with $\|\wh{X} - \wh{\chi}\|_2 \ge \frac{1}{\poly(n)}\|\wh{\chi}\|_2$, the output $L'$ of \textsc{PruneLocation}$(X,\wh{\chi}, L, k_0, k_1, \delta , p, n, \theta)$ has the following properties:

\begin{enumerate}[label={\bf \alph*.}]
	\item Let $\Stail$ denote the tail elements in the signal $\wh{X} - \wh{\chi}$, defined as
	$$\Stail = \Big \{ j \in \Big[\frac{n}{k_1}\Big] \, : \, \|(\wh{X} - \wh{\chi})_{I_j}\|_2 \leq \sqrt{\theta} - \sqrt{\frac{\delta}{k_0}}\|\wh{X} - \wh{\chi}\|_2 \Big \},$$
	where $I_j$ is defined in Definition \ref{def:block}.    Then, we have
	$$\EE \Big[ \big| L' \cap \Stail \big| \Big] \le \delta p \cdot |L \cap \Stail|.$$
	\item Let $\Shead$ denote the head elements in the signal $\wh{X} - \wh{\chi}$, defined as
	$$\Shead = \Big \{ j \in \Big[\frac{n}{k_1}\Big] \, : \, \|(\wh{X} - \wh{\chi})_{I_j}\|_2 \geq \sqrt{\theta} + \sqrt{\frac{\delta}{k_0}} \|\wh{X} - \wh{\chi}\|_2 \Big \}.$$
	Then, we have
	$$\EE \Big[ \sum_{j \in (L \cap \Shead) \backslash L' } \| (\wh{X}-\wh{\chi})_{I_j} \|_2^2 \Big] \leq  \delta p \sum_{j \in L \cap \Shead} \|(\wh{X}-\wh{\chi})_{I_j}\|_2^2.$$
\end{enumerate}
Moreover, provided that $\|\wh{\chi}\|_0 = O(k_0k_1)$, the sample complexity is $O(\frac{k_0 k_1}{\delta} \log \frac{1}{\delta p} \log \frac{1}{\delta})$, and the runtime is $O(\frac{k_0 k_1}{\delta} \log \frac{1}{\delta p} \log \frac{1}{\delta} \log n + k_1 \cdot |L| \log\frac{1}{\delta p})$.
}
\begin{proof}
	We begin by analyzing the properties of the random variables $W_j$ used in the threshold test. We define $X' = X - \chi$, let $\wh{U}$ be the output of \textsc{HashToBins}, and let $\wh{U}^*$ be its exact counterpart as defined in Lemma \ref{lem:hash2bins}.  It follows that we can write the random variable $W_j^{(t)}$ (\emph{cf.}, Algorithm \ref{alg:prune}) as
	\begin{align}
	W_j^{(t)}
	&= \sum_{f \in I_j} \Big| \wh{G}^{-1}_{o_f(f)} \wh{U}_{h(f)} \omega_n^{-\sigma \Delta f} \Big|^2 \nonumber \\
	&= \sum_{f \in I_j} \Big| \wh{G}^{-1}_{o_f(f)} \wh{U}^*_{h(f)} \omega_n^{-\sigma \Delta f} + \wh{G}^{-1}_{o_f(f)} (\wh{U}_{h(f)} - \wh{U}^*_{h(f)}) \omega_n^{-\sigma \Delta f} \Big|^2 \nonumber \\
	&= \sum_{f \in I_j} \Big| \wh{X}'_f + \err^{(t)}_f + \errtilde^{(t)}_f\Big|^2, \label{eq:wj_alt}
	\end{align}
	where (i) $\err^{(t)}_f = \wh{G}^{-1}_{o_f(f)} \sum_{f' \in [n] \backslash \{f\}} \wh{X}'_{f'} \wh{G}_{o_f(f')} \omega_n^{\sigma \Delta(f'-f)}$, with $(\sigma,\Delta)$ implicitly depending on $t$; this follows directly from Lemma \ref{lem:uhat}, along with the definitions $\pi(f) = \sigma f$ and $o_{f}(f') = \pi(f') - \frac{n}{B} h(f)$. (ii) $\errtilde^{(t)}_f = \wh{G}^{-1}_{o_f(f)} (\wh{U}_{h(f)}-\wh{U}^*_{h(f)}) \omega^{-\sigma \Delta f}$, a polynomially small error term (\emph{cf}., Lemma \ref{lem:hash2bins}).
	
	\textbf{Bounding $\err^{(t)}_f$ and $\errtilde^{(t)}_f$:} In Lemma \ref{lem:err_terms} below, we show that
	\begin{gather}
    	\EE_{\Delta,\pi} \Big[ |\err^{(t)}_f|^2 \Big] \le \frac{20}{B}\|\wh{X}'\|_2^2 \label{eq:avg_err_t} \\
    	|\errtilde^{(t)}_f| \le 2n^{-c+c'} \|\wh{X}'\|_2, \label{eq:err'_bound0}
	\end{gather}
	where $c$ is used in \textsc{HashToBins}, and $c'$ is value such that $\|\wh{X} - \wh{\chi}\|_2 \ge \frac{1}{n^{c'}}\|\wh{\chi}\|_2$.  For \eqref{eq:err'_bound0}, we upper bound the $\ell_{2}$ norm by the square root of the vector length times the $\ell_{\infty}$ norm, yielding
		\begin{equation}
	    	\sqrt{\sum_{f \in [n]} |\errtilde^{(t)}_f|^2} \le \sqrt{n}\max_{f \in [n]} |\errtilde^{(t)}_f| \le 2n^{-c+c'+1/2} \|\wh{X}'\|_2. \label{eq:err'_bound}
		\end{equation}
	
	We now calculate the probability of a given block $j$ passing the threshold test, considering two separate cases.
	
	\begin{itemize}
		\item {\bf If $j$ is in the tail:} The probability for $j$ to pass the threshold is closely related to $\PP[ W_j^{(t)} \ge \theta ] = \PP\big[ \sqrt{W_j^{(t)}} \ge \sqrt{\theta} \big]$. From \eqref{eq:wj_alt}, $\sqrt{W_j^{(t)}}$ is the $\ell_2$-norm of a sum of three signals, and hence we can apply the triangle inequality to obtain
		$$ \PP \big[ W_j^{(t)} \ge \theta \big] \le \PP \bigg[ \sum_{f \in I_j} | \err^{(t)}_f |^2 \ge \Big( \sqrt{\theta} - \sqrt{\sum_{f \in I_j} | \wh{X}'_f |^2} - 2n^{-c+c'+1/2} \|\wh{X}'\|_2 \Big)^2 \bigg],$$
		where we have applied \eqref{eq:err'_bound}.
				
		By definition, for any $j \in \Stail$, we have $\sqrt{\theta} - \|\wh{X}'_{I_j}\|_2 \ge \sqrt{\frac{\delta}{k_0}}\|\wh{X}'\|_2$. Hence, and recalling that $\delta \ge \frac{1}{n}$, if $c$ if sufficiently large so that $\sqrt{\frac{\delta}{k_0}} - 2n^{-c+c'+1/2} \ge \sqrt{\frac{0.9\delta}{k_0}}$, then Markov's inequality yields
		\begin{align}
			\PP \big[ W_j^{(t)} \ge \theta \big] 
			&\le \PP \bigg[ \sum_{f \in I_j} | \err^{(t)}_f |^2 \ge \frac{0.9\delta}{k_0}\|\wh{X}'\|_2^2 \bigg] \nn \\
			&\le \frac{\EE_{\Delta,\pi} \Big[ \sum_{f \in I_j} |\err^{(t)}_f|^2 \Big] }{\frac{0.9\delta}{k_0}\|\wh{X}'\|_2^2} \nn \\
			&\le \frac{ \frac{20 k_1}{B}\|\wh{X}'\|_2^2 }{\frac{0.9\delta}{k_0}\|\wh{X}'\|_2^2} \nn \\
			& \le \frac{1}{6} \nn 
		\end{align}
		where the third line follows form \eqref{eq:avg_err_t} and $|I_j| = k_1$, and the final line follows from the choice $B = 160 \frac{k_0k_1}{\delta}$.  Since $W_j$ is the median of $T$ independent such random variables, it can only exceed $\theta$ if there exists a subset of $t$ values of size $\frac{T}{2}$ with $W_j^{(t)} \ge \theta$.  Hence, 
		\begin{equation}
		\PP [W_j \ge \theta] \le {{T}\choose{T/2}} \Big( \frac{1}{6} \Big)^{T/2} \le \Big(\frac{2}{3}\Big)^{T/2} \le \delta p, \nonumber 
		\end{equation}
		where we applied ${{T}\choose{T/2}} \le 2^T$, followed by $T = 10\log\frac{1}{\delta p}$ (\emph{cf.}, Algorithm \ref{alg:prune}).
		
		\item {\bf If j is in the head:} We proceed similarly to the tail case, but instead use the triangle inequality in the form of a lower bound (i.e., $\|a + b\|_2 \ge \|a\|_2 - \|b\|_2$), yielding
		\begin{equation}
		\PP \big[ W_j^{(t)} \le \theta \big] \le \PP \bigg[ \sum_{f \in I_j} | \err^{(t)}_f |^2 \ge \Big( \sqrt{\sum_{f \in I_j} | \wh{X}'_f |^2} - \sqrt{\theta} - 2n^{-c+c'+1/2} \|\wh{X}'\|_2 \Big)^2 \bigg]. \nn
		\end{equation}
		By definition, for any $j \in \Shead$, we have $\| \wh{X}'_{I_j}\|_2 - \sqrt{\theta} \ge \sqrt{\frac{\delta}{k_0}}\|\wh{X}'\|_2 $. Hence, if $c$ if sufficiently large so that $\sqrt{\frac{\delta}{k_0}} - 2n^{-c+c'+1/2} \ge \sqrt{\frac{0.9\delta}{k_0}}$, then analogously to the tail case above, we have
		\begin{align}
			\PP \big[ W_j^{(t)} \le \theta \big] 
			&\le \PP \bigg[ \sum_{f \in I_j} | \err^{(t)}_f |^2 \ge \frac{0.9 \delta}{k_0}\|\wh{X}'\|_2^2 \bigg] \le \frac{1}{6}, \nn 
		\end{align}
		and consequently $\PP[W_j \le \theta] \le \delta p$.
	\end{itemize}
	
	\noindent\textbf{First claim of lemma:} Since $L' \subset L$, we have
	\begin{equation}
	\EE \Big[ \big| L' \cap \Stail \big| \Big] = \sum_{j \in L \cap \Stail} \PP \big[ j \in L' \big] = \sum_{j \in L \cap \Stail} \PP \big[ W_j \ge \theta \big]. \nn
	\end{equation}
	Since we established that $\PP[W_j \ge \theta]$ is at most $\delta p$, we obtain
	\begin{equation}
	\EE \Big[ \big| L' \cap \Stail \big| \Big] \le \sum_{j \in L \cap \Stail} \delta p = \delta p  \cdot |L \cap \Stail|. \nn
	\end{equation}
	
	\noindent\textbf{Second claim of lemma:} In order to upper bound $\sum_{j \in (L \cap \Shead) \backslash L' } \| \wh{X}'_{I_j} \|_2^2$, we first calculate its expected value as follows:
	\begin{equation}
	\begin{split}
	\EE \Big[ \sum_{j \in (L \cap \Shead) \backslash L' } \| \wh{X}'_{I_j} \|_2^2 \Big]
	&= \EE \Big[ \sum_{j \in L \cap \Shead} \| \wh{X}'_{I_j} \|_2^2 \, \Ic \big[ j \notin L' \big] \, \Big]\\
	&= \sum_{j \in L \cap \Shead} \| \wh{X}'_{I_j} \|_2^2 \, \PP \big[ j \notin L' \big]\\
	&= \sum_{j \in L \cap \Shead} \| \wh{X}'_{I_j} \|_2^2 \, \PP \big[ W_j \le \theta \big].
	\end{split} \nn
	\end{equation}
	The probability $\PP \big[ W_j \le \theta\big]$ for $j \in L \cap \Shead$ is at most $\delta p$, and hence
	\begin{equation}
	\EE \Big[ \sum_{j \in (L \cap \Shead) \backslash L' } \| \wh{X}'_{I_j} \|_2^2 \Big] \le \delta p \sum_{j \in L \cap \Shead } \| \wh{X}'_{I_j} \|_2^2. \nn
	\end{equation}
	
	\paragraph{Sample complexity and runtime} For the sample complexity, note that the algorithm only uses samples via its call to \textsc{HashToBins}.  By part (i) of Lemma \ref{lem:hash2bins} and the choices $B= 160 \frac{k_0 k_1}{\delta}$ and $F = 10 \log \frac{1}{\delta}$, the sample complexity is $O(FB) = O(\frac{k_0 k_1 }{\delta} \log \frac{1}{\delta})$ per hashing operation. Since we run the hashing in a loop $10\log \frac{1}{\delta p}$ times, the sample complexity is $O(\frac{k_0 k_1 }{\delta} \log \frac{1}{\delta} \log \frac{1}{\delta p})$.  
	
	The runtime depends on three operations. The first is calling \textsc{HashToBins}, for which an analogous argument as that for the sample complexity holds, with the extra $\log n$ factor arising from Lemma \ref{lem:hash2bins}.  The second operation is the computation of $W_j^{(t)}$, which takes $|I_j| = O(k_1)$ time for each $j \in L$. Hence, the total contribution from the loop is $O( k_1 \cdot |L| \log \frac{1}{\delta p})$. Finally, since the median can be computed in linear time, computing the medians for every $j \in L$ costs $O(|L| \log \frac{1}{\delta p})$ time, which is dominated by the computation of $W_j^{(t)}$.
	
\end{proof}

In the preceding proof, we made use of the following.

\begin{lem} \label{lem:err_terms}
    \emph{(\textsc{EstimateValues} guarantees -- re-stated from Section \ref{sec:utilities})}
    Fix  $(n,k_0,k_1,B)$, the signals $X \in \CC^n$ and $\wh{\chi} \in \CC^n$ with $\|\wh{X} - \wh{\chi}\|_2 \ge \frac{1}{n^{c'}}\|\wh{\chi}\|_2$, and the uniformly random parameters $\sigma,\Delta \in [n]$ with $\sigma$ odd, and let $\wh{U}$ be the output of $\textsc{HashToBins}(X,\wh{\chi},G,n,B,\sigma,\Delta)$ and $\wh{U}^*$ its exact counterpart.  Then defining $\err_f = \wh{G}^{-1}_{o_f(f)} \sum_{f' \in [n] \backslash \{f\}} \wh{X}'_{f'} \wh{G}_{o_f(f')} \omega_n^{\sigma \Delta(f'-f)}$ and $\errtilde_f = \wh{G}^{-1}_{o_f(f)} (\wh{U}_{h(f)}-\wh{U}^*_{h(f)}) \omega^{-\sigma \Delta f}$ (for $h$ and $\sigma_f$ in Definition \ref{def:hashing}), we have
	\begin{gather}
    	\EE_{\Delta,\pi} \Big[ |\err_f|^2 \Big] \le \frac{20}{B}\|\wh{X}'\|_2^2  \\
    	|\errtilde_f| \le 2n^{-c+c'} \|\wh{X}'\|_2
	\end{gather}
	for $c$ used in $\textsc{HashToBins}$.
\end{lem}
\begin{proof}
    We take the expectation of $|\err_f|^2$, first over $\Delta$:
	\begin{equation}
	\EE_{\Delta} \Big[ |\err_f|^2 \Big] = |\wh{G}_{o_f(f)}|^{-2} \sum_{f' \in [n] \backslash \{f\}} |\wh{X}'_{f'}|^2 |\wh{G}_{o_f(f')}|^2 \nn
	\end{equation}
	by Parseval. By Definition \ref{def:filterG} and the definition of $o_f(\cdot)$, we can upper bound $|\wh{G}_{o_f(f)}|^{-2} \le 2$.  Continuing, we take the expectation with respect to the random permutation $\pi$:
	\begin{align}
	\EE_{\Delta,\pi} \Big[ |\err_f|^2 \Big] 
	&\le \EE_{\pi} \Big[ 2 \sum_{f' \in [n] \backslash \{f\}} |\wh{X}'_{f'}|^2 |\wh{G}_{o_f(f')}|^2 \Big] \nonumber \\
	&= 2 \sum_{f' \in [n] \backslash \{f\}} |\wh{X}'_{f'}|^2 \EE_{\pi} \Big[ |\wh{G}_{o_f(f')}|^2 \Big] \le \frac{20}{B}\|\wh{X}'\|_2^2.
	\end{align}
	by Lemma \ref{lem:perm_property}. 
	% Since there are $k_1$ frequencies in $I_j$, we conclude that
	%\begin{equation}
	%\EE_{\Delta,\pi} \Big[ \sum_{f \in I_j} |\err_f|^2 \Big] \le \frac{20k_1}{B}\|\wh{X}'\|_2^2 = \frac{\delta}{8 k_0}\|\wh{X}'\|_2^2 \label{eq:Ij_avg}
	%\end{equation}
	%by the choice $B = 160\frac{k_0 k_1}{\delta}$.
	
	We now turn to $\errtilde_f$.  We know from Lemma \ref{lem:hash2bins} that $|\wh{U}_{h(f)}-\wh{U}^*_{h(f)}| \le \|\wh{U}-\wh{U}^*\|_{\infty} \le n^{-c}\|\wh{\chi}\|_2$.  Hence, and again using $|\wh{G}_{o_f(f)}|^{-2} \le 2$, we find that 
	\begin{equation}
    	|\errtilde_f| \le 2n^{-c} \|\wh{\chi}\|_2 \le 2n^{-c+c'} \|\wh{X}'\|_2, 
	\end{equation}
	where the second inequality follows since $\|\wh{\chi}\|_2 \le n^{c'} \|\wh{X}'\|_2$ for some $c' > 0$ by assumption. 
\end{proof}

\section{Estimating Individual Frequency Values} \label{sec:est_vals}

\begin{algorithm}
	\caption{Energy estimation procedure for individual frequencies} \label{alg:est_values}
	\begin{algorithmic}[1]
		
		\Procedure{EstimateValues}{$X, \wh{\chi}, L, n, k_0, k_1, \delta , p$}
		
		\State $B \gets \frac{1200}{\delta}k_0  k_1$
		\State $F \gets 10 \log \frac{1}{\delta}$
		\State $G \gets (n, B, F)$-flat filter \Comment See Definition \ref{def:filterG}
		\State $\Fc \gets \{ f \in [n] \, : \, \round(\frac{f}{k_1}) \in L \}$
		\State $T \gets 10\log \frac{2}{p}$
		\For{\texttt{$t \in \{1,\dotsc,T\}$}}
		\State $\Delta \gets$ uniform random sample from $[\frac{n}{k_1}]$
		\State $\sigma \gets$ uniform random sample from odd numbers in  $[\frac{n}{k_1}]$
		\State $\wh{U} \gets \textsc{HashToBins}(X,\wh{\chi},G,n,B,\sigma,\Delta)$ \Comment $o_f(f), h(f)$ in Definition \ref{def:hashing}
		\State $W^{(t)}_f \gets \wh{G}^{-1}_{o_f(f)} \wh{U}_{h(f)} \omega^{-\sigma \Delta f} $ for each $f \in \Fc$
		
		\EndFor
		
		\State $W_f \gets \Median_{\, t}(x^{(t)}_f)$ for each $f \in \Fc$ \Comment Separately for the real and imaginary parts
		
		\State \textbf{return} $W$
		\EndProcedure
		
	\end{algorithmic}
	
\end{algorithm}

Once we have located the blocks, we need to estimate the frequency values with them.  The function \textsc{EstimateValues} in Algorithm \ref{alg:est_values} performs this task for us via basic hashing techniques.  The following lemma characterizes the guarantee on the output.

\medskip
\noindent \textbf{Lemma  \ref{lem:estimate}} (Re-stated from Section \ref{sec:utilities}) {\em 
	For any integers $(n,k_0,k_1)$, list of block indices $L$, parameters $\delta \in \big(\frac{1}{n},\frac{1}{20}\big)$ and $p \in (0,1/2)$, and signals $X \in \CC^n$ and $\wh{\chi} \in \CC^n$ with $\| \wh{X} - \wh{\chi} \|_2 \ge \frac{1}{\poly(n)} \| \wh{\chi} \|_2$, the output $W$ of the function \textsc{EstimateValues}$(X,\wh{\chi}, L, n, k_0, k_1, \delta, p)$ has the following property:
		
		\begin{equation}
		\sum_{f \in \bigcup_{j \in L} I_j} |W_f - (\wh{X}-\wh{\chi})_f|^2 \le \delta \frac{|L|}{3k_0} \| \wh{X}-\wh{\chi} \|_2^2 \nn
		\end{equation}
		with probability at least $1-p$, where $I_j$ is the $j$-th block.   Moreover, the sample complexity is $O(\frac{k_0 k_1}{\delta} \log \frac{1}{p} \log\frac{1}{\delta})$, and if $\|\wh{\chi} \|_0=O(k_0 k_1)$, then the runtime is $O(\frac{k_0 k_1}{\delta} \log \frac{1}{p} \log\frac{1}{\delta} \log n + k_1 \cdot |L|\log \frac{1}{p})$.
}
\begin{proof}
	Let $X'=X-\chi$, and let $\wh{U}$ be the output of \textsc{HashToBins} and $\wh{U}^*$ its exact counterpart.
	We start by calculating $W_f^{(t)}$ for an arbitrary $f \in \Fc$:
	\begin{align}
	W_f^{(t)}
	&= \wh{G}^{-1}_{o_f(f)} \wh{U}_{h(f)} \omega^{-\sigma \Delta f} \nonumber \\
	&= \wh{G}^{-1}_{o_f(f)} \wh{U}^*_{h(f)} \omega^{-\sigma \Delta f} + \wh{G}^{-1}_{o_f(f)} (\wh{U}_{h(f)}-\wh{U}^*_{h(f)}) \omega^{-\sigma \Delta f} \nonumber \\
	&= \wh{X}'_f + \err^{(t)}_f + \errtilde^{(t)}_f \text{~~~~(by Lemma \ref{lem:uhat})},
	\label{eq:Wf_init}
	\end{align}
	where $\err^{(t)}_f = \wh{G}^{-1}_{o_f(f)} \sum_{f' \in [n] \backslash \{f\}} \wh{X}'_{f'} \wh{G}_{o_f(f')} \omega^{\sigma \Delta(f'-f)}$, and $\errtilde^{(t)}_f = \wh{G}^{-1}_{o_f(f)} (\wh{U}_{h(f)}-\wh{U}^*_{h(f)}) \omega^{-\sigma \Delta f}$, for $(\sigma,\Delta)$ implicitly depending on $t$.
	
	\textbf{Bounding $\err^{(t)}_f$ and $\errtilde^{(t)}_f$:} Using Lemma \ref{lem:err_terms} in Appendix \ref{sec:prune}, we have
	\begin{gather}
		\EE_{\Delta,\pi} \Big[ |\err^{(t)}_f|^2 \Big] \le \frac{20}{B}\|\wh{X}'\|_2^2 \label{eq:avg_err} \\
		|\errtilde^{(t)}_f| \le 2n^{-c+c'} \|\wh{X}'\|_2, \label{eq:err_tilde}
	\end{gather}
	where $c$ is used in $\textsc{HashToBins}$, and $c'$ is the exponent in the $\poly(n)$ notation of the assumption $\| \wh{X} - \wh{\chi} \|_2 \ge \frac{1}{\poly(n)} \| \wh{\chi} \|_2$.
	
	In order to characterize $|W_f^{(t)} - \wh{X}'_f|^2$, we use the following:
	\begin{equation}
	|\errtilde_f|^2 + 2|\err^{(t)}_f| \cdot |\errtilde^{(t)}_f| \le 4n^{2(-c+c')} \|\wh{X}'\|_2^2 + 4n^{-c+c'} \|\wh{X}'\|_2 \cdot |\err^{(t)}_f|. \label{eq:err_tilde2}
	\end{equation}
	which follows directly from \eqref{eq:err_tilde}. 
	
	\textbf{Characterizing $|W_f^{(t)} - \wh{X}'_f|^2$:}
	We have from \eqref{eq:Wf_init}, \eqref{eq:avg_err}, and \eqref{eq:err_tilde2} that 
	\begin{align}
    	\EE[|W_f^{(t)} - \wh{X}'_f|^2] 
        	&\le \EE\big[ |\err^{(t)}_f|^2 + 2|\err^{(t)}_f| \cdot |\errtilde^{(t)}_f| + |\errtilde^{(t)}_f|^2\big] \nn \\
        	&\le \frac{20}{B}\|\wh{X}'\|_2^2 + 4n^{2(-c+c')} \|\wh{X}'\|_2^2 + 4n^{-c+c'} \|\wh{X}'\|_2  \EE\big[ |\err^{(t)}_f| \big] \nn \\
        	&\le \frac{20}{B}\|\wh{X}'\|_2^2 + 4n^{2(-c+c')} \|\wh{X}'\|_2^2 + 4\sqrt{\frac{20}{B}}n^{-c+c'} \|\wh{X}'\|_2^2, \label{eq:avg_total}
	\end{align}
	where the last line follows by writing $\EE\big[ |\err^{(t)}_f| \big] \le \sqrt{\EE\big[ |\err^{(t)}_f|^2 \big]}$ via Jensen's inequality, and  then applying \eqref{eq:avg_err}.  
	
	Since $B = \frac{1200 k_0 k_1}{\delta}$ and we have assumed $\delta \ge \frac{1}{n}$, we have $B \le 1200n^3$, and hence we have for sufficiently large $c$ that \eqref{eq:avg_total} simplifies to $\EE[|W_f^{(t)} - \wh{X}'_f|^2] \le \frac{25}{B} \|\wh{X}'\|_2^2$.  This means that
	\begin{equation}
	\PP_{\Delta,\pi} \Big[ |W_f^{(t)} - \wh{X}'_f|^2 \ge \frac{160t}{B}\|\wh{X}'\|_2^2 \Big] \le \frac{\EE_{\Delta,\pi} \Big[ |W_f^{(t)} - \wh{X}'_f|^2 \Big] }{\frac{160t}{B}\|\wh{X}'\|_2^2} \le \frac{1}{6t}. \label{eq:single_run}
	\end{equation}
	by Markov's inequality. 
	
	\textbf{Taking the median:} Recall that $W_f$ is the median of $T$ independent random variables, with the median taken separately for the real and imaginary parts.  Since $|W|^2 = |\mathrm{Re}(W)|^2 + |\mathrm{Im}(W)|^2$, we find that \eqref{eq:single_run} holds true when $W_f^{(t)} - \wh{X}'_f$ is replaced by its real or imaginary part.  Hence, with probability at least $1 - {T \choose T/2} \big( \frac{1}{6t} \big)^{T/2}$, we have $|\mathrm{Re}(W_f^{(t)} - \wh{X}'_f)|^2 < \frac{160t}{B}\|\wh{X}'\|_2^2$, and analogously for the imaginary part.  Combining these and applying the union bound, we obtain
	\begin{equation}
	\PP\bigg[ |W_f - \wh{X}'_f|^2 \ge \frac{320t}{B}\|\wh{X}'\|_2^2 \bigg] \le 2{T \choose T/2} \Big( \frac{1}{6t} \Big)^{T/2} \le 2\Big(\frac{2}{3t}\Big)^{T/2} \le \frac{p}{t^{T/2}}, \label{tailbound}
	\end{equation}
	where we first applied ${T \choose T/2} \le 2^T$, and then the choice $T = 10\log\frac{2}{p}$ from Algorithm \ref{alg:est_values} and the choice of $p \le 1/2$.
	
	We now bound the error as follows:
	\begin{equation}
	\begin{split}
	|W_f - \wh{X}'_f|^2 \le \frac{320}{B}\|\wh{X}'\|_2^2 + \Big| |W_f - \wh{X}'_f|^2 - \frac{320}{B}\|\wh{X}'\|_2^2 \Big|_+.
	\end{split} \label{eq:err_bound}
	\end{equation}
	We write the expected value of the second term as
	\begin{equation}
	\begin{split}
	\EE \Big[ \Big| |W_f - \wh{X}'_f|^2 - \frac{320}{B}\|\wh{X}'\|_2^2 \Big|_+ \Big]
	&= \int_{0}^{\infty} \PP \Big[ \Big| |W_f - \wh{X}'_f|^2 - \frac{320}{B}\|\wh{X}'\|_2^2 \Big|_+ \ge u \Big] du\\
	&= \int_{0}^{\infty} \PP \Big[ |W_f - \wh{X}'_f|^2 \ge \frac{320}{B}\|\wh{X}'\|_2^2 + u \Big] du\\
	&= \int_{1}^{\infty} \frac{320}{B}\|\wh{X}'\|_2^2 \PP \Big[ |W_f - \wh{X}'_f|^2 \ge \frac{320v}{B}\|\wh{X}'\|_2^2 \Big] dv
	\end{split} \nn
	\end{equation}
	where we applied the change of variable $v=1 + \frac{u}{\frac{320}{B}\|\wh{X}'\|_2^2} $. By incorporating \eqref{tailbound} into this integral, we obtain
	\begin{align}
	\EE \Big[ \Big| |W_f - \wh{X}'_f|^2 - \frac{320}{B}\|\wh{X}'\|_2^2 \Big|_+ \Big]
	&\le \frac{320}{B}\|\wh{X}'\|_2^2 \int_{1}^{\infty} \frac{p}{v^{T/2}} dv \nonumber \\
	&\le \frac{320}{B}\|\wh{X}'\|_2^2 \cdot \frac{p}{T/2-1} \nonumber \\
	&\le \frac{80}{B}\|\wh{X}'\|_2^2 \cdot p, \label{eq:err_final}
	\end{align}
	where the second line is by explicitly evaluating the integral (with $T > 2$), and the third by $T/2 -1 \ge 4$ (\emph{cf.}, Algorithm \ref{alg:est_values}).
	
	Summing \eqref{eq:err_bound} over $\Fc = \cup_{j \in L} I_j$, we find that the total error is upper bounded as follows:
	\begin{equation}
	\begin{split}
	\sum_{f \in \Fc} |W_f - \wh{X}'_f|^2 
	&\le \frac{320 |\Fc|}{B}\|\wh{X}'\|_2^2 + \sum_{f \in \Fc} \Big| |W_f - \wh{X}'_f|^2 - \frac{320}{B}\|\wh{X}'\|_2^2 \Big|_+ .
	\end{split} \nn
	\end{equation}
	From \eqref{eq:err_final}, the expected value of the second term is at most $p \cdot \frac{80|\Fc|}{B}\|\wh{X}'\|_2^2$, and hence
	\begin{equation}
	\sum_{f \in \Fc} |W_f - \wh{X}'_f|^2 \le \frac{400 |\Fc|}{B}\|\wh{X}'\|_2^2 \nonumber
	\end{equation}
	with probability at least $1-p$, by Markov's inequality.  The lemma now follows by the choice $B = \frac{1200}{\delta}k_0  k_1$ in Algorithm \ref{alg:est_values}.
	
	\paragraph{Sample complexity and runtime:} To calculate the sample complexity, note that the only operation in the algorithm that takes samples is the call to \textsc{HashToBins}. By Lemma \ref{lem:hash2bins}, and the choices $B= 1200 \frac{k_0 k_1}{\delta}$ and $F = 10 \log \frac{1}{\delta}$, the sample complexity is $O(\frac{k_0 k_1}{\delta} \log \frac{1}{\delta})$ per hashing performed. Since we run the hashing in a loop $10\log \frac{2}{ p}$ times, this amounts to a total of $O(\frac{k_0 k_1 }{\delta} \log \frac{1}{\delta} \log \frac{1}{ p})$.
	
	The runtime depends on two operations. The first one is calling \textsc{HashToBins}, whose analysis follows similarly to the aforementioned sample complexity analysis using the assumption $\|\wh{\chi}\|_0 = O(k_0k_1)$, but with an extra $\log n$ factor compared to the sample complexity, as per Lemma \ref{lem:hash2bins}.  
	
	The other operation is computation of $W_f^{(t)}$, which takes unit time for each $f \in \Fc$. Since the size of $|\Fc|=k_1 \cdot |L|$, running it in a loop costs $O( k_1 \cdot |L| \log \frac{1}{p})$. Computing the median is done in linear time which consequently results in $|\Fc| T=O( k_1 \cdot |L| \log \frac{1}{p})$.
\end{proof}

\bibliographystyle{IEEEtranSA}
\bibliography{JS_References,addon.bib,paper}

\end{document}